\documentclass[format=acmsmall, review=false, authorversion=true, screen=true]{acmart}
\renewcommand\footnotetextcopyrightpermission[1]{} %
\usepackage[utf8]{inputenc}
\usepackage{hyphenat}
\usepackage{xspace}
\usepackage{amsmath}
\usepackage{mathtools,amssymb,mathrsfs}
\usepackage{mathrsfs}
\usepackage{stmaryrd}

\usepackage{paralist} 
\usepackage[boxed,ruled,vlined]{algorithm2e} 
\usepackage{enumerate}
\usepackage{enumitem} 
\usepackage{mdframed}
\usepackage{booktabs} %
\usepackage{url}
\usepackage{subcaption}
\usepackage{tabularx}

\usepackage{multirow}

\usepackage{tikz} 
\usepackage{subcaption}

\usetikzlibrary{positioning} 
\usetikzlibrary{shapes,arrows,calc,matrix}

\newcommand{\fail}{{\mbox{\it fail}}}
\newcommand{\ok}{{\mbox{\it ok}}}
\newcommand{\IF}{{\mbox{\it IF}}}
\newcommand{\levels}{\mbox{\it levels}}

\newcommand{\treecomp}{\mathit{treecomp}}

\newcommand{\calF}{{\mathcal F}}
\newcommand{\calH}{{\mathcal H}}
\newcommand{\calO}{{\mathcal O}}

\newcommand{\calT}{{\mathcal T}}

\newcommand{\NP}{\text{\rm NP}}

\newcommand{\ra}{\ensuremath{\rightarrow}}

\newcommand{\la}{\ensuremath{\leftarrow}}

\newcommand{\ffo}[1]{\ensuremath{\mathit{#1}}}

\newcommand{\icBMIP}{$ic$-BMIP}
\newcommand{\dBDP}{$d$-BDP}

\newcommand{\vc}{\mbox{\rm vc}}
\newcommand{\cigap}[1]{\mbox{\it cigap}(#1)}
\newcommand{\tigap}[1]{\mbox{\it tigap}(#1)}

\newcommand{\poly}{\ffo{poly}}

\newcommand{\fhw}{\ffo{fhw}}
\newcommand{\ghw}{\ffo{ghw}}
\newcommand{\hw}{\ffo{hw}}
\newcommand{\vertices}{\ffo{V}}
\newcommand{\V}{\vertices}
\newcommand{\weight}{\ffo{weight}}
\newcommand{\width}{\ffo{width}}
\newcommand{\nodes}{\ffo{nodes}}

\newcommand{\edges}{\ffo{edges}}

\newcommand{\component}{\ffo{comp}}
\newcommand{\parent}{\ffo{parent}}
\newcommand{\tc}{\ffo{treecomp}}
\newcommand{\cov}{\ffo{supp}}
\newcommand{\Int}{\ffo{Int}}
\newcommand{\treeroot}{\ffo{root}}

\newcommand{\rep}{\ffo{rep}}
\newcommand{\first}{\ffo{first}}

\newcommand{\critp}{\ffo{critp}}

\newcommand{\supp}{\ffo{supp}}
\newcommand{\vsupp}{\ffo{vsupp}}
\newcommand{\type}{\ffo{type}}
\newcommand{\iflevel}{\ffo{iflevel}}

\newcommand{\INSET}{\ensuremath{\mathbb{B}}}
\newcommand{\TYPES}{\ensuremath{\mathbb{T}}}
\newcommand{\CLASSES}{\ensuremath{\mathbb{C}}}
\newcommand{\FRINGE}{\ensuremath{\mathbb{F}}}

\newcommand{\class}{\ffo{class}}
\newcommand{\sett}{\ffo{set}}
\newcommand{\markk}{\ffo{mark}}

\newcommand{\iwidth}[1]{\mbox{\it iwidth}(#1)}
\newcommand{\cmiwidth}[2]{\mbox{\it $#1$-miwidth(#2)}}
\newcommand{\ddegree}[1]{\mbox{\it degree}(#1)}

\newcommand{\VTu}{V(T_u)}
\newcommand{\VTs}{V(T_s)}

\newcommand{\rec}[1]{\textsc{Check}(#1)}
\newcommand{\boundedopt}{$K$-\textsc{Bounded-FHW-Optimization}}

\newcommand{\mcF}{\ensuremath{\mathcal{F}}}
\newcommand{\mcH}{\ensuremath{\mathcal{H}}}
\newcommand{\mcC}{\ensuremath{\mathcal{C}}}
\newcommand{\classC}{\ensuremath{\mathscr{C}}}

\newcommand{\defFHD}{\ensuremath{\left< T, (B_u)_{u\in T}, (\gamma_u)_{u\in T} 
\right>}}
\newcommand{\defGHD}{\ensuremath{\left< T, (B_u)_{u\in T}, (\lambda_u)_{u\in T} 
\right>}}
\newcommand{\defHD}{\defGHD}

\newcommand{\defFHDohne}{\ensuremath{\left< T, (B_u), (\gamma_u) \right>}}

\newcommand{\comp}[1]{[#1]-component}
\renewcommand{\path}[1]{[#1]-path}
\newcommand{\adjacent}[1]{[#1]-adjacent}
\newcommand{\connected}[1]{[#1]-connected}

\newcommand{\ppi}{\pi(\hat{u}_1,\hat{u}_N)}

\newcommand{\decomp}{\texttt{$(k,\epsilon,c)$-frac-decomp}\xspace}
\newcommand{\fdecomp}{\texttt{f-decomp}\xspace}
\newcommand{\findFHD}{\texttt{find-fhd}\xspace}

\newcommand{\iforest}{\texttt{Intersection-Forest}\xspace}
\newcommand{\uitree}{\texttt{Union-of-Intersections-Tree}\xspace}
\newcommand{\cupcap}{\bigcup\bigcap}

\newcommand{\lab}{\mathit{label}}
\newcommand{\inter}{\mathit{int}}

\newcommand{\HH}{\ensuremath{H}}

\newcommand{\overbar}[1]{\mkern 1.5mu\overline{\mkern-1.5mu#1\mkern-1.5mu}\mkern 
1.5mu}

\newcommand{\ptime}{{\sc Ptime}\xspace}
\newcommand{\alogspace}{{\sc ALogSpace}\xspace}

\newcommand{\np}{{\sc NP}\xspace}

\newif\ifrdfshort
   \rdfshortfalse

\newcommand{\mcG}{\ensuremath{\mathcal{G}}}

\ifrdfshort

\else

\fi

\newcommand{\nop}[1]{}

\newlength{\knotenbreite}
\newlength{\knotenhoehe}
\newlength{\ueberhoehe}
\newlength{\knotenradius}
\newcolumntype{a}{>{\columncolor{gray!25}}c}

\newcommand{\htnode}[2]{
   $B_u$ \& #1 \\ \hline 
   $\lambda_u$ \&  #2 \\
}

\makeatletter
\providecommand*{\cupdot}{%
  \mathbin{%
    \mathpalette\@cupdot{}%
  }%
}
\newcommand*{\@cupdot}[2]{%
  \ooalign{%
    $\m@th#1\cup$\cr
    \sbox0{$#1\cup$}%
    \dimen@=\ht0 %
    \sbox0{$\m@th#1\cdot$}%
    \advance\dimen@ by -\ht0 %
    \dimen@=.5\dimen@
    \hidewidth\raise\dimen@\box0\hidewidth
  }%
}

\providecommand*{\bigcupdot}{%
  \mathop{%
    \vphantom{\bigcup}%
    \mathpalette\@bigcupdot{}%
  }%
}
\newcommand*{\@bigcupdot}[2]{%
  \ooalign{%
    $\m@th#1\bigcup$\cr
    \sbox0{$#1\bigcup$}%
    \dimen@=\ht0 %
    \advance\dimen@ by -\dp0 %
    \sbox0{\scalebox{2}{$\m@th#1\cdot$}}%
    \advance\dimen@ by -\ht0 %
    \dimen@=.5\dimen@
    \hidewidth\raise\dimen@\box0\hidewidth
  }%
}
\makeatother

\makeatletter
\providecommand*{\capdot}{%
  \mathbin{%
    \mathpalette\@capdot{}%
  }%
}
\newcommand*{\@capdot}[2]{%
  \ooalign{%
    $\m@th#1\cap$\cr
    \hidewidth$\m@th#1\cdot$\hidewidth
  }%
}

\providecommand*{\bigcapdot}{%
  \mathop{%
    \vphantom{\bigcap}%
    \mathpalette\@bigcapdot{}%
  }%
}
\newcommand*{\@bigcapdot}[2]{%
  \ooalign{%
    $\m@th#1\bigcap$\cr
    \sbox0{$#1\bigcap$}%
    \dimen@=\ht0 %
    \advance\dimen@ by -\dp0 %
    \sbox0{\scalebox{2}{$\m@th#1\cdot$}}%
    \advance\dimen@ by -\ht0 %
    \dimen@=.5\dimen@
    \hidewidth\raise\dimen@\box0\hidewidth
  }%
}
\makeatother

\newenvironment{myintro}%
  {\list{}{\leftmargin=0.1in\rightmargin=0.1in}\item[]}%
  {\endlist}

\setcopyright{none}
\title[General and Fractional Hypertree Decompositions: 
Hard and Easy Cases]{General and Fractional Hypertree Decompositions: \\
Hard and Easy Cases}

\author{Wolfgang Fischl}
\affiliation{%
  \institution{TU Wien}
}
\email{wolfgang.fischl@tuwien.ac.at}

\author{Georg Gottlob}
\affiliation{%
  \institution{University of Oxford  \&  TU Wien}
}
\email{georg.gottlob@cs.ox.ac.uk}

\author{Reinhard Pichler}
\affiliation{%
  \institution{TU Wien}
}
\email{reinhard.pichler@tuwien.ac.at}

\ccsdesc[500]{Information systems~Relational database query languages}
\ccsdesc[500]{Theory of computation~Problems, reductions and completeness}

\begin{document}

\begin{abstract}
\normalsize 
Hypertree decompositions, as well as the more powerful generalized hypertree 
decompositions (GHDs), and the yet more general fractional hypertree 
decompositions (FHD) are hypergraph decomposition methods successfully used for 
answering conjunctive queries
and
for solving %
constraint satisfaction problems.
Every hypergraph $H$ has a width relative to each of these %
methods:   
its hypertree width $\hw(H)$, its generalized hypertree width $\ghw(H)$, and 
its  
fractional hypertree width $\fhw(H)$, respectively.
It is known that  $\hw(H)\leq k$ can be checked in polynomial time for fixed $k$, while 
checking $\ghw(H)\leq k$ is NP-complete for $k \geq 3$. 
The complexity of checking  $\fhw(H)\leq k$  for a fixed $k$ has been open for over a decade.

We settle this open problem by showing that  
checking  $\fhw(H)\leq k$ is NP-complete, even for $k=2$.
The same construction allows us to prove also the NP-completeness of checking 
 $\ghw(H)\leq k$ for $k=2$. 
After that, we identify meaningful restrictions
which make checking for bounded $\ghw$ or $\fhw$ tractable or
allow for an efficient approximation of the $\fhw$.
\end{abstract}

\maketitle

\section{Introduction and Background}
\label{sect:introduction}

\noindent{\bf Research Challenges Tackled.}\  
In this work we tackle computational problems on hypergraph decompositions, 
which
play a prominent role for %
answering Conjunctive Queries (CQs) and 
solving Constraint Satisfaction Problems (CSPs), which we discuss %
below.

\nop{The treewidth~\cite{} $tw(G)$ of a graph $G$ is a measure of  its degree 
of cyclicity. The treewidth $tw(G)$ is defined as the smallest width over 
all tree decompositions of a graph (definitions will be given below).} 
Many \np-hard graph-based problems become tractable for instances whose 
corresponding graphs have bounded treewidth.
There are, however, many problems for which the structure of an instance is 
better described by a hypergraph than by a graph, for example,  the 
above mentioned CQs and CSPs. 
Given that treewidth does not generalize hypergraph acyclicity\footnote{We here 
refer to the standard notion of hypergraph acyclicity, as used 
in~\cite{DBLP:conf/vldb/Yannakakis81} and \cite{fagin1983degrees}, where it is 
called $\alpha$-acyclicity. This notion is more general than other types of 
acyclicity that have been introduced in the literature.}, proper  hypergraph 
decomposition methods have been developed, in particular, {\it hypertree 
decompositions (HDs)\/}  
\cite{2002gottlob}, the more general  {\it generalized 
hypertree decompositions (GHDs)}~\cite{2002gottlob}, and the yet 
more general   
{\it fractional hypertree decompositions (FHDs)}~\cite{2014grohemarx}, and 
corresponding notions of width of a hypergraph $H$ have been defined: the  {\it 
hypertree width} $\hw(H)$,  the {\em generalized hypertree width} 
$\ghw(H)$, and the {\em fractional hypertree width} $\fhw(H)$, where for every 
hypergraph $H$, $\fhw(H)\leq \ghw(H)\leq \hw(H)$ holds. 
Definitions are given in 
Section~\ref{sect:prelim}. A number of highly relevant hypergraph-based 
problems such as  
CQ\hyp{}evaluation and CSP\hyp{}solving become tractable for classes of instances of 
bounded 
$\hw$, $\ghw$, or, $\fhw$. For each of the mentioned types of 
decompositions it would thus be %
useful to be able to recognize for each 
constant $k$ whether a given hypergraph $H$ has corresponding  
width at most  
$k$, and if so, to compute such a %
decomposition. More formally, for 
{\it decomposition\/} $\in \{$HD, GHD, FHD$\}$ and $k > 0$, we consider the 
following family of problems:

\medskip
\noindent
\rec{{\it decomposition\/},\,$k$}\\
\begin{tabular}{ll}
 \bf input & hypergraph $H = (V,E)$;\\
 \bf output & {\it decomposition\/}  of $H$ of width $\leq k$ if it 
exists and  \\ & answer `no' otherwise.
\end{tabular}

\smallskip

As shown in~\cite{2002gottlob}, \rec{HD,\,$k$}\  is in \ptime. However, 
little %
is known
about  \rec{FHD,\,$k$}. 
In fact, this has been %
an open problem since the 2006 paper \cite{DBLP:conf/soda/GroheM06},
where Grohe and Marx state: \lq\lq{}It remains an important open question whether 
there
is a polynomial-time algorithm that determines (or approximates)
the fractional hypertree width and constructs a corresponding
decomposition.\rq\rq{}
The 2014 journal version
still 
mentions this as %
open %
and it is conjectured that the problem might be \np-hard. 
The open problem is restated in \cite{bevern2015}, where further evidence for 
the hardness of the problem is 
given by  showing that ``it is not expressible in monadic second-order logic 
whether a 
hypergraph has bounded (fractional, generalized) hypertree width''.
We will tackle this open problem here:

\begin{myintro}
\noindent{\bf Research Challenge 1:} \ Is  \rec{FHD,\,$k$} tractable?
\end{myintro}

\nop{**************
In order to analyze the above problem, we will study the complexity of deciding 
whether, for constant $k$, deciding $\fhw(H)\leq k$ is \np-complete. If so, 
then 
certainly also 
\rec{FHD,\,$k$}  is \np-hard.
**************}

Let us now  turn to generalized hypertree decompositions. In~\cite{2002gottlob} 
the complexity of  \rec{GHD,\,$k$} was stated as an open problem. 
In \cite{2009gottlob}, it was shown that 
\rec{GHD,\,$k$} is \np-complete for $k\geq 3$. 
For 
$k=1$ the problem is trivially tractable because
$\ghw(H)=1$ just means 
$H$ is acyclic. However the case $k=2$ has been left open.  This case is quite 
interesting, because it was observed that the majority of practical queries 
from 
various  benchmarks that are not acyclic have 
$\ghw = 2$ \cite{DBLP:journals/pvldb/BonifatiMT17,pods/FischlGLP19}, and that a decomposition in such cases can be very 
helpful.
Our second 
research goal is to finally  settle the complexity of \rec{GHD,\,$k$} 
completely.

\begin{myintro}
\noindent{\bf Research Challenge 2:} \ Is \rec{GHD,\,$2$} tractable?
\end{myintro}

For those problems which are  known to be intractable, for example, 
\rec{GHD,\,$k$} for $k\geq 3$, and for those others that will turn out to be 
intractable, we would like to find large islands of tractability 
that correspond to meaningful restrictions of the input 
hypergraph instances. Ideally, such restrictions should  fulfil two main 
criteria: (i) they need to be {\it realistic} in the sense that they apply to a 
large number of CQs and/or CSPs in real-life applications, and 
(ii) 
they need to be {\it non-trivial}  in the sense that the restriction itself 
does 
not already imply bounded $\hw$, $\ghw$, or $\fhw$. 
Trivial restrictions would be, for example, 
acyclicity or bounded 
treewidth.
Hence, our third research problem is as follows:

\begin{myintro}
\noindent{\bf Research Challenge 3:} \ Find realistic, non-trivial restrictions 
on hypergraphs which entail the tractability of  the 
\rec{{\it decomp\/},\,$k$} problem for {\it decomp\/} $\in 
\{$GHD, 
FHD$\}$.
\end{myintro}

Where we do not achieve {\sc Ptime} algorithms 
for the precise computation of a decomposition of optimal width, we would like 
to find tractable methods for achieving good approximations. 
Note that 
for GHDs, the problem of approximations is solved, since $\ghw(H) \leq 3 \cdot \hw(H) +1$ holds for every 
hypergraph $H$ \cite{DBLP:journals/ejc/AdlerGG07}.
In contrast, for FHDs, the best known polynomial-time approximation is cubic. More precisely, 
in \cite{DBLP:journals/talg/Marx10}, a polynomial-time algorithm is presented which, 
given a hypergraph $H$ with $\fhw(H) = k$,
computes an FHD of width $\calO(k^3)$. 
We would like to find 
meaningful restrictions that guarantee significantly tighter approximations
in polynomial time. This leads to the fourth research problem:

\begin{myintro}
\noindent{\bf Research Challenge 4:} \ Find realistic, non-trivial restrictions 
on hypergraphs which allow us to compute in {\sc Ptime} good 
approximations of $\fhw(k)$.
\end{myintro}

\smallskip

\noindent{\bf Background and Applications.}\ Hypergraph decompositions have 
meanwhile found their way into commercial database systems such as LogicBlox 
\cite{DBLP:conf/sigmod/ArefCGKOPVW15,
olteanu2015size,BKOZ13,KhamisNRR15,KhamisNR16} and advanced research prototypes 
such as 
EmptyHeaded~\cite{DBLP:conf/sigmod/AbergerTOR16,tu2015duncecap,aberger2016old}. 
Moreover, 
since CQs and CSPs of bounded hypertree width fall into the highly 
parallelizable complexity class LogCFL, hypergraph decompositions have also 
been 
discovered as a useful tool for parallel query processing with MapReduce 
\cite{DBLP:journals/corr/AfratiJRSU14}. Hypergraph decompositions, in 
particular, HDs and GHDs have been used in many other contexts, e.g., 
in combinatorial auctions~\cite{gottlob2013decomposing} and automated selection 
of Web services based on recommendations from social 
networks~\cite{hashmi2016snrneg}.
There exist exact algorithms for computing the generalized or fractional 
hypertree width 
\cite{moll2012}; clearly, 
they require exponential time even if the optimal width
is bounded by some fixed $k$.

CQs are the most basic and arguably the most important class
of queries
in the database world.
Likewise, CSPs constitute one of the most fundamental classes of problems
in Artificial Intelligence. 
Formally, CQs and CSPs are the same problem and correspond to first-order 
formulae using $\{\exists,\wedge\}$ but disallowing $\{\forall, \vee, \neg\}$ 
as 
connectives, that need to be evaluated over a set of finite relations:  the 
{\em 
database relations} for CQs, and the {\em constraint relations} for CSPs.
In practice, CQs have often fewer conjuncts (query atoms) and larger relations, 
while CSPs have more conjuncts but smaller relations.
These problems are well-known to be \np-complete
\cite{DBLP:conf/stoc/ChandraM77}. 
Consequently, there has been an intensive search for tractable fragments of CQs 
and/or CSPs
over the past decades. 
For our work, the approaches based on decomposing the 
structure of 
a given CQ or CSP are most relevant,~see~e.g.\
\cite{DBLP:conf/adbt/GyssensP82,%
DBLP:journals/ai/DechterP89,%
DBLP:conf/aaai/Freuder90,%
DBLP:journals/ai/GyssensJC94,%
DBLP:journals/jcss/KolaitisV00,%
DBLP:conf/stoc/GroheSS01,%
DBLP:conf/cp/DalmauKV02,%
DBLP:journals/tcs/ChekuriR00,%
2002gottlob,%
DBLP:conf/cp/ChenD05,%
DBLP:journals/jacm/Grohe07,%
DBLP:journals/jcss/CohenJG08,%
DBLP:journals/mst/Marx11,%
DBLP:journals/jacm/Marx13,%
DBLP:journals/siamcomp/AtseriasGM13,%
2014grohemarx}.
The underlying structure of both %
is nicely captured by 
hypergraphs. 
The hypergraph $H = (V(H), E(H))$ underlying a CQ (or a CSP)  $Q$ has as vertex 
set $V(H)$ the set of variables occurring in $Q$; moreover, for every atom in 
$Q$, $E(H)$ contains a hyper\-edge consisting 
of all variables occurring in this atom.  
From now on, we shall mainly talk about hypergraphs with the understanding that 
all our results are equally applicable to CQs and~CSPs.

\medskip

\noindent
{\bf Main Results.}    
First of all,  we have investigated  the above mentioned open problem 
concerning 
the recognizability of $\fhw \leq k$ for fixed $k$. Our initial hope was to 
find 
a simple adaptation of the \np-hardness proof in \cite{2009gottlob}
for recognizing $\ghw(H) \leq k$, for $k\geq 3$. Unfortunately, this proof 
dramatically fails for the fractional case. In fact,  the hypergraph-gadgets in 
that proof are such that both  \lq\lq{}yes\rq\rq{} and \lq\lq{}no\rq\rq{} 
instances may yield the same $\fhw$. However, via crucial modifications, 
including 
the introduction of novel gadgets, we succeed to construct a reduction from 
3SAT that allows us to control the $\fhw$ of the resulting  hypergraphs
such that those hypergraphs arising from  \lq\lq{}yes\rq\rq{} 3SAT instances 
have  $\fhw(H)=2$ and those 
arising from  \lq\lq{}no\rq\rq{} instances have $\fhw(H)>2$. Surprisingly, 
thanks to our new gadgets, the resulting proof is 
actually significantly simpler than the \np-hardness proof for recognizing 
$\ghw(H) \leq k$ in \cite{2009gottlob}. We thus obtain the following result: 

\begin{myintro}
\noindent{\bf Main Result 1:} Deciding $\fhw(H) \leq 2$ for hypergraphs $H$ is 
\np-complete and, therefore, 
\rec{FHD,\,$k$} is  intractable
even for $k = 2$.
\end{myintro} 

\noindent
This  result can be extended to the \np-hardness of recognizing
$\fhw(H) \leq k$ for arbitrarily large $k$. Moreover, the same 
construction can be used to prove that recognizing ghw $\leq 2$ is 
also \np-hard, thus killing two birds with one stone. 
\nop{**************************
We thus  finally close 
the gap (for $k=2$)  left open regarding the complexity of deciding 
$\ghw(H)\leq 
k$, which has been
mentioned as an open problem since 2001 (PODS version 
of~\cite{gottlob2003robbers}), 
and which was solved for $k\geq 3$ in~\cite{2009gottlob} :
**************************} %

\begin{myintro}
\noindent{\bf Main Result 2:} Deciding $\ghw(H) \leq 2$ for hypergraphs $H$ is 
\np-complete and, therefore, 
\rec{GHD,\,$2$} is  intractable
even for $k = 2$.
\end{myintro}

The Main Results 1 and 2 are presented in Section \ref{sect:hardness}. 
These results close some smouldering open problems with bad news. We thus 
further 
concentrate on Research Challenges 3 and 4 in order to obtain some positive results 
for 
restricted hypergraph classes. 

We first study GHDs, where we succeed to identify very 
general, realistic, 
and 
non-trivial restrictions that make the \rec{GHD,\,$k$} problem tractable. 
These results are based on new insights about
the differences of GHDs and HDs and the introduction of a novel 
technique for expanding a hypergraph $H$ to an edge-augmented 
hypergraph $H'$ s.t.\ the width $k$ GHDs of $H$ correspond %
to the width 
$k$ 
HDs of $H'$. The crux here is to find restrictions under which only a 
polynomial number of edges needs to be added to $H$ to obtain $H'$. The HDs of 
$H'$ can then be computed in polynomial time. 

In particular, we concentrate on the {\em bounded intersection property 
(BIP)}, which, for a class $\classC$ of hypergraphs  requires that for some 
constant $i$, for each pair of distinct edges $e_1$ and $e_2$ of each 
hypergraph 
$H\in {\classC}$,
$|e_1\cap e_2|\leq i$, and its generalization, the {\em bounded 
multi-intersection property (BMIP)}, which, informally, requires that for some 
constant $c$ any intersection of $c$ distinct hyperedges of $H$ has at most $i$ 
elements for some constant $i$. In 
\cite{pods/FischlGLP19}
we report on tests with a large 
number of known CQ and CSP benchmarks and it turns out that a very large number 
of instances coming from real-life applications enjoy the BIP and a yet more 
overwhelming number enjoys the BMIP for very low constants $c$ and $i$. We 
obtain the following good news, which are presented in Section \ref{sect:ghd}. 

\begin{myintro}
\noindent{\bf Main Result 3:} For classes of hypergraphs fulfilling the BIP or 
BMIP, for every constant $k$,  the problem \rec{GHD,\,$k$} is tractable. 
Tractability holds even 
for classes~$\classC$ 
of 
hypergraphs where for some constant $c$ all intersections of $c$ distinct edges 
of every $H\in{\classC}$ of size $n$ have $\calO(\log n)$ elements.
Our complexity analysis reveals that 
the problem 
 \rec{GHD,\,$k$} 
is %
fixed-parameter tractable
w.r.t.\ 
parameter
$i$ 
of the BIP.
\end{myintro}

The tractability proofs for BIP and BMIP do not directly carry over to  FHDs. 
We thus consider the degree $d$ of  a hypergraph $H = (V(H), E(H))$, 
which is defined as the 
maximum number of hyperedges in which a vertex occurs, i.e., 
$d = \max_{v \in V(H)} |\{ e \in E(H) \mid v \in E(H)\}|$.
We say that a class $\classC$ of hypergraphs  has the {\em bounded degree property (BDP)\/}, 
if there exists $d \geq 1$, 
such that every hypergraph $H\in {\classC}$ has degree $\leq d$.
We 
obtain the following result, which is  presented in Section~\ref{sect:fhd-exact}. 

\begin{myintro}
\noindent{\bf Main Result 4:} 
For classes of hypergraphs fulfilling the BDP and
for every constant $k$,  the problem \rec{FHD,\,$k$} is tractable. 
\end{myintro} 

To get yet bigger tractable classes, we also consider approximations of an optimal FHD. 
Towards this goal, we study the $\fhw$ in case of the BIP and we establish an 
interesting connection between the BMIP and the 
Vapnik--Chervonenkis dimension (VC-dimension) of hypergraphs. 
Our research, presented in Section~\ref{sect:fhd} is summarized as 
follows.

\begin{myintro}
\noindent{\bf Main Result 5:} For rather general, realistic, and non-trivial 
hypergraph restrictions, there exist {\sc Ptime} algorithms that, for 
hypergraphs $H$ with $\fhw(H)=k$, where $k$ is a constant,  produce FHDs whose 
widths are significantly smaller than the best previously known 
approximation. 
In particular, the BIP allows us to compute in polynomial time 
an FHD whose width  is $\leq k + \epsilon$
for arbitrarily chosen constant $\epsilon > 0$. 
The  BMIP  or bounded VC-dimension 
allow us to compute in polynomial time an FHD whose width  is $\calO(k \log k)$. 
\end{myintro} 

We finally turn our attention also to the optimization problem of fractional hypertree width, i.e., given a hypergraph $H$, 
determine $\fhw(H)$ and find an FHD of width $\fhw(H)$. All our algorithms for the \rec{FHD,\,$k$} problem 
have a runtime exponential in the desired width $k$. Hence, even with the restrictions 
to the BIP or BMIP we cannot expect 
an efficient approximation of $\fhw$ if $\fhw$ can become arbitrarily large. We will therefore study the following 
\boundedopt\ problem
for constant $K \geq 1$: 

\medskip
\noindent
\boundedopt\\
\begin{tabular}{lrl}
 \bf input &  \multicolumn{2}{l}{hypergraph $H = (V,E)$;} \\
 \bf output & if $\fhw(H) \leq K$: & find an FHD $\mcF$ of $H$ with minimum width; \\
 & otherwise: & answer ``$\fhw(H) > K$''.
\end{tabular}

\medskip
\noindent
For this bounded version of the optimization problem, we  will prove the following result: 

\begin{myintro}
\noindent{\bf Main Result 6:} There exists a polynomial time approximation scheme (PTAS; for details see Section \ref{sect:fhd}) 
for the \boundedopt\ problem in case of the BIP for any fixed $K\geq 1$.  
\end{myintro}

\section{Preliminaries}
\label{sect:prelim}

\subsection{Hypergraphs}

A {\em hypergraph} is a pair $H = (V(H), E(H))$, consisting of a set $V(H)$ of 
{\em vertices} and a set $E(H)$ of {\em hyperedges} (or, simply {\em edges\/}), 
which are non-empty 
subsets of $V(H)$. We assume that hypergraphs do not have isolated 
vertices, i.e.\ for each $v \in V(H)$, there is at least one edge $e \in E(H)$, 
s.t.\ $v \in e$. For a set $C \subseteq V(H)$, we define $\edges(C) = \{ e \in 
E(H) \mid e \cap C \neq \emptyset \}$ and for a set $S \subseteq E(H)$, we 
define $\V(S) = \{ v \in e \mid e \in S \}$. 
Actually, 
for a set $S$ of edges, it is convenient to write $\bigcup S$ (and $\bigcap S$, respectively)  to denote the set of vertices obtained by taking the union
(or the intersection, respectively) of the edges in $E$. Hence, we can write $\V(S)$ simply as $\bigcup S$.

For a hypergraph $H$ and a set $C \subseteq V(H)$, we say that a pair of 
vertices
$v_1,v_2 \in V(H)$ is \adjacent{$C$} if there exists an edge $e \in E(H)$ such
that $\{ v_1,v_2 \} \subseteq (e \setminus C)$. A \path{$C$} $\pi$ from $v$ 
to $v'$ consists of a sequence $v = v_0,\dots,v_h = v'$ of vertices and a 
sequence of edges $e_0, \dots, e_{h-1}$ ($h \geq 0$) such that $\{ v_i, v_{i+1} 
\} \subseteq ( e_i \setminus C)$, for each $i \in [0\ldots h-1]$. 
We denote by
$\vertices(\pi)$ the set of vertices  occurring 
in the sequence $v_0,\ldots, v_h$. 
Likewise, we denote by
$\edges(\pi)$ the set of edges occurring 
in the sequence $e_0,\ldots,e_{h-1}$. 
A set $W \subseteq
V(H)$ of vertices is \connected{$C$} if $\forall v,v' \in W$ there is a 
\path{$C$} from $v$ to $v'$. A \comp{$C$} is a maximal \connected{$C$}, 
non-empty
set of vertices $W \subseteq V(H)\setminus C$. 

\subsection{(Fractional) Edge Covers}
Let $H = (V(H),E(H))$ be a hypergraph 
and consider (edge-weight) functions $\lambda \colon E(H) \ra \{0,1\}$ and
$\gamma \colon E(H) \ra [0,1]$. 
For $\theta \in \{\lambda, \gamma\}$, 
we denote by $B(\theta)$ the set of all 
vertices {\em covered\/} by $\theta$:
\[ B(\theta) = \left\{ v\in V(H) \mid \sum_{e\in E(H), v\in e} \theta(e) \geq 1 
\right\}\]
The weight of 
such a 
function 
$\theta$ is defined as
\[ \weight(\theta) = \sum_{e \in E(H)} \theta(e). \]
Following \cite{2002gottlob}, we will sometimes consider
$\lambda$ as a set with $\lambda \subseteq E(H)$ 
(i.e., the set of edges $e$ with $\lambda(e) = 1$)
and the weight of $\lambda$ as the cardinality of this set. %
However, for the  sake of a 
uniform treatment with function $\gamma$,  we shall %
prefer 
to treat $\lambda$ as a 
function.

\begin{definition}
An {\em edge cover} of a hypergraph $H$ %
is a function 
$\lambda : E(H) \ra \{0,1\}$ 
such that $V(H) = B(\lambda)$. The {\em edge cover number} $\rho(H)$
is the  minimum weight of all edge covers of $H$.
\end{definition}

Note that the edge cover number
can be calculated by the following integer linear 
program (ILP).
 \[
 \begin{aligned}
  \text{minimize: } & \sum_{e\in E(H)} \lambda(e) & \\
  \text{subject to: } & \quad\sum_{\mathclap{e \in E(H), v \in e}} \;\lambda(e) 
                        \geq 1, & & \text{for all } v \in V(H)\\
                      & \lambda(e) \in \{0,1\} & & \text{for all } e \in E(H)
 \end{aligned}
 \]
By substituting all $\lambda(e)$ by $\gamma(e)$ and by relaxing the last condition of the ILP  above
to $\gamma(e) \geq 0$, 
we arrive
at the linear program (LP) for computing the fractional edge cover number
to be defined next.
Note that even though our weight function is defined
to take values between 0 and 1, we do not need to add $\gamma(e) \leq 1$ as a 
constraint,
because implicitly by the minimization itself the weight on an edge for
an edge cover is never greater than 1.
Also note that now the program above
is an LP, which (in contrast to an ILP) can be solved in \ptime even if  $k$ is not fixed.

\begin{definition}
A {\em fractional edge cover} of a hypergraph $H=(V(H),E(H))$ is a 
function 
$\gamma : E(H) \ra [0,1]$  
such that $V(H) = B(\gamma)$. The {\em fractional edge cover number} $\rho^*(H)$ of 
$H$ %
is the 
 minimum weight of all fractional edge covers of $H$.
We write $\cov(\gamma)$ to denote the {\em support\/} of 
$\gamma$, i.e., 
$\cov(\gamma) := \{ e\in E(H) \mid \gamma(e) > 0\}$.
\end{definition}

Clearly, we have $\rho^*(H) \leq \rho(H)$ for every hypergraph $H$, and 
$\rho^*(H)$ can %
be much smaller than $\rho(H)$. However, below we give 
an example,  which is important for 
our proof of Theorem \ref{thm:npcomp} and where $\rho^*(H)$ and $\rho(H)$ 
coincide.

\begin{lemma}
\label{lem:cliquewidth}
    Let $K_{2n}$ be a clique of size $2n$. Then the equalities $\rho(K_{2n}) = \rho^*(K_{2n}) 
    = n$ hold.
\end{lemma}

\begin{proof}
Since we have to cover each vertex with weight $\geq 1$, the total 
weight on the 
vertices of the graph is $\geq 2n$. As the weight of each edge adds to the 
weight of 
at most 2 vertices, we need at least weight $n$ on the edges to achieve $\geq 
2n$ 
weight on the vertices.
On the other hand,  we can use $n$ edges each with weight 1 to cover $2n$ 
vertices. Hence, in total, we get
$n \leq \rho^*(K_{2n}) \leq \rho(K_{2n}) \leq  n$.
  \end{proof}

\subsection{HDs, GHDs, and FHDs}

We now define three types of hypergraph decompositions:

\begin{definition}
 \label{def:GHD}
  A {\em generalized hypertree decomposition\/} (GHD) of a hypergraph 
$H=(V(H),E(H))$ 
is a tuple 
$\left< T, (B_u)_{u\in N(T)}, 
(\lambda_u)_{u\in N(T)} \right>$, such that 
$T = \left< N(T),E(T)\right>$ is a rooted tree and 
the 
following conditions hold:

\begin{enumerate}[label=\emph{(\arabic*})]
 \item for each $e \in E(H)$, there is a node $u \in N(T)$ with $e \subseteq 
B_u$;
 \item for each $v \in V(H)$, the set $\{u \in N(T) \mid v \in B_u\}$ is 
connected 
in $T$;
 \item for each $u\in N(T)$, $\lambda_u$ is a function $\lambda_u \colon E(H) 
\ra 
\{0,1\}$
with 
$B_u \subseteq  B(\lambda_u)$.
\end{enumerate}
\end{definition}

Let us clarify some notational conventions used throughout this paper.
To avoid confusion, we will consequently refer to the 
elements in 
$V(H)$ as {\em vertices\/} (of the hypergraph) and to the elements in $N(T)$ as 
the {\em nodes\/}
of $T$ (of the decomposition). 
Now consider a decomposition $\mcG$ with tree structure~$T$. 
For a node $u$ in $T$, 
we write $T_u$ to denote the subtree of $T$ rooted at $u$.
By slight abuse of notation, we will often write $u' \in T_u$ to denote
that $u'$ is a node in the subtree $T_u$ of $T$.
Moreover, we define $\VTu 
:= \bigcup_{u' \in T_u} B_{u'}$
and, for a set $V' 
\subseteq V(H)$, we define $\nodes(V') = 
\{ u \in T \mid B_u \cap V' \neq \emptyset \}$.
If we want to make explicit the decomposition $\mcG$, 
we also write $\nodes(V', \mcG)$ synonymously with $\nodes(V')$.
By further overloading the $\nodes$ operator, we also write 
$\nodes(T_u)$ or $\nodes(T_u, \mcG)$ to denote the 
nodes in a subtree $T_u$ of $T$, i.e.,  $\nodes(T_u) = \nodes(T_u,\mcG) =\{ v \mid v \in T_u \}$.

\begin{definition}
 \label{def:HD}
 A {\em hypertree decomposition\/} (HD) of a hypergraph 
$H=(V(H),E(H))$  is a GHD, which in addition also 
satisfies the following condition:
\begin{enumerate}[label=\emph{(\arabic*})]
 \item[(4)] for each $u\in N(T)$, $ V(T_u) \cap B(\lambda_u) \subseteq B_u$ 
\end{enumerate}
\end{definition}

\begin{definition}
 \label{def:FHD}
 A 
{\em fractional hypertree decomposition\/} (FHD) 
\cite{2014grohemarx}
of a hypergraph 
$H=(V(H),E(H))$ is a tuple 
$\left< T, (B_u)_{u\in N(T)}, (\gamma)_{u\in N(T)} \right>$, where
conditions (1) and (2) of Definition~\ref{def:GHD} plus condition 
(3') hold:
\begin{enumerate}[label=\emph{(\arabic*})]
 \item[(3')] for each $u\in N(T)$, $\gamma_u$ is a function $\gamma_u : E(H) \ra 
[0,1]$
with $B_u \subseteq  B(\gamma_u)$.
\end{enumerate}
\end{definition}

The width of a GHD, HD, or FHD is the maximum weight of the functions 
$\lambda_u$ or $\gamma_u$, 
respectively, over all nodes $u$ in $T$. Moreover, the generalized hypertree 
width,
hypertree width, and fractional hypertree width of $H$ (denoted $\ghw(H)$, 
$\hw(H)$, $\fhw(H)$) is the minimum width over all GHDs, HDs, and FHDs of $H$, 
respectively.
Condition~(2) is called the ``connectedness condition'', and condition~(4) is 
referred to as ``special condition'' \cite{2002gottlob}. The set $B_u$ is often 
referred to as the 
``bag'' at node $u$. 
Note 
that, 
strictly speaking, only HDs require that the underlying tree $T$ be rooted. 
For the sake of a uniform treatment we assume that also the tree underlying a 
GHD or an FHD is rooted (with the understanding that the root is arbitrarily 
chosen).

We now recall two fundamental properties
of the various notions of decompositions and width.

\begin{lemma}
  \label{lem:subhypergraph}
  Let $H$ be a hypergraph and let $H'$ be a vertex induced subhypergraph of $H$, then 
$\hw(H') \leq \hw(H)$,   
$\ghw(H') \leq \ghw(H)$,
and $\fhw(H') \leq \fhw(H)$ hold. 
\end{lemma}

\begin{lemma}
\label{lem:clique}
Let $H$ be a hypergraph. If $H$ has a subhypergraph $H'$ such that 
$H'$ is a 
clique, then every HD, GHD, or FHD of $H$ has a node $u$ 
such 
that $V(H') \subseteq B_u$.
\end{lemma}
\noindent
Strictly speaking, Lemma~\ref{lem:clique} is a well-known property of tree 
decompositions -- independently of the $\lambda$- or $\gamma$-label.

\nop{***************************************
\medskip\noindent
Last, we define the notion of {\em full} nodes. Intuitively, a node $u$ is called full in a decomposition if it is not possible to add to the bag $B_u$
a new vertex $v$ without increasing the width of the decomposition.
  
\begin{definition}
   Let $\mcF= \left< T, (B_u)_{u\in T}, (\gamma_u)_{u\in T} \right>$ be 
   an FHD of $H$ of width $\leq k$,
   then a node $u$ in $T$ is said to be {\em full in $\mcF$\/} 
   (or simply {\em full\/}, if $\mcF$ is understood from the context),
   if for any vertex 
   $v \in V(H) \setminus B_u$
   it is the case that 
   $
   \rho^*(B(\gamma_u) \cup v) > k$.
\end{definition}
****************************************}

\section{NP-Hardness}
\label{sect:hardness}

The main result in this section is the \np-hardness
of \rec{{\it decomp\/},\,$k$} with  
{\it decomp\/} $\in \{$GHD, FHD$\}$ and $k = 2$.
At the core of the \np-hardness proof is the construction of a hypergraph 
$H$ with certain properties. The gadget in Figure~\ref{fig:gadgetH0} will play 
an integral part of this construction. 

\newcommand{\lemGadgetH}{%
Let $M_1$, $M_2$ be disjoint sets and $M=M_1\cup M_2$. Let $H = (V(H),E(H))$ 
be a hypergraph and $H_0 =$ $(V_0, E_A \cup E_B \cup E_C)$ a subhypergraph of 
$H$ with $V_0=\{a_1,a_2,b_1,b_2,c_1,c_2,d_1,d_2\} \cup M$ and
  \begin{align*}
   E_A = \{ &\{a_1,b_1\} \cup M_1, \{ a_2, b_2 \} \cup M_2, 
         \{a_1,b_2\},\{a_2,b_1\}, \{a_1, a_2\} \} \\ 
   E_B = \{ &\{b_1,c_1\} \cup M_1, \{ b_2, c_2 \} \cup M_2, \{b_1,c_2\},
                   \{b_2,c_1\}, \{ b_1,b_2\}, \{c_1,c_2\} \} \\ 
   E_C = \{ & \{c_1,d_1\} \cup M_1, \{ c_2, d_2 \} \cup M_2, 
           \{c_1,d_2\},
                   \{c_2,d_1\}, \{ d_1,d_2\} \}  
  \end{align*}
\noindent  
where no element from the set $R = \{ a_2, b_1, b_2, c_1, c_2, d_1, d_2 \}$ 
occurs in any edge of $E(H) \setminus (E_A\cup E_B\cup E_C)$.
\noindent  
Then, every FHD $\mcF =\left<T,(B_u)_{u\in T},(\gamma_u)_{u\in T}\right>$ of 
width $\leq 2$ of H has nodes $u_A, u_B, u_C$ s.t.:
\begin{itemize} 
 \item $\{a_1,a_2,b_1,b_2\} \subseteq B_{u_A} \subseteq M \cup \{a_1,a_2,b_1,b_2\}$
 \item $B_{u_B} = \{b_1,b_2,c_1,c_2\} \cup M$,
 \item $\{c_1,c_2,d_1,d_2\} \subseteq B_{u_C} \subseteq M \cup \{c_1,c_2,d_1,d_2\}$, and
 \item $u_B$ is on the path from $u_A$ to $u_C$.%
\end{itemize}
}

\begin{figure}[t]
    \centering

\tikzset{
    between/.style args={#1 and #2}{
         at = ($(#1)!0.5!(#2)$)
    }
}

\begin{tikzpicture}[
   vert/.style={fill, circle, inner sep = 1pt},
   ell/.style={ellipse,draw,minimum width=2.5cm, inner sep=0cm}]
   \node[vert,label=above:$a_1$] (a1) {};
   \node[vert,label=below:$a_2$, below=0.8 of a1] (a2) {};
   \node[vert,label=above:$b_1$, right=2 of a1] (b1) {};
   \node[vert,label=below:$b_2$, right=2 of a2] (b2) {};
   \node[vert,label=above:$c_1$, right=2 of b1] (c1) {};
   \node[vert,label=below:$c_2$, right=2 of b2] (c2) {};
   \node[vert,label=above:$d_1$, right=2 of c1] (d1) {};
   \node[vert,label=below:$d_2$, right=2 of c2] (d2) {};

   \draw (a1) -- (a2);
   
   \node[ell, between=a1 and b1] {$M_1$};
   \draw (a1) -- (b2);
   \draw (a2) -- (b1);

   \node[ell, between=a2 and b2] {$M_2$};
   \draw (b1) -- (b2);
   
    \draw (c1) -- (c2);
    \node[ell, between=b1 and c1] {$M_1$};
   \draw (c1) -- (b2);
   \draw (c2) -- (b1);
   \node[ell, between=b2 and c2] {$M_2$};

   \draw (d1) -- (d2);
   \node[ell, between=c1 and d1] {$M_1$};
   \draw (c1) -- (d2);
   \draw (c2) -- (d1);
   \node[ell, between=c2 and d2] {$M_2$};
\end{tikzpicture}

    \caption{Basic structure of $H_0$ in Lemma~\ref{lem:gadgetH0}} 
    \label{fig:gadgetH0}
\end{figure}
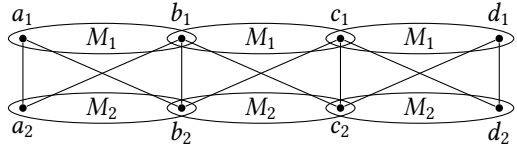

\begin{lemma}
\label{lem:gadgetH0} 
\lemGadgetH
\end{lemma}
\begin{proof}
Consider an arbitrary FHD 
$\mcF = \left<T,(B_u)_{u\in T},(\gamma_u)_{u\in T}\right>$ of 
width $\leq 2$ of H. 
      Observe that $a_1, a_2, b_1$, and $b_2$ form a clique of size 4. Hence, by
      Lemma~\ref{lem:clique}, there is a node $u_A$ in $\mcF$, such that
      $\{a_1,a_2,b_1,b_2\} \subseteq B_{u_A}$. 
      It remains to show that also $B_{u_A} \subseteq M \cup \{a_1,a_2,b_1,b_2\}$ holds.
      To this end, we use a similar reasoning as in the proof of Lemma~\ref{lem:cliquewidth}: 
      to cover each vertex in $\{a_1,a_2,b_1,b_2\}$, we have to put weight $\geq 1$ on each of these 4 vertices. 
      By assumption, the only edges containing 2 out of these 4 vertices are the edges in $E_A \cup \{ \{b_1,b_2\} \}$. 
      All other edges in $E(H)$ contain at most 1 out of these 4 vertices. Hence, in order to cover  
      $\{a_1,a_2,b_1,b_2\}$ with weight $\leq 2$, we are only allowed to put non-zero weight on the edges in 
      $E_A \cup \{ \{b_1,b_2\} \}$. It follows, that $B_{u_A} \subseteq M \cup \{a_1,a_2,b_1,b_2\}$ indeed holds.

      Analogously, for the cliques $b_1, b_2, c_1, c_2$ and $c_1, c_2, d_1, d_2$, 
      there must exist nodes $u_B$ and $u_C$ in $\mcF$ with 
    $\{b_1,b_2,c_1,c_2\}  \subseteq  B_{u_B} \subseteq M \cup \{b_1,b_2,c_1,c_2\}$ and
    $\{c_1,c_2,d_1,d_2\} \subseteq B_{u_C} \subseteq M \cup \{c_1,c_2,d_1,d_2\}$. 
    
    It remains to show that $u_B$ is on the path from $u_A$ to $u_C$ and 
    $B_{u_B} = \{b_1,b_2,c_1,c_2\} \cup M$ holds.
We first show that $u_B$ is on the path between $u_A$ and $u_C$. 
Suppose to the contrary that it is not. We distinguish two cases. 
First, assume that $u_A$ is on the path between $u_B$ and $u_C$. 
Then, by 
connectedness, $\{ c_1, c_2 \} \subseteq B_{u_A}$, which contradicts the property 
$B_{u_A} \subseteq M \cup \{a_1,a_2,b_1,b_2\}$ shown above.
Second, assume $u_C$ is on the path between $u_A$ and $u_B$. In this case, 
we have $\{ b_1, b_2 \} \subseteq B_{u_C}$, which contradicts the property 
$B_{u_C} \subseteq M \cup \{c_1,c_2,d_1,d_2\}$ shown above.

We now show that also $B_{u_B} = \{b_1,b_2,c_1,c_2\} \cup M$ holds.
Since we have already established    
    $\{b_1,b_2,c_1,c_2\}  \subseteq  B_{u_B} \subseteq M \cup \{b_1,b_2,c_1,c_2\}$,
it suffices to show $M \subseteq B_{u_B}$.    
First, let $T'_a$ be the subgraph of $T$ induced by 
$\nodes(\{a_1,a_2\},\mcF)$
     and let $T'_d$ be the subgraph of $T$ induced by 
$\nodes(\{d_1,d_2\},\mcF)$.
     We show that each of the subgraphs $T'_a$ and $T'_d$ is connected
     (i.e., a subtree of $T$) and that the two subtrees are disjoint.
     The connectedness is immediate: by the connectedness condition, each of 
     $\nodes(\{a_1\},\mcF)$, $\nodes(\{a_2\},\mcF)$, 
     $\nodes(\{d_1\},\mcF)$, and $\nodes(\{d_2\},\mcF)$ is connected. 
     Moreover, since $H$ contains an edge $\{a_1,a_2\}$ (resp.\ $\{d_1,d_2\}$), 
     the two subtrees induced by
     $\nodes(\{a_1\},\mcF)$, $\nodes(\{a_2\},\mcF)$ (resp.\ 
     $\nodes(\{d_1\},\mcF)$, $\nodes(\{d_2\},\mcF)$) must be connected, hence
     $T'_a$ and $T'_d$ are subtrees of $T$.
     It remains to show that $T'_a$ and 
     $T'_d$ are disjoint. 
     
Suppose to the contrary that there exists a node $u$ 
which is both in $T'_a$ and in $T'_d$, i.e., 
$a_i, d_j \in B_u$ for some $i \in \{1,2\}$ and $j \in \{1,2\}$.
We claim that $u$ must be on the path between $u_A$ and $u_C$. Suppose
it is not. This means that either 
$u_A$ is on the path between $u$ and $u_C$ or $u_C$ is on the path between $u$ 
and $u_A$. 
In the first case, $B_{u_A}$ has to contain $d_j$ by the connectedness 
condition, which we have already ruled out above. In the second case,
$B_{u_C}$ has to  contain $a_i$, which we have also ruled out above. 
Hence, $u$ is indeed on the path between $u_A$ and $u_C$. 
We have already shown above that also $u_B$ is on the path between $u_A$ and 
$u_C$. Hence, there
are two cases depending on how $u$ and $u_B$ are arranged on the path between 
$u_A$ and $u_C$.
First, assume $u$ is on the path between   $u_A$ and $u_B$. In this case, 
$B_{u_B}$ also contains $d_j$, which we have ruled out above.
Second, assume $u$ is on the path  between $u_B$ and $u_C$. 
Then $B_{u_B}$ has to contain $a_i$, which we have also ruled out above.
Thus, there can be no node $u$ in $T$ with $a_i, d_j \in B_u$ for some 
$i$, $j$ and therefore the subtrees    $T'_a$ and    
$T'_d$ are 
     disjoint and connected by a path containing $u_B$. 
     
As every edge
     must be covered, there are nodes in 
     $T'_a$ that cover $\{a_1,b_1\} \cup M_1$ and $\{a_2,b_2\} \cup 
     M_2$, respectively. 
Hence, the subtree $T'_a$ covers $M = M_1 \cup M_2$, 
i.e., 
     $M \subseteq \bigcup_{u\in T'_a} B_u$. Likewise, 
     $T'_d$ covers $M$. Since both subtrees are
     disjoint and $u_B$ is on the path between them, by the connectedness 
     condition, we have $M \subseteq B_{u_B}$.
\end{proof}

\newcommand{\thmNpcomp}{%
The \rec{{\it decomp\/},\,$k$} problem is \np-complete for 
{\it decomp\/} $\in \{$GHD, FHD$\}$ and $k = 2$.%
}

\begin{theorem}
 \label{thm:npcomp}
\thmNpcomp 
\end{theorem}

\newcommand{\lemComplEdge}{%
   Let $\mcF = \left< T, (B_u)_{u\in T}, (\gamma_u)_{u\in T} \right>$ be 
   an  FHD of width $\leq 2$ of the hypergraph
   $H$ constructed above. For every 
   node $u$ with $S \cup \{z_1, z_2\} \subseteq B_u$ and every pair $e, e'$ of
   complementary edges, it holds that $\gamma_u(e) = \gamma_u(e')$.
}

\newcommand{\lemCovering}{%
   Let $\mcF = \left< T, (B_u)_{u\in T}, (\gamma_u)_{u\in T} \right>$ be 
   an  FHD of width $\leq 2$ of the hypergraph
   $H$ constructed above and let $p \in [2n+3;m]^-$. For every 
   node $u$ with $S \cup A'_p \cup
   \overbar{A_p} \cup \{ z_1, z_2 \} \subseteq B_u$, 
the only way to cover $S 
   \cup A'_p \cup \overbar{A_p} \cup \{ z_1, z_2 \}$ 
by a fractional edge cover $\gamma$ of weight $\leq 2$ is 
by putting non-zero weight 
exclusively on edges 
   $e^{k,0}_p$ and $e^{k,1}_p$ with $k \in \{1,2,3\}$. 
Moreover, 
   $\sum_{k=1}^3 \gamma(e^{k,0}_p) = 1$ and $\sum_{k=1}^3 \gamma(e^{k,1}_p) = 1$  must hold.
}

\newcommand{\clmA}{
The nodes $u'_A,u'_B,u'_C$ (resp. $u_A,u_B,u_C$) are not 
on the path from $u_A$ to $u_C$ (resp. $u'_A$ to $u'_C$).
}

\newcommand{\clmB}{
The following equality holds:
$\nodes(A\cup A',\mcF) \cap \{u_A,u_B,u_C,u'_A,u'_B,u'_C\} = 
\emptyset$.

}

\newcommand{\clmC}{
The FHD $\mcF$ has a path containing nodes %
$\hat{u}_1, \dots, $ $\hat{u}_N$ for some $N$, such that the edges 
$e_{\min \ominus 1}, e_{\min}$, $e_{\min \oplus 1}$, \dots, $e_{\max \ominus 
1}$, $e_{\max}$
are covered in this order. More formally, there is a mapping $f: \{ 
\min\ominus 1,$ $\ldots, \max \} \ra \{1, \ldots, N\}$, s.t.
\begin{itemize}
 \item $\hat{u}_{f(p)}$ covers $e_p$ and
 \item if $p < p'$ then $f(p) \leq f(p')$.
\end{itemize}
By a {\em path containing nodes\/} $\hat{u}_1, \dots, \hat{u}_N$ we mean that 
$\hat{u}_1$ and $\hat{u}_N$ are nodes in $\mcF$, such that the nodes 
$\hat{u}_2, 
\dots, \hat{u}_{N-1}$ lie (in this order) on the path from $\hat{u}_1$ to  
$\hat{u}_N$. 
Of course, the path from $\hat{u}_1$ to  $\hat{u}_N$ may also contain further 
nodes, but we are not interested in whether they 
cover  any of the edges $e_p$.
}

\newcommand{\clmD}{
In the FHD $\mcF$ of $H$ of 
width $\leq 2$, the path from $u_A$ to $u'_A$ has non-empty intersection with
$\ppi$.
}

\newcommand{\clmE}{In the FHD $\mcF$ of $H$ of 
width $\leq 2$ there are two distinct nodes $\hat{u}$ and $\hat{u}'$ in the intersection 
of the path from $u_A$ to $u'_A$ with $\ppi$, s.t.\ $\hat{u}$ is the node in $\ppi$ closest to $u_A$
and $\hat{u}'$ is the node in $\ppi$ closest to $u'_A$. 
Then, on the path $\ppi$, $\hat{u}$ comes before $\hat{u}'$.
See Figure~\ref{fig:u-and-uprime} (a) for a graphical illustration of the arrangement of the 
nodes $\hat{u}_1$, $\hat{u}$, $\hat{u}'$, and 
$\hat{u}_N$
on the path $\ppi$.}

\newcommand{\clmF}{
In the FHD $\mcF$ of $H$ of 
width $\leq 2$ the path %
$\ppi$ has at least 3 nodes $\hat{u}_i$, i.e., $N \geq 
3$. 
}

\newcommand{\clmG}{In the FHD $\mcF$ of $H$ of 
width $\leq 2$ all the nodes $\hat{u}_2, \ldots, \hat{u}_{N-1}$ are on 
the path from $u_A$ to $u'_A$.
For the nodes $\hat{u}$ and $\hat{u}'$ from Claim E, this means that the nodes 
$\hat{u}_1, \hat{u},$ $\hat{u}_2,\hat{u}_{N-1}$, 
$ \hat{u}'$, $\hat{u}_N$ are arranged in precisely this order on the path $\ppi$
from $\hat{u}_1$ to $\hat{u}_N$, cf.\ Figure~\ref{fig:u-and-uprime} (b).
The node $\hat{u}$ may possibly coincide with $\hat{u}_1$    
and $\hat{u}'$ may possibly coincide with $\hat{u}_{N}$. 
}

\newcommand{\clmH}{
 Each of the nodes 
$\hat{u}_1, \dots, \hat{u}_N$ covers 
exactly one of the edges
$e_{\min \ominus 1}$, $e_{\min}$, $e_{\min \oplus 1}$, \dots, $e_{\max \ominus 
1}$, $e_{\max}$.
}

\newcommand{\clmI}{
The constructed truth assignment $\sigma$ %
is %
a model of $\varphi$.
}

\begin{proof}
The problem is clearly in \np: guess a tree decomposition and check in 
polynomial
 time for each node $u$ whether $\rho(B_u) \leq 2$ 
or $\rho^*(B_u) \leq 2$, respectively, holds.
The \np-hardness is proved by a reduction from 3SAT. 
Before presenting this reduction, we first 
introduce some 
useful notation.

\nop{*********************************
Below, we first describe 
the 
problem reduction, i.e., given an arbitrary instance $\varphi$ of 3SAT, we will 
construct 
a hypergraph $H$. For the correctness of this construction, we have to prove 
that 
$\ghw(H)\leq 2$ if and only if $\fhw(H)\leq 2$ if and only if $\varphi$ is 
satisfiable. We will show the two directions of these equivalences separately, 
i.e., 
on the one hand, we show 
$\ghw(H)\leq 2$ $\Rightarrow$ $\varphi$ is satisfiable and 
$\fhw(H)\leq 2$ $\Rightarrow$ $\varphi$ is satisfiable.
on the other hand, we show 
$\ghw(H)\leq 2$ $\Leftarrow$ $\varphi$ is satisfiable and 
$\fhw(H)\leq 2$ $\Leftarrow$ $\varphi$ is satisfiable;

It will be convenient to prove further lemmas before proving these implications.
Moreover, the proofs of the implications will be split into several claims to 
emphasize the structure of the proofs.
*********************************} %

\medskip
\noindent
{\bf Notation.}
\nop{*********************************
From this we  
construct 
a hypergraph $H = (V(H),E(H))$, i.e., 
an instance of \rec{{\it decomp\/},\,$k$} with 
{\it decomp\/} $\in \{$GHD, FHD$\}$ and $k = 2$.
*********************************} %
  For $i,j \geq 1$, we denote $\{1,\ldots,i\} \times \{1,\ldots,j\}$ by $[i;j]$.
  For each $p \in [i;j]$, we denote by $p \oplus 1$ ($p \ominus 1$) 
  the 
  successor (predecessor) of $p$ in the 
  usual lexicographic order on pairs, that is, the order $(1,1),\ldots,(1,j),$ 
   $(2,1),\ldots,(i,1)$, $\ldots, (i,j)$. We 
  refer to the first element $(1,1)$ 
as $\min$ and 
  to the last element $(i,j)$ as $\max$. 
  We denote by $[i;j]^-$ the set $[i;j]\setminus\{\max\}$, i.e.\ $[i;j]$ without
  the last element.
  
Now let $\varphi = \bigwedge_{j=1}^m 
  (L_j^1 
  \vee L_j^2 \vee L_j^3)$ be an arbitrary instance of  3SAT with $m$ clauses 
and 
variables 
  $x_1,\ldots,x_n$. 
From this we  will
construct 
a hypergraph $H = (V(H),E(H))$,
which consists of two copies $H_0, H'_0$ of the 
  (sub-)hypergraph $H_0$ of Lemma~\ref{lem:gadgetH0} plus additional edges
  connecting $H_0$ and $H'_0$. We use the sets $Y = \{ y_1, \ldots, y_n\}$ 
  and $Y' = \{ y'_1, \ldots, y'_n\}$ to encode the truth values of the 
variables 
of $\varphi$. 
  We  denote by $Y_l$ ($Y'_l$) the set $Y 
  \setminus \{y_l\}$ ($Y' \setminus \{ y'_l \}$). Furthermore, we use the sets 
  $A = \{ a_p \mid p \in [2n+3; m] \}$ and $A' = \{ a'_p \mid p \in 
  [2n+3;m]\}$, and
we define the following subsets of  $A$ and $A'$, respectively:
  \begin{align*}
     A_p &= \{ a_{\min},\ldots,a_p \} & 
  \overbar{A_p} &=  \{ a_{p},\ldots,a_{\max} \} \\
     A'_p &= \{ a'_{\min},\ldots,a'_p \} & 
  \overbar{A'_p} &= \{ a'_{p},\ldots,a'_{\max} \} 
  \end{align*}
  
  In addition,
  we will use another set $S$ of elements, that controls and restricts the ways
  in which edges are combined in a possible FHD or GHD. %
  Such a decomposition will have, implied by Lemma~\ref{lem:gadgetH0}, 
  two nodes $u_B$ and $u'_B$ 
  such that $S \subseteq 
  B_{u_B}$ and $S \subseteq B_{u'_B}$. From this, we will reason on the path 
  connecting $u_B$ and $u'_B$.
  
The concrete set $S$ used in our construction of  $H$ is obtained as follows.
Let $Q = [2n+3;m] \cup 
  \{(0,1),(0,0),(1,0)\}$, hence $Q$ is an extension of the set $[2n+3;m]$ with 
  special elements $(0,1),(0,0),(1,0)$. Then we define the set $S$ as 
  $ S = Q \times \{1,2,3\}.$ 

  The elements in $S$
  are pairs, which we denote as $(q \mid k)$. The values $q \in Q$ are themselves pairs of 
integers
  $(i,j)$. Intuitively, $q$ indicates the position of a node on the 
  ``long'' path $\pi$ in the desired FHD or GHD. The integer $k$ 
  refers to a literal in the $j$-th clause.
  We will write the wildcard $*$ to indicate that a component in some element 
of 
$S$
  can take an arbitrary value.  For example,
  $(\min \mid *)$ denotes the set of tuples 
  $(q \mid k)$
  where $q = \min = (1,1)$ and
  $k$ can take an arbitrary value in $\{1,2,3\}$.
  We will denote by $S_p$ the set $(p \mid *)$. For instance, 
  $(\min \mid *)$  will be denoted as $S_{\min}$. Further, for 
  $p \in [2n+3; m]$ and
  $k \in \{1,2,3\}$, we define singletons
  $S^{k}_p = \{ (p \mid k)\}$.    
  
\medskip
\noindent
{\bf Problem reduction.}
Let $\varphi = \bigwedge_{j=1}^m 
  (L_j^1 
  \vee L_j^2 \vee L_j^3)$ be an arbitrary instance of  3SAT with $m$ clauses 
and 
variables 
  $x_1,\ldots,x_n$. 
From this we 
construct a hypergraph $H = (V(H),E(H))$, 
that is,
an instance of \rec{{\it decomp\/},\,$k$} with 
{\it decomp\/} $\in  \{$GHD, FHD$\}$ and $k = 2$.

We start by defining the vertex set $V(H)$: 
  \begin{align*}
    V(H) = &\; S \;\cup\; A \;\cup\; A' \;\cup\;  Y \;\cup\; Y' 
\;\cup\; \{ 
z_1, 
z_2 \} \;\cup \\
    &\; \{ a_1, a_2, b_1, b_2, c_1, c_2, d_1, d_2, %
    a'_1, a'_2, b'_1, b'_2, c'_1, c'_2, d'_1, d'_2 \}.
  \end{align*}

The edges of $H$ are defined in 3 steps. First, we take two 
    copies of the subhypergraph $H_0$ used in Lemma~\ref{lem:gadgetH0}:
    
  \begin{itemize}
   \item Let $H_0 = (V_0, E_0)$ be the hypergraph of Lemma~\ref{lem:gadgetH0} 
   with $V_0=\{a_1,a_2,b_1,b_2$, $c_1,c_2,d_1,d_2\} \cup M_1 \cup M_2$ and 
   $E_0 = E_A \cup E_B \cup E_C$, where we set $M_1 = S \setminus 
   S_{(0,1)} \cup \{z_1\}$ and $M_2 = Y \cup 
   S_{(0,1)} \cup \{z_2\}$.
   \item Let $H'_0 = (V'_0, E'_0)$ be the corresponding hypergraph, 
         with $V'_0=\{a'_1, a'_2, b'_1,$ $b'_2, 
         c'_1, c'_2, d'_1, d'_2\}\cup M'_1 \cup M'_2$ and $E'_A, E'_B, E'_C$ 
are the primed versions of the 
         egde sets $M'_1 = S 
         \setminus S_{(1,0)} \cup \{z_1\}$ and $M'_2 = Y' \cup S_{(1,0)} \cup  
         \{z_2\}$.
   \end{itemize}
  
In the second step, we define the edges which (as we will see) enforce the 
existence of a ``long'' path~$\pi$
between the nodes covering $H_0$ and the nodes covering  $H'_0$ in any 
FHD of width $\leq 2$. 
  \begin{itemize}
   \item $e_{p} = A'_p \cup \overbar{A_p}$,
         for $p \in [2n+3;m]^-$,
   \item $e_{y_i} = \{ y_i, y'_i \}$, for $1 \leq i \leq n$,
   \item For $p = (i,j)  \in [2n+3;m]^-$ and $k \in \{1,2,3\}$:
     \begin{align*}
      e^{k,0}_p = & \begin{cases}
                    \overbar{A_p} \cup (S\setminus S^{k}_p) 
                       \cup Y  \cup \{z_1\} & \mbox{if } L^k_j = x_l \\
                    \overbar{A_p} \cup (S\setminus S^{k}_p)
                       \cup Y_l \cup \{z_1\} & \mbox{if } 
                       L^k_j = \neg x_l, 
                  \end{cases} \\
          e^{k,1}_p = & \begin{cases}
                    A'_p \cup S^{k}_p \cup 
                    Y'_l \cup \{ z_2\} & \mbox{if } L^k_j = x_l \\
                    A'_p \cup S^{k}_p \cup 
                    Y' \cup \{z_2\} & \mbox{if } L^k_j = \neg x_l. 
                  \end{cases}        
     \end{align*}
  \end{itemize}

Finally, we need edges that connect $H_0$ and $H'_0$ with the above edges covered 
by the nodes of the 
``long'' path $\pi$ in a GHD or FHD:
  
  \begin{itemize}
   \item $e^0_{(0,0)}= %
              \{ a_1 \} \cup A \cup S \setminus S_{(0,0)} \cup Y \cup \{ z_1\}
         $
   \item $e^1_{(0,0)} = S_{(0,0)} \cup Y' \cup \{ z_2\}$
   \item $e^0_{\max} = %
                S \setminus S_{\max} \cup Y
          \cup \{ z_1\}$
   \item $e^1_{\max} = \{a'_1 \} \cup A' \cup  S_{\max} \cup 
                Y' \cup \{z_2\}$
  \end{itemize}
  
This concludes the construction of the hypergraph $H$. Before we prove the correctness of the problem reduction, we give an example that will help to illustrate the intuition underlying this construction.

\begin{example}
\label{bsp:NPhardnessProof}
Suppose that an instance of 3SAT is given by the propositional 
formula $\varphi = (x_1 \vee \neg x_2 \vee x_3) \wedge 
(\neg x_1 \vee x_2 \vee \neg x_3)$, i.e.: we have $n= 3$ variables 
and $m = 2$ clauses. From this we construct a hypergraph 
$H = (V(H),E(H))$. 
First, we instantiate the sets $Q,A,A',S,Y$, and $Y'$ from our problem 
reduction.
  \begin{eqnarray*}
A & =  & \{ a_{(1,1)}, a_{(1,2)}, a_{(2,1)}, a_{(2,2)}, \dots, a_{(9,1)}, 
a_{(9,2)} \},
\\
A' & =  & \{ a'_{(1,1)}, a'_{(1,2)}, a'_{(2,1)}, a'_{(2,2)}, \dots, a'_{(9,1)}, 
a'_{(9,2)} \},
\\
Q &  = & \{(1,1), (1,2), (2,1), (2,2), \dots, (9,1), (9,2)\} \cup  \{(0,1),(0,0),(1,0)\}, 
\\
S & = &  Q \times \{1,2,3\}, %
\\
Y & =  & \{y_1, y_2, y_3\}, \\
Y' & = & \{y'_1, y'_2, y'_3\}.
  \end{eqnarray*}
According to our problem reduction, the set $V(H)$ of vertices of $H$ is 
  \begin{align*}
    V(H) = &\; S \;\cup\; A \;\cup\; A' \;\cup\; Y \;\cup\; Y' 
           \;\cup\; \{ z_1, z_2 \} \;\cup \\
    &\; \{ a_1, a_2, b_1, b_2, c_1, c_2, d_1, d_2 \} \cup \{ a'_1, a'_2, 
    b'_1, b'_2, c'_1, c'_2, d'_1, d'_2 \}.
  \end{align*}
The  set $E(H)$ of edges of $H$ is defined in several steps. First, the edges 
in 
$H_0$ and $H'_0$ are defined: We thus have the subsets 
$E_A,E_B,E_C,E'_A,$ $E'_B, E'_C \subseteq E(H)$, whose definition is based on the 
sets 
$M_1 = S \setminus S_{(0,1)} \cup \{z_1\}$,
$M_2 = Y \cup  S_{(0,1)} \cup \{z_2\}$,
$M'_1  =  S \setminus S_{(1,0)} \cup \{z_1\}$, 
and
$M'_2   =  Y' \cup S_{(1,0)} \cup \{z_2\}$.
The definition of the edges 
\begin{eqnarray*}
e_{p} & = & A'_p \cup \overbar{A_p} \hskip52pt \mbox{for } p \in \{(1,1), (1,2), \dots (8,1),(8,2),(9,1)\}, 
\\ 
e_{y_i} & = & \{ y_i, y'_i \} \hskip55pt \mbox{ for } 1 \leq i \leq 3,
\\
e^0_{(0,0)} & = & \{ a_1 \} \cup A \cup S \setminus S_{(0,0)} \cup Y \cup \{ 
z_1\} ,
\\
e^1_{(0,0)} & = & S_{(0,0)} \cup Y' \cup \{ z_2\}, 
\\
e^0_{(9,2)} & = & S \setminus S_{(9,2)} \cup Y \cup \{ z_1\}, \hskip2pt \mbox{ and}
\\
e^1_{(9,2)} & = & \{a'_1 \} \cup A' \cup  S_{(9,2)} \cup Y' \cup \{z_2\}
\end{eqnarray*}
is straightforward. We concentrate on the edges
$e^{k,0}_p$ and $e^{k,1}_p$ 
for $p \in \{(1,1), (1,2), \dots (8,1),(8,2)$, $(9,1)\}$ and $k \in \{1,2,3\}$.
These edges play the key role for covering the bags of the nodes 
along the ``long'' path $\pi$
in any FHD or GHD of $H$. This path can be thought of as being 
structured in 
9 blocks. Consider an arbitrary $i \in \{1, \dots, 9\}$.
Then $e^{k,0}_{(i,1)}$ and $e^{k,1}_{(i,1)}$ encode the $k$-th literal of the 
first clause
and $e^{k,0}_{(i,2)}$ and $e^{k,1}_{(i,2)}$ encode the $k$-th literal of the 
second clause 
(the latter is only defined for $i \leq 8$).
These edges are defined as follows:
the edges $e^{1,0}_{(i,1)}$ and $e^{1,1}_{(i,1)}$ encode the first literal of 
the first clause, i.e., 
the positive literal $x_1$. We thus have
  \begin{eqnarray*}
e^{1,0}_{(i,1)} &  = &  
\overbar{A_{(i,1)}} \cup (S\setminus S^{1}_{(i,1)}) \cup \{y_1,y_2,y_3\}  
\cup 
\{z_1\} \mbox{ and}
\\
e^{1,1}_{(i,1)}  &  = &  A'_{(i,1)} \cup S^{1}_{(i,1)} \cup 
\{y'_2,y'_3\} 
\cup 
\{ z_2\} 
  \end{eqnarray*}
The edges $e^{2,0}_{(i,1)}$ and $e^{2,1}_{(i,1)}$ encode the second literal of 
the first clause, i.e., 
the negative literal $\neg x_2$. Likewise, $e^{3,0}_{(i,1)}$ and 
$e^{3,1}_{(i,1)}$ encode the third literal of the first clause, i.e., 
the positive literal $x_3$. Hence,
  \begin{eqnarray*}
e^{2,0}_{(i,1)} &  = &  
\overbar{A_{(i,1)}} \cup (S\setminus S^{2}_{(i,1)}) \cup \{y_1,y_3\}  \cup 
\{z_1\}, 
\\
e^{2,1}_{(i,1)}  &  = &  A'_{(i,1)} \cup S^{2}_{(i,1)} \cup 
 \{y'_1,y'_2,y'_3\} 
\cup \{ z_2\} 
\\
e^{3,0}_{(i,1)} &  = &  
\overbar{A_{(i,1)}} \cup (S\setminus S^{3}_{(i,1)}) \cup \{y_1,y_2,y_3\}  
\cup 
\{z_1\}, 
\mbox{ and}
\\
e^{3,1}_{(i,1)}  &  = &  A'_{(i,1)} \cup S^{3}_{(i,1)} 
\cup \{y'_1,y'_2\} 
\cup   
\{ z_2\} 
  \end{eqnarray*}
Analogously, the 
edges $e^{1,0}_{(i,2)}$ and $e^{1,1}_{(i,2)}$ (encoding the 
first literal of the second clause, i.e., $\neg x_1$), 
the edges $e^{2,0}_{(i,2)}$ and $e^{2,1}_{(i,2)}$ (encoding the 
second literal of the second clause, i.e., $x_2$), and 
the edges $e^{3,0}_{(i,2)}$ and $e^{3,1}_{(i,2)}$ (encoding the 
third literal of the second clause, i.e., $\neg x_3$) are defined as follows:
  \begin{eqnarray*}
e^{1,0}_{(i,2)} &  = &  
\overbar{A_{(i,2)}} \cup (S\setminus S^{1}_{(i,2)}) \cup \{y_2,y_3\}  \cup 
\{z_1\}, \\
e^{1,1}_{(i,2)}  &  = &  A'_{(i,2)} \cup S^{1}_{(i,2)} \cup 
 \{y'_1,y'_2,y'_3\} 
\cup \{ z_2\}, 
\\ 
e^{2,0}_{(i,2)} &  = &  
\overbar{A_{(i,2)}} \cup (S\setminus S^{2}_{(i,2)}) \cup \{y_1,y_2,y_3\}  
\cup 
\{z_1\}, 
\\
e^{2,1}_{(i,2)}  &  = &  A'_{(i,2)} \cup S^{2}_{(i,2)} 
\cup \{y'_1,y'_3\} 
\cup 
\{ z_2\} 
\\
e^{3,0}_{(i,2)} &  = &  
\overbar{A_{(i,2)}} \cup (S\setminus S^{3}_{(i,2)}) \cup \{y_1,y_2\}  \cup 
\{z_1\}, 
\mbox{ and}
\\
e^{3,1}_{(i,2)}  &  = &  A'_{(i,2)} \cup S^{3}_{(i,2)} \cup 
\{y'_1,y'_2,y'_3\} 
\cup \{ z_2\}.
  \end{eqnarray*}
The crucial property of these pairs of edges 
$e^{k,0}_{(i,j)}$ and $e^{k,1}_{(i,j)}$ is that they together encode
the $k$-th literal of the $j$-th clause in the following way: 
if the literal is of the form $x_l$
(resp.\ of the form $\neg x_l$), 
then 
$e^{k,0}_{(i,j)} \cup e^{k,1}_{(i,j)}$ covers all of $Y \cup Y'$ except for 
$y'_l$ (resp.\ except for $y_l$).

Formula $\varphi$ in this example is clearly satisfiable, e.g., by the truth assignment $\sigma$ with 
$\sigma(x_1)=$ true and $\sigma(x_2) =\sigma(x_3)=$ false. Hence, for the 
problem
reduction to be correct, there must exist a GHD (and thus also an FHD) of width 
2 
of $H$. In Figure~\ref{fig:DecompPath}, the tree structure $T$ plus the 
bags $(B_t)_{t\in T}$ of such a GHD is displayed. Moreover, in   
Table~\ref{tab:np_decomp},
the precise definition of $B_u$ and $\lambda_u$ of every node $u\in T$ is given:
in the column labelled $B_u$, the set of vertices contained in $B_u$ for each node $u \in T$ is shown. 
In the column labelled $\lambda_u$, the two edges with weight 1 are shown. 
For the row with label $u_{p \in [2n+3;m]^-}$, the entry in the last column is 
$e^{k_p,0}_p, e^{k_p,1}_p$. By this we mean that, for every $p$, an appropriate value
$k_p \in \{1,2,3\}$ has to be determined. It will be explained below how to find an appropriate
value $k_p$ for each $p$.
The set $Z$ in the bags of this GHD is defined as 
$Z = \{y_i \mid \sigma(x_i) = $ true\,$\} \cup \{y'_i \mid \sigma(x_i) = $ 
false\,$\}$.
In this example, for the chosen truth assignment $\sigma$, 
we thus have $Z= \{y_1,y'_2,y'_3\}$.
The bags $B_t$ and the edge covers $\lambda_t$ for each $t\in T$ are explained 
below.

The nodes $u_C,u_B,u_A$ to cover the edges of the subhypergraph $H_0$ and the 
nodes
$u'_A,u'_B,u'_C$ to cover the edges of the subhypergraph $H'_0$ are clear by 
Lemma \ref{lem:gadgetH0}. The purpose of the nodes $u_{\min \ominus 1}$ and $u_{\max}$ is mainly to 
make sure that 
each edge $\{y_i,y'_i\}$ is covered by some bag. 
Recall that the set $Z$ contains exactly one of $y_i$ and $y'_i$ for every $i$. 
Hence, the node $u_{\min \ominus 1}$ (resp.\ $u_{\max}$) covers each edge $\{y_i,y'_i\}$, such that
$y'_i \in Z$ (resp.\ $y_i \in Z$). 

We now have a closer look at the nodes $u_{(1,1)}$ to $u_{(9,1)}$
on the ``long'' path $\pi$. More precisely, let us look at the nodes 
$u_{(i,1)}$ and $u_{(i,2)}$ for some $i \in \{1, \dots, 8\}$, i.e., the 
``$i$-th 
block''.
It will turn out that 
the bags at these nodes can be covered by edges from $H$ because $\varphi$ is 
satisfiable.
Indeed, our choice of $\lambda_{u_{(i,1)}}$ and $\lambda_{u_{(i,2)}}$ is guided 
by the literals 
satisfied by the truth assignment $\sigma$, namely: for $\lambda_{u_{(i,j)}}$, 
we have to choose
some $k_j$, such that the $k_j$-th literal in the $j$-th clause is true in 
$\sigma$. 
For instance, we may define 
$\lambda_{u_{(i,1)}}$ and $\lambda_{u_{(i,2)}}$ as follows: 
  \begin{align*}
\lambda_{u_{(i,1)}} &  =  \{ e^{1,0}_{(i,1)}, e^{1,1}_{(i,1)} \} &\lambda_{u_{(i,2)}} &  =  \{ e^{3,0}_{(i,2)}, e^{3,1}_{(i,2)} \} 
  \end{align*}
The covers $\lambda_{u_{(i,1)}}$ and $\lambda_{u_{(i,2)}}$ were chosen because 
the first literal of the first clause and the third literal of the second 
clause 
are true in $\sigma$. 
Now let us verify that
$\lambda_{u_{(i,1)}}$ and $\lambda_{u_{(i,2)}}$ are indeed covers
of $B_{u_{(i,1)}}$ and $B_{u_{(i,2)}}$, respectively.
By the definition of the 
edges $e^{k,0}_{(i,j)}, e^{k,1}_{(i,j)}$ for $j \in \{1,2\}$ and $k \in 
\{1,2,3\}$, it is 
immediate that 
$e^{k,0}_{(i,j)} \cup e^{k,1}_{(i,j)}$
covers 
$\overbar{A_{(i,j)}} \cup A'_{(i,j)} \cup S \cup \{z_1,z_2\}$.
The only non-trivial question is if $\lambda_{u_{(i,j)}}$ also covers $Z$.
Recall that by definition, 
$(e^{1,0}_{(i,1)} \cup e^{1,1}_{(i,1)}) \supseteq (Y \cup Y') \setminus 
\{y'_1\}$.
Our truth assignment $\sigma$ sets $\sigma(x_1) = $ true. Hence, by our 
definition of $Z$,
we have $y_1 \in Z$ and $y'_1 \not\in Z$. This means that 
$e^{1,0}_{(i,1)} \cup e^{1,1}_{(i,1)}$ indeed covers $Z$ and, hence, all of 
$B_{u_{(i,1)}}$. 
Note that we could have also chosen 
$\lambda_{u_{(i,1)}}   =   \{ e^{2,0}_{(i,1)}, e^{2,1}_{(i,1)} \}$, since 
also the second literal of the first clause (i.e., $\neg x_2$) is true in 
$\sigma$. In this case, 
we would have 
$(e^{2,0}_{(i,1)} \cup e^{2,1}_{(i,1)})
\supseteq (Y \cup Y') \setminus \{y_2\}$ and $Z$ indeed does not contain $y_2$.
Conversely, setting  
$\lambda_{u_{(i,1)}}   =   \{ e^{3,0}_{(i,1)}, e^{3,1}_{(i,1)} \}$ would fail, 
because
in this case, $y'_3 \not\in 
(e^{3,0}_{(i,1)} \cup e^{3,1}_{(i,1)})$ since $x_3$ occurs positively in the 
first clause. 
On the other hand, we have $y'_3 \in Z$  by definition of $Z$, 
because $\sigma(x_3) =$ false holds.

Checking that $\lambda_{u_{(i,2)}}$ as defined above covers $Z$ is done 
analogously. Note that in the second
clause, only the third literal is satisfied by $\sigma$. Hence, 
setting $\lambda_{u_{(i,2)}}   =   \{ e^{3,0}_{(i,2)}, $ $e^{3,1}_{(i,2)} \}$ 
is 
the 
only option to cover $B_{u_{(i,2)}}$ (in particular, to cover $Z$). 
Finally, note that $\sigma$ as defined above is not the only satisfying truth 
assignment of $\varphi$. For instance, we could have chosen 
$\sigma(x_1) = \sigma(x_2) = \sigma(x_3) =$ true. In this case, we would 
define $Z = \{ y_1,y_2,y_3\}$ and the covers $\lambda_{u_{(i,j)}}$ would have 
to 
be 
chosen according to an arbitrary choice of one literal per clause that is 
satisfied by 
this assignment $\sigma$.
\hfill$\Diamond$
\end{example}

\nop{****************************************
important observations. First, in every GHD and FHD of the hypergraph defined 
above, the weight on complementary edges must be equal in every node of the
decomposition and second, the only way to cover $S 
   \cup A'_p \cup \overbar{A_p} \cup \{z_1,z_2\}$ with weight $\leq 2$ 
is by 
   using only edges 
   $e^{k,0}_p$ and $e^{k,1}_p$ with $k \in \{1,2,3\}$ (see 
Lemma~\ref{lem:compl_edge} and~\ref{lem:covering} for details).
****************************************} %

  \begin{table*}
   \centering
   \begin{tabular}{|c|c|c|}
     \hline
     $u \in T$ & $B_u$ & $\lambda_u$ \\ 
     \hline 
     $u_C$ & $\{d_1, d_2, c_1, c_2\} \cup Y \cup S \cup \{z_1,z_2\}$ & 
$\{c_1,d_1\}\cup M_1$, 
                                                      $\{c_2,d_2\}\cup M_2$   \\
     $u_B$ & $\{c_1, c_2, b_1, b_2\} \cup Y \cup S \cup \{z_1,z_2\}$ & 
$\{b_1,c_1\}\cup M_1$, 
                                                      $\{b_2,c_2\}\cup M_2$  \\
     $u_A$ & $\{b_1, b_2, a_1, a_2\} \cup Y \cup S \cup \{z_1,z_2\}$ & 
$\{a_1,b_1\}\cup M_1$, 
                                                      $\{a_2,b_2\}\cup M_2$  \\
     $u_{\min\ominus 1}$ & $\{a_1 \} \cup A \cup Y \cup S \cup Z 
\cup \{z_1,z_2\}$ &
           $e^0_{(0,0)}, e^1_{(0,0)}$  \\
     $u_{p \in [2n+3;m]^-}$ & $A'_p \cup \overbar{A_p} \cup 
           S \cup Z \cup \{z_1,z_2\}$ & $e^{k_p,0}_p, e^{k_p,1}_p$  \\
     $u_{\max}$ & $\{ a'_1\} \cup A' \cup Y' \cup S \cup Z \cup \{z_1,z_2\}$ & 
           $e^0_{\max}, e^1_{\max}$  \\
     $u'_A$ & $\{a'_1, a'_2, b'_1, b'_2\} \cup Y' \cup S \cup \{z_1,z_2\}$ & 
$\{a'_1,b'_1\} \cup 
                                              M_1'$, $\{a'_2,b'_2\}\cup M'_2$  
\\
     $u'_B$ & $\{b'_1, b'_2, c'_1, c'_2\} \cup Y' \cup S \cup \{z_1,z_2\}$ & 
$\{b'_1,c'_1\} \cup 
                                              M_1'$, $\{b'_2,c'_2\}\cup M'_2$  
\\
     $u'_C$ & $\{c'_1, c'_2, d'_1, d'_2\} \cup Y' \cup S \cup \{z_1,z_2\}$ & 
$\{c'_1,d'_1\} \cup 
                                              M_1'$, $\{c'_2,d'_2\}\cup M'_2$  
\\
     \hline
   \end{tabular}
   \vspace{3mm}
   \caption{Definition of $B_u$ and $\lambda_u$ for GHD of $H$.}
   \label{tab:np_decomp} 
  \end{table*}
  \begin{figure*}
    \centering
    \includegraphics[width=\textwidth]{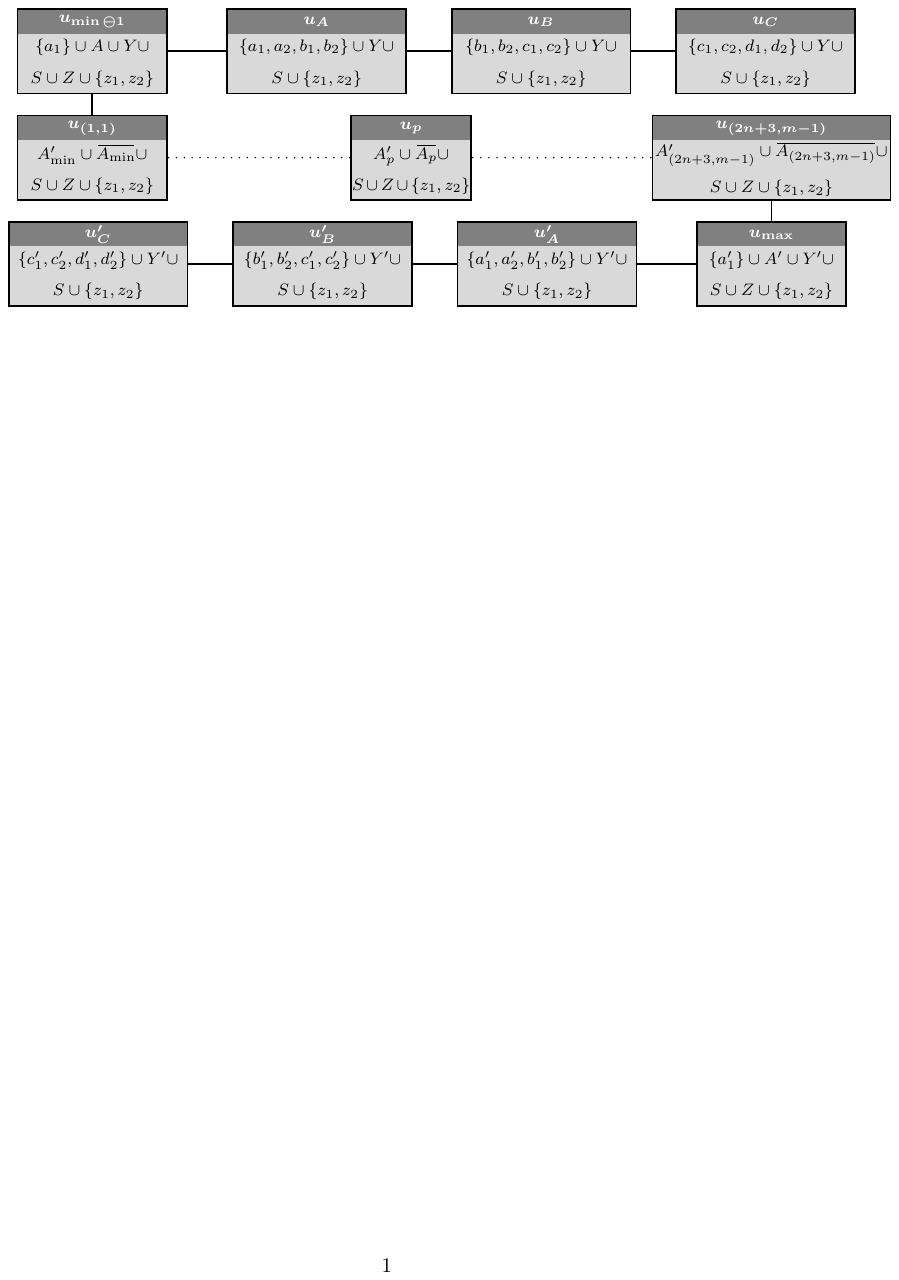}
    \caption{Intended path of the GHD of hypergraph $H$ in the proof of 
             Theorem~\ref{thm:npcomp}} 
    \label{fig:DecompPath}
  \end{figure*}

To prove the correctness of our problem reduction, we have to show the two equivalences: first, that 
$\ghw(H) \leq 2$ if and only if 
$\varphi$ is satisfiable and second, that
$\fhw(H) \leq 2$ 
if 
and only if $\varphi$ is satisfiable. We prove the two directions of these equivalences 
separately.

  \medskip
 \noindent
 {\bf Proof of the ``if''-direction.}
  First assume that $\varphi$ is satisfiable. It suffices to show that 
then 
  $H$ has a GHD of width $\leq 2$, because $\fhw(H) \leq \ghw(H)$ holds. 
  Let $\sigma$ be a 
  satisfying truth assignment. Let us fix for each $j \leq m$, some $k_j \in
  \{ 1,2,3 \}$ such that $\sigma(L^{k_j}_j) = 1$. By $l_j$, we denote the index 
of the variable in the literal $L^{k_j}_j$, that is, $L^{k_j}_j = x_{l_j}$ or 
$L^{k_j}_j = 
\neg x_{l_j}$. For $p=(i,j)$, let $k_p$ refer to $k_j$ and let $L^{k_p}_p$ 
refer to $L^{k_j}_j$. Finally, we define $Z$ as
  $Z = \{ y_i \mid \sigma(x_i) = 1 \} \cup \{ y'_i \mid \sigma(x_i) = 0 \}$.
  
  A GHD  $\mcG = \left< T, (B_u)_{u\in T}, (\lambda_u)_{u\in T} \right>$ of 
width 2 
for $H$ is 
  constructed as follows. $T$ is a path $u_C$, $u_B$, $u_A$, 
$u_{\min\ominus1}$, 
  $u_{\min}$,\dots, $u_{\max}$, $u'_A$, $u'_B$, $u'_C$. The 
construction is 
  illustrated in 
  Figure~\ref{fig:DecompPath}.
  The precise definition of $B_u$ and $\lambda_u$ is given in
  Table~\ref{tab:np_decomp}. Clearly, the GHD  has width $\leq 2$. We 
  now show that $\mcG$ is 
indeed a GHD of $H$:
  \begin{enumerate}
   \item[(1)] For each edge $e \in E$, there is a node $u \in T$, such that 
         $e \subseteq B_u$:
         \begin{itemize}
          \item $\forall e \in E_X : e \subseteq B_{u_X}$ for all $X \in 
\{A,B,C\}$, 
          \item $\forall e' \in E'_X : e' \subseteq B_{u'_X}$ for all $X \in 
\{A,B,C\}$, 
          \item $e_p \subseteq B_{u_p}$ for $p \in [2n+3;m]^-$,
          \item $e_{y_i} \subseteq B_{u_{\min\ominus 1}}$ 
                (if $y'_i \in Z$) or 
                $e_{y_i} \subseteq B_{u_{\max}}$ 
                (if $y_i \in Z$), respectively,
          \item $e^{k,0}_p \subseteq B_{u_{\min\ominus 1}}$ 
                for $p \in [2n+3;m]^-$,
          \item $e^{k,1}_p \subseteq B_{u_{\max}}$ 
                for $p \in [2n+3;m]^-$,
          \item $e^0_{(0,0)} \subseteq B_{u_{\min\ominus 1}}$, 
                $e^1_{(0,0)} \subseteq B_{u_{\max}}$,
          \item $e^0_{\max} \subseteq B_{u_{\min\ominus 1}}$ and 
                $e^1_{\max} \subseteq B_{u_{\max}}$.
         \end{itemize}
         All of the above inclusions can be 
         verified in 
         Table~\ref{tab:np_decomp}. 
        
   \item[(2)] For each vertex $v \in V$, the set $\{ u\in T \mid v \in 
         B_u \}$ induces a connected subtree of $T$, which %
         is easy to 
verify in Table~\ref{tab:np_decomp}.
         
   \item[(3)] For each $u \in T$, $B_u \subseteq B(\lambda_u)$:
         the only inclusion which cannot be easily verified in 
         Table~\ref{tab:np_decomp} is $B_{u_p} \subseteq 
         B(\lambda_{u_p})$. In fact, this is the 
         only place in the proof where we
         make use of the assumption that $\varphi$ is satisfiable.
         First, notice that the set $A'_p \cup \overbar{A_p} \cup
         S \cup \{z_1,z_2\}$ is clearly a subset of $B(\lambda_{u_p})$. It 
remains to
         show that $Z \subseteq B(\lambda_{u_p})$ holds for 
         arbitrary $p \in [2n+3;m]^-$. We show this property by a case distinction
         on the form of $L^{k_p}_p$.
         
         \smallskip
         
         Case (1): First, assume that $L^{k_p}_p = x_{l_j}$ holds. Then
         $\sigma(x_{l_j}) = 1$ and, therefore, $y'_{l_j} \not\in Z$. But, by
         definition of $e^{k_p,0}_p$ and $e^{k_p,1}_p$, vertex $y'_{l_j}$ is the 
only 
         element of $Y \cup Y'$ not contained in $B(\lambda_{u_p})$. 
         Since $Z \subseteq (Y \cup Y')$ and $y'_{l_j} \not\in Z$, we have that
         $Z \subseteq B(\lambda_{u_p})$. 
         
         \smallskip

         Case (2): Now assume that $L^{k_p}_p = \neg x_{l_j}$ holds. Then
         $\sigma(x_{l_j}) = 0$ and, therefore, $y_{l_j} \not\in Z$. But, by
         definition of $e^{k_p,0}_p$ and $e^{k_p,1}_p$, vertex $y_{l_j}$ is the 
only 
         element of $Y \cup Y'$ not contained in $B(\lambda_{u_p})$. 
         Since $Z \subseteq (Y \cup Y')$ and $y_{l_j} \not\in Z$, we have that
         $Z \subseteq B(\lambda_{u_p})$.
  \end{enumerate}

 \smallskip
 \noindent
 {\bf  Two crucial lemmas.}
  Before we prove the ``only if'-direction, 
  we define the notion of complementary 
edges and state two important lemmas related to this notion.

  \begin{definition}
  \label{def:complementary}
   Let $e$ and $e'$ be two edges from the hypergraph $H$ as defined 
before. We say $e'$ is the {\em complementary} edge of $e$ (or, simply, 
$e,e'$ are complementary edges) whenever
   \begin{itemize}
    \item $e \cap S  =  S \setminus S'$ for some $S' \subseteq S$ and
    \item $e' \cap S =  S'$.
   \end{itemize}
  \end{definition}
  
  Observe that for every edge in our construction that covers 
  $S\setminus S'$ for some $S' \subseteq S$ there is a complementary 
edge that covers $S'$, for example 
$e^{k,0}_p$ and $e^{k,1}_p$, $e^0_{(0,0)}$ and $e^1_{(0,0)}$, and so on. In 
particular %
there is no edge that covers $S$ completely. 
Moreover, consider arbitrary subsets $S_1,S_2$ of $S$, 
s.t.\ (syntactically)
$S \setminus S_i$ is part of the definition of $e_i$  for some $e_i \in E(H)$ 
with $i \in \{1,2\}$.
Then $S_1$ and $S_2$ are~disjoint. 

We now present two lemmas needed for the ``only if''-direction.

\begin{lemma}
\label{lem:compl_edge} 
\lemComplEdge
\end{lemma}
\begin{proof}
     First, we try to cover $z_1$ and $z_2$. For $z_1$ we have to put total 
     weight~$1$ on the edges in $E^0$, and to cover 
     $z_2$ we have to put total  weight~$1$ on the edges in $E^1$, where
     \begin{align*}
      E^0 =&   \{ e^{k,0}_p \mid p\in [2n+3;m]^- \mbox{ and } 1 \leq k 
                \leq 3\} \; \cup \\ 
             & \{ e^0_{(0,0)},e^0_{\max} \} \;\cup \\
             & \{\{ a_1,b_1 \} \cup M_1, \{ b_1,c_1\}\cup M_1, \{c_1,d_1\}  
                \cup M_1 \} \;\cup \\
            & \{\{ a'_1,b'_1 \} \cup M'_1, \{ b'_1,c'_1\}\cup M'_1,
                  \{c'_1,d'_1\} \cup M'_1 \} \\
     E^1 = &  \{ e^{k,1}_p \mid p\in [2n+3;m]^- \mbox{ and } 1 \leq k 
               \leq 3\} \; \cup  \\
            & \{ e^1_{(0,0)}, e^1_{\max} \} \; \cup \\
            & \{\{ a_2,b_2 \} \cup M_2, \{ b_2,c_2\}\cup M_2, \{c_2,d_2\}
            \cup M_2 \} \;\cup \\
            & \{ \{ a'_2,b'_2 \} \cup M'_2, \{ b'_2,c'_2\}\cup M'_2,
                  \{c'_2,d'_2\}\cup M'_2 \}
     \end{align*}
     In order to also cover $S$ with weight 2, we are only allowed to 
     assign weights to the above edges.
     Let $S_i$ be a subset of $S$, s.t.\ $S \setminus S_i \subseteq 
     e^0_i$, where $e^0_i \in E^0$. Suppose 
     $\gamma_u(e^0_i) = w_i$. Still, we need to put weight $1$ on the vertices
     in $S_i$. In order to do so, we can put at most weight $1-w_i$ on the 
     edges in $E^0 \setminus \{ e^0_i\}$, which covers $S_i$ with weight at most 
     $1-w_i$.
     The only edge in $E^1$ that intersects $S_i$ is the complementary edge 
     $e^1_i$ of 
     $e^0_i$. Hence, we have to set $\gamma_u(e^1_i) \geq w_i$. 
     This 
     holds for
     all edges $e^1 \in E^1$. Moreover, recall that both 
$\sum_{e^0\in E^0} \gamma_u(e^0) = 1$  and    
     $\sum_{e^1\in E^1} \gamma_u(e^1) = 1$ hold.
Hence,  we cannot afford to set $\gamma_u(e_i^1)>w_i$ for some $i$, since this
would lead to $\sum_{e^1\in E^1} \gamma_u(e^1) > 1$.
We thus have $\gamma_u(e_i^0) = 
     \gamma_u(e_i^1) = w_i$ for every $e_i^0 \in E^0$ and its complementary 
     edge      $e_i^1 \in E^1$. 
  \end{proof}
\begin{lemma}
\label{lem:covering} 
\lemCovering
\end{lemma}
\begin{proof}
     As in the proof of Lemma~\ref{lem:compl_edge}, to cover $z_1$ we 
have to put 
     weight $1$ on the edges in $E^0$ and to cover 
     $z_2$ we have to put weight $1$ on the edges in $E^1$, where $E^0$ and $E^1$ 
     are defined as in the proof of Lemma~\ref{lem:compl_edge}. Since we have 
     $\width(\mcF) \leq 2$, we have to cover $A'_p \cup \overbar{A_p} \cup S$
     with the weight already put on the edges in $E^0 \cup E^1$. 
     In order to cover 
     $A'_p$, we have to put weight 1 on the edges in $E^1_p$, where
     \[ E^1_p = \{ e^{k,1}_r \mid r \geq p \} \cup \{ e^1_{\max} \}. \]
     Notice that, $E^1_p \subseteq E^1$ and 
     therefore $\sum_{e \in E^1 \setminus E^1_p} \gamma_u(e) = 0$. 
     Similar,
     in order to cover $\overbar{A_p}$, we have to put weight 1 on the edges in
     $E^0_p$, where 
     \[ E^0_p = \{ e^{k,0}_s \mid s \leq p \} \cup  \{ e^0_{(0,0)} \}.\] 
     Again, since 
     $E^0_p \subseteq E^0$,
     $\sum_{e \in E^0 \setminus E^0_p} \gamma_u(e) = 0$. It remains to cover 
     $S \cup \{z_1,z_2\}$. By Lemma~\ref{lem:compl_edge}, in order to cover 
$S$, 
$z_1$ and 
     $z_2$, we have to put the same weight $w$ on complementary 
     edges $e$ and $e'$. The only complementary edges in the sets $E^0_p$ and
     $E^1_p$ are edges of the form $e^{k,0}_p$ and $e^{k,1}_p$ with $k \in 
\{1,2,3\}$. 
 In total, we thus have
     $\sum_{k=1}^3 e^{k,0}_p = 1$ and $\sum_{k=1}^3 e^{k,1}_p = 1$.
  \end{proof}

 \noindent
 {\bf Proof of the ``only if''-direction.}
  It remains to show that $\varphi$ is satisfiable if $H$ has a GHD or FHD of 
  width $\leq 2$. Due to the inequality $\fhw(H) \leq \ghw(H)$,
  it suffices to show that $\varphi$ is satisfiable if $H$ 
  has an FHD of  width $\leq 2$.
  For this, let $\mcF = \left< T, (B_u)_{u\in T}, (\gamma_u)_{u\in T} 
  \right>$ be such an 
  FHD. Let $u_A, u_B, u_C$ and $u'_A, u'_B, u'_C$ be the nodes that are
  guaranteed by Lemma~\ref{lem:gadgetH0}.
  We state several %
  properties of the path connecting $u_A$ and $u'_A$, which heavily 
rely on Lemmas~\ref{lem:compl_edge} and~\ref{lem:covering}. 

\medskip
{\sc  Claim A.}{\it \clmA}

\medskip
{\sc Proof of Claim A.}
We only show that none of the nodes $u'_i$ with $i \in \{A,B,C\}$ is on the 
path from $u_A$ to $u_C$. 
The other 
property is shown analogously. Suppose to the contrary that some $u'_i$ is on 
the path from $u_A$ to $u_C$. Since $u_B$ is also on the path between $u_A$ and 
$u_C$ we distinguish two cases:

\begin{itemize}
\item Case~(1): $u'_i$ is on the path between $u_A$ 
and $u_B$; then $\{ b_1, b_2 \} \subseteq B_{u'_i}$. This 
contradicts the property shown in Lemma \ref{lem:gadgetH0} that $u'_i$ cannot cover any
vertices outside $H'_0$. 

\item Case~(2): $u'_i$ is on the path between $u_B$ and $u_C$;
then $\{ c_1, c_2 \} \subseteq B_{u'_i}$, which again contradicts Lemma \ref{lem:gadgetH0}.
\end{itemize}
Hence, the paths from $u_A$ to $u_C$ and from $u'_A$ to $u'_C$ are indeed disjoint. $\hfill\diamond$

\medskip
{\sc  Claim B.}{\it \clmB}

\medskip
{\sc Proof of Claim B.}
Suppose to the contrary that there is a $u_X$ (the proof for $u'_X$ is analogous) 
for some $X \in 
\{A,B,C\}$, s.t.\ $u_X \in \nodes(A\cup A',\mcF)$; then
there is some $a \in (A\cup A')$, s.t.\ $a \in B_{u_X}$. This contradicts 
the property shown in Lemma \ref{lem:gadgetH0} that $u_X$ cannot cover any
vertices outside $H_0$. $\hfill \diamond$

\medskip
We are now interested in the sequence of nodes $\hat{u}_i$ that cover the 
edges 
$e^0_{(0,0)}, e_{\min}, e_{\min \oplus 1}$, \dots,
$e_{\max \ominus 1}$, $e_{\max}$. 
Before we formulate Claim~C, 
it is convenient to introduce the following notation. To be able to refer to 
the 
edges $e^0_{(0,0)}$, $e_{\min}$, $e_{\min \oplus 1}$, \dots, $e_{\max\ominus 
1}$, $e^1_{\max}$ in a uniform way,
we use $e_{\min \ominus 1}$ as synonym of $e^0_{(0,0)}$ and 
$e_{\max}$ as synonym of $e^1_{\max}$. We can thus define the natural order 
$e_{\min \ominus 1} < e_{\min} < e_{\min \oplus 1} < \dots < e_{\max \ominus 1} 
< e_{\max}$ on these edges.

\medskip
{\sc Claim C.}{\it \clmC}

\medskip
{\sc Proof of Claim C.}
Suppose to the contrary that no such path exists. Let $p \geq \min$ be the maximal value such 
that there is a path containing nodes $\hat{u}_1, \hat{u}_2, \ldots, 
\hat{u}_l$, which cover 
$e_{\min\ominus 1}, \ldots, e_p$ in this order. 
Clearly, there exists a node $\hat{u}$ that covers $e_{p \oplus 1} = A'_{p 
\oplus 1} \cup \overbar{A_{p\oplus 1}}$. We distinguish 
four cases:

\begin{itemize}
 \item Case~(1): $\hat{u}_1$ is on the path from $\hat{u}$ to all other nodes
                 $\hat{u}_i$, with $1 < i \leq l$. By the connectedness 
                 condition, $\hat{u}_1$ covers $A'_p$.
                 Hence, in total $\hat{u}_1$ covers $A'_p \cup 
                 A$ with 
                 $A'_{p} = \{a'_{\min}, \dots, a'_{p}\}$ and 
                 $A = \{a_{\min}, \dots, a_{\max}\}$.
                 Then 
                 $\hat{u}_1$ covers all edges $e_{\min\ominus 1}, \ldots, e_p$. 
Therefore, 
                 the path containing nodes $\hat{u}_1$ and $\hat{u}$ covers 
                 $e_{\min\ominus 1}, \ldots,$ $e_{p\oplus 1}$ in this order,
                 which contradicts
                 the maximality of $p$.
  \item Case~(2): $\hat{u} = \hat{u}_1$,
                 hence, $\hat{u}_1$ covers $A'_{p\oplus 1} \cup 
                 A$ with 
                 $A'_{p\oplus 1} = \{a'_{\min}, \dots, a'_{p\oplus 1}\}$ and 
                 $A = \{a_{\min}, \dots$, $a_{\max}\}$.
                 Then, 
                 $\hat{u}_1$ covers all %
                 $e_{\min\ominus 1}, \ldots, 
                 e_{p\oplus 1}$,
                 which contradicts
                 the maximality of $p$.

 \item Case~(3): $\hat{u}$ is on the path from $\hat{u}_1$ to $\hat{u}_l$
           and $\hat{u} \neq \hat{u}_1$. 
                 Hence, $\hat{u}$ is between two nodes $\hat{u}_{i}$ and 
                 $\hat{u}_{i+1}$ for some $1 \leq i < l$ or 
                 $\hat{u} = \hat{u}_{i+1}$ for some $1 \leq i < l-1$.
                 The following
                 arguments hold for both cases.
                 Now, there is some
                 $q \leq p$, such that $e_{q}$ is covered by $\hat{u}_{i+1}$
                 and $e_{q \ominus 1}$ is covered by $\hat{u}_{i}$.
                 Therefore,  $\hat{u}$ covers  $\overbar{A_q}$
                 either by the connectedness condition (if $\hat{u}$ is between 
$\hat{u}_{i}$ and $\hat{u}_{i+1}$)
                 or simply because $\hat{u} = \hat{u}_{i+1}$.
                 Hence, in total, $\hat{u}$
                 covers $A'_{p\oplus 1} \cup \overbar{A_q}$ with 
                 $A'_{p\oplus 1} = \{a'_{\min}, \dots, a'_{p \oplus 1}\}$ and 
                 $\overbar{A_q} = \{a_q, a_{q\oplus 1}, \dots, a_p, a_{p\oplus 
1}, \dots a_{\max}\}$.
                 Then, $\hat{u}$ covers all edges $e_q, e_{q\oplus 1}, 
                 \ldots, e_{p\oplus 1}$. Therefore, the path containing nodes
                 $\hat{u}_1, \dots, \hat{u}_i, \hat{u}$
                 covers $e_{\min\ominus 1}, \ldots, e_{p\oplus 1}$ in this 
order,
                 which contradicts
                 the maximality of $p$.
 \item Case~(4): There is a $u^*$ on the path from $\hat{u}_1$ to $\hat{u}_l$, 
                 such that the paths from $\hat{u}_1$ to $\hat{u}$ and from 
                 $\hat{u}$ to $\hat{u}_l$ go 
                 through $u^*$ and, moreover, $u^* \neq \hat{u}_1$.                 
                 Then, $u^*$ is either between %
                 $\hat{u}_i$ and $\hat{u}_{i+1}$ for some $1 \leq i < l$
                 or
                 $u^* = \hat{u}_{i+1}$ for some $1 \leq i < l-1$. The following
                 arguments hold for both cases. There is some $q \leq p$, 
                 such that
                 $e_q$ is covered by $\hat{u}_{i+1}$ and $e_{q\ominus 1}$ is
                 covered by $\hat{u}_i$. By the connectedness condition, 
                 $u^*$ covers 
                 \begin{itemize}
                  \item $A'_p = \{ a'_{\min},\ldots,a'_p \}$, since $u^*$ is on 
                        the path from $\hat{u}$ to $\hat{u}_l$, and
                  \item $\overbar{A}_q = \{a_q,  \dots, a_p, 
                        a_{p\oplus 1}, \dots a_{\max}\}$, since $u^*$ is on
                        the path from $\hat{u}_1$ to $\hat{u}_{i+1}$ or 
                        $u^* = \hat{u}_{i+1}$.
                 \end{itemize}
                 Then $u^*$ covers all edges $e_q,e_{q\oplus 
                 1},\ldots,e_p$. Therefore, the path containing the nodes 
                 $\hat{u}_1,\ldots,\hat{u}_{i}$, $u^*$, $\hat{u}$ 
                 covers $e_{\min\ominus 1}, 
                 \ldots, e_{p\oplus 1}$ in this order,
                 which contradicts
                 the maximality of~$p$.                
                 $\hfill \diamond$
\end{itemize}

So far we have shown, that there are three disjoint paths from $u_A$ to $u_C$, 
from $u'_A$ to $u'_C$ and from $\hat{u}_1$ to $\hat{u}_N$, respectively. It is 
easy to see, 
that $u_A$ 
is closer to the path $\hat{u}_1$, \dots, $\hat{u}_N$ than $u_B$ and $u_C$, 
since otherwise 
$u_B$ and $u_C$ would have to cover $a_1$ as well, which is impossible 
by Lemma~\ref{lem:gadgetH0}. 
The same also
holds for $u'_A$. In the next claims we will argue that the path from $u_A$
to $u'_A$ goes through some node $\hat{u}$ of the path from $\hat{u}_1$ to 
$\hat{u}_N$. We write $\ppi$ as a short-hand notation for 
the path 
from $\hat{u}_1$ to $\hat{u}_N$. Next, we state some important properties of 
$\ppi$ and the path from $u_A$ to $u_A'$.

\medskip
{\sc  Claim D.} 
{\it \clmD}

\medskip
{\sc Proof of Claim D.} Suppose to the contrary that the path from $u_A$ to 
$u'_A$ is disjoint from 
$\ppi$. We distinguish three cases:
\begin{itemize}
 \item Case~(1):
$u_A$ is on the path from $u'_A$ to (some node in) $\ppi$. Then, by the 
connectedness condition, $u_A$ must contain $a'_1$, which 
contradicts Lemma~\ref{lem:gadgetH0}.
 \item Case~(2): $u'_A$ is on the path from $u_A$ to $\ppi$. 
Analogously to Case (1), we get a contradiction by the fact that then 
$u'_A$ must contain $a_1$.
\item Case~(3): There is a node $u^*$ on the path from $u_A$ to $u'_A$, which 
is closest to $\ppi$, i.e., $u^*$ lies on the path from $u_A$ to $u'_A$ and 
both 
paths, the one
connecting $u_A$ with $\ppi$   and the one connecting $u'_A$ with $\ppi$, go 
through $u^*$. 
Hence, by the connectedness condition, the bag of $u^*$ contains $S \cup 
\{z_1,z_2,$ $a_1,a'_1\}$. By Lemma~\ref{lem:compl_edge}, 
in order to cover $S \cup \{ z_1, z_2 \}$ with weight $\leq 2$,
we are only allowed to put 
non-zero weight on pairs of complementary edges.
However, then it is impossible to achieve also weight $\geq 1$ on 
$a_1$ and $a'_1$ at the same time. $\hfill \diamond$
\end{itemize}

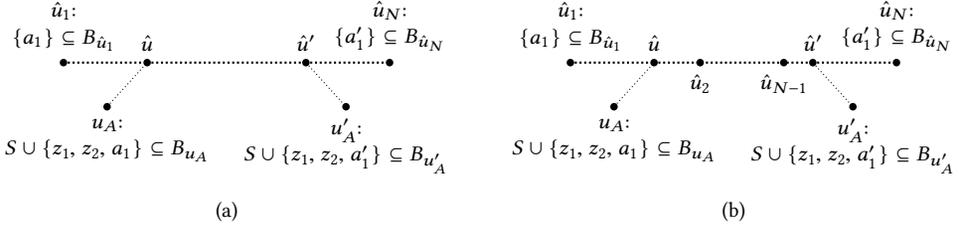
\begin{figure}[t]
    \centering 
  \footnotesize
    
    \begin{minipage}[c]{0.48\textwidth}
       \centering
       \begin{tikzpicture}[
   vert/.style={fill, circle, inner sep = 1pt},
   ell/.style={ellipse,draw,minimum width=2.5cm, inner sep=0cm}]
   \node[vert,label={[align=center]above:$\hat{u}_1$: \\ $\{ a_1 \} \subseteq B_{\hat{u}_1}$}] (u1) {};
   \node[vert,label=above:$\hat{u}$, right=1 of u1] (u) {};
   \node[vert,label={[align=center]below:$u_A$: \\ $S \cup \{ z_1, z_2, a_1 \} \subseteq B_{u_A}$},below right=.5 and .5 of u1] (ua) {};
   \node[vert,label=above:$\hat{u}'$, right=2 of u] (up) {};
   \node[vert,label={[align=center]above:$\hat{u}_N$: \\ $\{ a'_1 \} \subseteq B_{\hat{u}_N}$}, right=1 of up] (un) {};
   \node[vert,label={[align=center]below:$u'_A$: \\ $S \cup \{ z_1, z_2, a'_1 \} \subseteq B_{u'_A}$},below left=.5 and .5 of un] (uap) {};
   
   \draw[densely dotted, thick] (u1) -- (u);
   \draw[densely dotted, thick] (u) -- (up);
   \draw[densely dotted, thick] (up) -- (un);
   \draw[densely dotted] (ua) -- (u);
   \draw[densely dotted] (up) -- (uap);

\end{tikzpicture}

     \centering\medskip\noindent (a)
     \end{minipage}
     \begin{minipage}[c]{0.48\textwidth}
       \centering
       \begin{tikzpicture}[
   vert/.style={fill, circle, inner sep = 1pt},
   ell/.style={ellipse,draw,minimum width=2.5cm, inner sep=0cm}]
   \node[vert,label={[align=center]above:$\hat{u}_1$: \\ $\{ a_1 \} \subseteq B_{\hat{u}_1}$}] (u1) {};
   \node[vert,label=above:$\hat{u}$, right=1 of u1] (u) {};
   \node[vert,label={[align=center]below:$u_A$: \\ $S \cup \{ z_1, z_2, a_1 \} \subseteq B_{u_A}$},below right=.5 and 0.5 of u1] (ua) {};
   \node[vert,label=above:$\hat{u}'$, right=2 of u] (up) {};
   \node[vert,label={[align=center]above:$\hat{u}_N$: \\ $\{ a'_1 \} \subseteq B_{\hat{u}_N}$}, right=1 of up] (un) {};
   \node[vert,label={[align=center]below:$u'_A$: \\ $S \cup \{ z_1, z_2, a'_1 \} \subseteq B_{u'_A}$},below left=.5 and .5 of un] (uap) {};
   
   \node[vert,label=below:$\hat{u}_2$, right=.5 of u] (u2) {};
   \node[vert,label=below:$\hat{u}_{N-1}$, right=1 of u2] (un1) {};

   \draw[densely dotted, thick] (u1) -- (u);
   \draw[densely dotted, thick] (u) -- (up);
   \draw[densely dotted, thick] (up) -- (un);
   \draw[densely dotted] (ua) -- (u);
   \draw[densely dotted] (up) -- (uap);
\end{tikzpicture}

     \centering\medskip\noindent (b)
     \end{minipage}

    \caption{Arrangement of the nodes 
    $\hat{u}_1$, $\hat{u}$, $\hat{u}'$, and $\hat{u}_N$   from Claim E (a) and Claim G (b).} 
    \label{fig:u-and-uprime}
\end{figure}

\medskip
{\sc Claim E.} 
{\it \clmE}

\medskip 
{\sc Proof of Claim E.}
  First, we show that $\hat{u}$ and $\hat{u}'$ are indeed distinct. 
Suppose towards a contradiction that they are not, i.e.\ $\hat{u} = \hat{u}'$. Then, 
by connectedness, $\hat{u}$ has to cover $S \;\cup $ $\{ z_1, 
z_2\}$, because $S \;\cup $ $\{ z_1, z_2\}$ is contained in 
$B_{u_A}$ and in  $B_{u'_A}$.
Moreover, again by connectedness, $\hat{u}$ also has to cover
$\{a_1, a'_1\}$, because $a_1$ is contained in $B_{\hat{u}_1}$ and in $B_{u_A}$ and 
$a'_1$ is contained in $B_{\hat{u}_N}$ and in $B_{u'_{A}}$.
As in Case (3) in the proof of Claim D, this is impossible
by Lemma~\ref{lem:compl_edge}. Hence, $\hat{u}$ and $\hat{u}'$ are distinct. 

Second, we show that, on the path from 
$\hat{u}_1$ to  $\hat{u}_N$, the node $\hat{u}$ comes
before $\hat{u}'$.
Suppose to the contrary that $\hat{u}'$ comes before $\hat{u}$. 
Then, by the connectedness condition,  $\hat{u}$ covers the following (sets of) 
vertices: 
\begin{itemize}
\item $a'_1$, since we are assuming that $\hat{u}'$ comes before $\hat{u}$, i.e., 
$\hat{u}$ is on the path from $\hat{u}_N$ to $u'_A$;
\item $a_1$, since $\hat{u}$ is on the path from  $\hat{u}_1$ to $u_A$;
\item $S \cup \{z_1,z_2\}$, since $\hat{u}$ is on the path from  $u_A$ to 
$u'_A$.
\end{itemize}
In total, $\hat{u}$  has to cover all vertices in 
$S \;\cup $  $\left\{ z_1, z_2, a_1, a'_1\right\}$. 
Again, by Lemma~\ref{lem:compl_edge}, 
this 
is impossible with weight $\leq 2$. $\hfill\diamond$

\medskip
{\sc Claim F.} 
{\it \clmF}

\medskip
{\sc Proof of Claim F.}
First, it is easy to verify that $N \geq 2$ must hold. Otherwise, a single node 
would have to cover
$\{e_{\min \ominus 1}, e_{\min}$, $e_{\min \oplus 1}$, \dots, $e_{\max \ominus 
1}$, $e_{\max}\}$
and, hence, in particular, $S \cup \{z_1,z_2,a_1,a'_1\}$, 
which is impossible as we  have already seen in Case (3) of the proof of Claim D.

It remains to prove $N \geq 3$. Suppose to the contrary that $N = 2$. By the 
problem reduction, hypergraph $H$ has distinct edges $e_{\min \ominus 1}$, $e_{\min}$ and $e_{\max}$.
Hence, 
$\hat{u}_1$ covers at least $e_{\min \ominus 1}$ and $\hat{u}_2$ covers at least $e_{\max}$.
Recall from Claim~E the nodes $\hat{u}$ and $\hat{u}'$, which constitute the endpoints of the intersection
of the path from $u_A$ to $u'_A$ with the path $\ppi$, cf.\ Figure~\ref{fig:u-and-uprime}(a).  
Here we are assuming $N = 2$. We now show that, by the connectedness condition of FHDs,
the nodes $\hat{u}$ and $\hat{u}'$ must cover certain vertices, which will lead to 
a contradiction by Lemma~\ref{lem:compl_edge}.

\begin{itemize}
 \item vertices covered by $\hat{u}$: node $\hat{u}$ is on the path between $u_A$ and $u'_A$. Hence, it covers $S \cup \{ z_1,z_2\}$. Moreover, $\hat{u}$ is on the path between
 $\hat{u}_1$ and $u_A$ (or even coincides with $\hat{u}_1$). Hence, 
 it also covers $a_1$. In total, $\hat{u}$  covers at least $S \cup \{ z_1,z_2, a_1\}$.

 \item vertices covered by $\hat{u}'$: node $\hat{u}'$ is on the path between $u_A$ and $u'_A$. Hence, it covers $S \cup \{ z_1,z_2\}$. Moreover, $\hat{u}'$ is on the path between
 $\hat{u}_2$ and $u'_A$ (or even coincides with $\hat{u}_2$). Hence, 
 it also covers $a'_1$. In total, $\hat{u}'$  covers at least $S \cup \{ z_1,z_2, a'_1\}$.
\end{itemize} 

\smallskip

\noindent
One of the nodes $\hat{u}_1$ or $\hat{u}_2$ must also cover the edge $e_{\min}$.
We inspect these 2 cases separately:

\begin{itemize}
\item Case~(1): suppose that the edge $e_{\min}$ is covered by $\hat{u}_1$. 
Then, $\hat{u}_1$ covers vertex $a'_{\min}$, which is also covered by $\hat{u}_2$.
Hence, also $\hat{u}$  covers $a'_{\min}$. In total, $\hat{u}$ covers 
$S \cup \{ z_1,z_2, a_1, a'_{\min}\}$. However, 
by Lemma~\ref{lem:compl_edge}, we know that, 
to cover $S \cup \{ z_1, z_2 \}$ with weight $\leq 2$,
we are only allowed to put 
non-zero weight on pairs of complementary edges. Hence, 
it is impossible to achieve also weight $\geq 1$ on $a_1$ and on $a'_{\min}$
at the same time. 

\item Case~(2): suppose that the edge $e_{\min}$ is covered by $\hat{u}_2$. 
Then, $\hat{u}_2$ covers vertex $a_{\min}$ (actually, it even covers all of $A$), 
which is also covered by $\hat{u}_1$.
Hence, also $\hat{u}'$  covers $a_{\min}$. In total, $\hat{u}'$ covers 
$S \cup \{ z_1,z_2, a'_1, a_{\min}\}$. Again, this is impossible 
by Lemma~\ref{lem:compl_edge}.
\end{itemize}

Hence, the path $\ppi$ indeed has at least 3 nodes $\hat{u}_i$. $\hfill\diamond$

\medskip
{\sc Claim G.} 
{\it \clmG}

\medskip
{\sc Proof of Claim G.}
We have to prove that $\hat{u}$ lies between 
$\hat{u}_1$ and $\hat{u}_2$ (not including $\hat{u}_2$)
and $\hat{u}'$ lies between 
$\hat{u}_{N-1}$ and $\hat{u}_N$ (not including $\hat{u}_{N-1}$).
For the first property, 
suppose to the contrary that 
$\hat{u}$ does not lie between $\hat{u}_1$ and $\hat{u}_2$
or $\hat{u} = \hat{u}_2$.
This means, that there exists $i \in \{2, \dots, N-1\}$
such that $\hat{u}$ lies between $\hat{u}_i$ and $\hat{u}_{i+1}$, 
including  the case that $\hat{u}$
coincides with $\hat{u}_i$. Note that, by Claim E, $\hat{u}$ cannot coincide with $\hat{u}_N$, 
since there is yet another
node $\hat{u}'$ between  $\hat{u}$ and $\hat{u}_N$.

By definition of $\hat{u}_i$ and $\hat{u}_{i+1}$, there is a $p \in [2n+3;m]$, 
such that both $\hat{u}_i$ and $\hat{u}_{i+1}$ cover $a'_p$. 
Then, by the connectedness condition,  $\hat{u}$ covers the following (sets of) 
vertices: 
\begin{itemize}
\item $a'_p$, since $\hat{u}$ is on the path from $\hat{u}_{i}$ 
to $\hat{u}_{i+1}$
(or 
$\hat{u}$ coincides with $\hat{u}_{i}$), 
\item $a_1$, since $\hat{u}$ is on the path from  $\hat{u}_1$ to $u_A$,
\item $S \cup \{z_1,z_2\}$, since $\hat{u}$ is on the path from  $u_A$ to 
$u'_A$.
\end{itemize}
However, by Lemma~\ref{lem:compl_edge}, we know that, 
to cover $S \cup \{ z_1, z_2 \}$ with weight $\leq 2$,
we are only allowed to put 
non-zero weight on pairs of complementary edges.
Hence, it is impossible to achieve also weight $\geq 1$ on $a'_p$ and $a_1$
at the same time.

It remains to show that $\hat{u}'$ lies between 
$\hat{u}_{N-1}$ and $\hat{u}_N$ (not including $\hat{u}_{N-1}$).
Suppose to the contrary that it does not. Then, 
analogously to the above considerations for $\hat{u}$, 
it can be shown that there exists some $p \in [2n+3;m]$, such that $\hat{u}'$
covers the vertices
$S \cup \{z_1,z_2, a_p, a'_1\}$. Again, this is impossible
by Lemma~\ref{lem:compl_edge}. $\hfill \diamond$
\nop{********************************* 
\begin{itemize}
\item $\hat{u}'$ lies between $\hat{u}_{N-1}$ and $\hat{u}_N$ (not including $\hat{u}_{N-1}$):
suppose to the contrary that it does not. This means, that there exists $i \in \{2, \dots, N-1\}$
such that $\hat{u}'$ lies between $\hat{u}_{i-1}$ and $\hat{u}_{i}$. This includes the case that $\hat{u}$
coincides with $\hat{u}_{i}$; by Claim E, $\hat{u}'$ cannot coincide with $\hat{u}_N$, since there is yet another
node $\hat{u}$ between  $\hat{u}'$ and $\hat{u}_1$.
 
By definition of $\hat{u}_{i-1}$ and $\hat{u}_{i}$, there is a $p \in [2n+3;m]$, 
such that both $\hat{u}_{i-1}$ and $\hat{u}_{i}$ cover $a_p$. 
Then, by the connectedness condition,  $\hat{u}$ covers the following (sets of) 
vertices: 
\begin{itemize}
\item $a_p$, since $\hat{u}'$ is on the path from $\hat{u}_{i-1}$
to $\hat{u}_{i}$, 
\item $a'_1$, since $\hat{u}'$ is on the path from  $\hat{u}_N$ to $u'_A$,
\item $S \cup \{z_1,z_2\}$, since $\hat{u}'$ is on the path from  $u_A$ to 
$u'_A$.
\end{itemize}
Again, we get a contradiction with Lemma~\ref{lem:compl_edge}. $\hfill \diamond$
\end{itemize}
*********************************}
\medskip
By Claim C, the decomposition $\mcF$ contains a path 
$\hat{u}_1 \cdots \hat{u}_N$ that covers the edges 
$e_{\min \ominus 1}, e_{\min}$, $e_{\min \oplus 1}$, \dots, $e_{\max \ominus 
1}$, $e_{\max}$
in this order. We next strengthen this  
property by showing that every node $\hat{u}_i$ covers exactly one edge $e_p$.

\medskip
{\sc  Claim H.}{\it \clmH}

\medskip
{\sc Proof of Claim H.}
We prove this property for the ``outer nodes'' $\hat{u}_1$, $\hat{u}_N$ 
and for the ``inner nodes'' $\hat{u}_2 \cdots \hat{u}_{N-1}$ separately.
We start with the ``outer nodes''.
The proof for $\hat{u}_1$ and $\hat{u}_N$ is symmetric. We thus only work out 
the details for 
$\hat{u}_1$. Suppose to the contrary that $\hat{u}_1$ not only covers $e_{\min 
\ominus 1}$ 
but also $e_{\min}$. We distinguish two cases according to the position of 
node $\hat{u}$ in Figure~\ref{fig:u-and-uprime} (b):

\begin{itemize}
 \item Case~(1): $\hat{u} = \hat{u}_1$. 
Then,  $\hat{u}_1$ has to cover the following (sets of) vertices: 
\begin{itemize}
\item $S \cup \{z_1,z_2\}$, since $\hat{u}$ is on the path from $u_A$ to 
$u'_A$ and we are assuming $\hat{u} = \hat{u}_1$. 

\item $a_1$, since $\hat{u}_1$ covers $e_{\min \ominus 1}$,

\item $a'_{\min}$, since we are assuming that $\hat{u}_1$ also covers $e_{\min}$. 

\end{itemize}

By applying Lemma~\ref{lem:compl_edge}, we may conclude that 
the set $S \cup \{z_1,z_2,a_1,a'_{\min}\}$
cannot be covered by a fractional edge cover of weight $\leq 2$.

\item Case~(2): $\hat{u} \neq \hat{u}_1$. Then 
$\hat{u}$ is on the path from $\hat{u}_1$ to 
$\hat{u}_{2}$. Hence, 
$\hat{u}$ has to cover the following (sets of) vertices: 
\begin{itemize}
\item $S \cup \{z_1,z_2\}$, since $\hat{u}$ is on the path from $u_A$ to $u'_A$,

\item $a_1$, since $\hat{u}$ is on the path from $u_A$ to $\hat{u}_1$, 

\item $a'_{\min}$, since $\hat{u}$ is on the path from $\hat{u}_1$ to 
$\hat{u}_2$.
\end{itemize}
As in Case (1) above, 
$S \cup \{z_1,z_2,a_1,a'_{\min}\}$
cannot be covered by a fractional edge cover of weight $\leq 2$
due to Lemma~\ref{lem:compl_edge}.
\end{itemize}

It remains to consider the ``inner'' nodes $\hat{u}_i$ with  $2 \leq i \leq 
N-1$. 
Each such $\hat{u}_i$ has to cover $S \cup \{z_1,z_2\}$ since all these nodes 
are on the path from $u_A$ to $u'_A$ by Claim G. Now suppose that $\hat{u}_i$ 
covers 
$e_p = A'_p \cup \overbar{A_p}$
for some $p \in \{e_{\min}, \dots, e_{\max \ominus 1}\}$. 
By Lemma~\ref{lem:covering}, covering all of the vertices $A'_p \cup 
\overbar{A_p} \cup S \cup \{z_1,z_2\}$ by a fractional edge cover of weight 
$\leq 2$ 
requires that we put total weight $1$ on  the
edges $e^{k,0}_p$ and total weight $1$ on the edges $e^{k,1}_p$
with $k \in \{ 1,2,3\}$.
However, then it is 
impossible to cover also $e_{p'}$ for some $p'$ with $p' \neq p$.
This concludes the proof of Claim F.$\hfill \diamond$

\medskip
We can now associate with each $\hat{u}_i$ for $1 \leq i \leq N$ the 
corresponding 
edge $e_p$ and write $u_p$ to denote the node that covers the edge $e_p$.
By Claim G, we know that all of the nodes $u_{\min} \dots, u_{\max\ominus 1}$ 
are on the path from $u_A$ to $u'_A$. Hence, by the connectedness condition, 
all these nodes cover $S \cup \{z_1,z_2\}$.

We are now ready to construct a satisfying truth assignment $\sigma$ of 
$\varphi$. For each $i \leq 2n+3$, let $X_i$ be the set $B_{u_{(i,1)}} \cap (Y 
\cup Y')$. As 
  $Y \subseteq B_{u_A}$ and $Y' \subseteq B_{u'_A}$, the sequence $X_1 \cap 
  Y, \ldots, X_{2n+3}\cap Y$ is non-increasing and the sequence $X_1 \cap Y', 
  \ldots,
  X_{2n+3}\cap Y'$ is non-decreasing. Furthermore, as all edges $e_{y_i} = \{ 
  y_i, y'_i \}$ must be covered by some node in $\mcF$, we conclude that for 
  each 
  $i$ and $j$, $y_j \in X_i$ or $y'_j \in X_i$.
  Then, there is some $s \leq 2n+2$ such that $X_s = X_{s+1}$. Furthermore,
  all nodes between $u_{(s,1)}$ and $u_{(s+1,1)}$ cover $X_s$. 
  We derive a truth assignment for $x_1, \ldots, x_n$ from $X_s$ as follows. For
  each $l \leq n$, we set $\sigma(x_l) = 1$ if $y_l \in X_s$ and otherwise
  $\sigma(x_l) = 0$. Note that in the latter case $y'_l \in X_s$.

\medskip
{\sc  Claim I.}
{\it \clmI}

\medskip
{\sc Proof of Claim I.}
We have to show that every clause $c_j =   L_j^1   \vee L_j^2 \vee L_j^3$ of $\varphi$ 
is true in $\sigma$. 
Choose an arbitrary $j \in \{1, \dots, m\}$.
We have to show that there exists a literal in $c_j$ which is true in $\sigma$.
To this end, we inspect the node $u_{(s,j)}$, which, by construction, 
lies between $u_{(s,1)}$ and $u_{(s+1,1)}$. 
Let $p = (s,j)$.
Then we have $A'_{p} \cup \overbar{A_{p}} \cup S   \cup  \{  z_1, z_2\}
\subseteq B_{u_{p}}$. Moreover, by the definition of $X_s$, we also have 
$X_s \subseteq B_{u_{p}}$.
  By Lemma~\ref{lem:covering}, the only way to cover
  $B_{u_{p}}$ with weight $\leq 2$ is by using exclusively the edges 
  $e^{k,0}_p$ and $e^{k,1}_p$ with $k \in \{1,2,3\}$.
More specifically, we have 
 $\sum_{k=1}^3 \gamma_{u_{p}}(e^{k,0}_p) = 1$ and $\sum_{k=1}^3 \gamma_{u_{p}}(e^{k,1}_p) = 1$.
Therefore, $\gamma_{u_{p}}(e^{k,0}_p) > 0$ for some $k$. 
We distinguish two cases depending on the form of literal $L^k_j$:

\begin{itemize}
\item Case (1): First, suppose $L^k_j = x_l$. By Lemma~\ref{lem:compl_edge}, 
complementary edges must have equal weight. Hence, 
from $\gamma_{u_{p}}(e^{k,0}_p) > 0$ it follows that 
also 
  $\gamma_{u_{p}}(e^{k,1}_p) > 0$ holds. Thus, the weight on $y'_l$ is less
  than $1$, which means that $y'_l \not\in B(\gamma_{u_{p}})$ and 
  consequently $y'_l \not\in X_s$. Since this implies that $y_l \in X_s$, we indeed 
have that 
  $\sigma(x_l) = 1$. 
  
\item Case (2):  Conversely, suppose $L^k_j =  \neg x_l$. Since $\gamma_{u_{p}}(e^{k,0}_p) > 0$, the weight on $y_l$ is 
  less
  than $1$, which means that $y_l \not\in B(\gamma_{u_{p}})$ and 
  consequently $y_l \not\in X_s$. Hence, we have $\sigma(x_l) = 0$. 
  
\end{itemize}  
In either case, literal  $L^k_p$ is satisfied by $\sigma$ and therefore, the $j$-th clause $c_j$ is 
satisfied by $\sigma$. Since $j$ was arbitrarily chosen, $\sigma$  indeed satisfies $\varphi$.$\hfill \diamond$

\medskip
\noindent
Claim~I completes the proof of Theorem~\ref{thm:npcomp}.
\end{proof}

\noindent
We conclude this section  by mentioning that the above reduction 
is easily extended to $k+\ell$ for arbitrary
$\ell \geq 1$: for integer values $\ell$, simply add a clique of $2 \ell$ fresh 
vertices 
$v_1, \dots, v_{2 \ell}$ to $H$ and connect each $v_i$ with each ``old'' vertex 
in $H$. Now assume a rational value $\ell \geq 1$, i.e., 
$\ell = r / q$ for natural numbers $r,q$ with $r >q > 0$. 
To achieve a rational bound $k + r/q$, we add %
$r$ fresh 
vertices and 
add hyperedges 
$\{v_i, v_{i\oplus 1},\dots, v_{i \oplus (q-1)}\}$ with $i \in \{1, \dots, 
r\}$
to $H$, where $a \oplus b$ denotes $a + b$ modulo $r$. Again, we connect 
each 
$v_i$ with each ``old'' vertex in $H$. With this construction we
can give NP-hardness proofs for any (fractional) $k \geq 3$. For all fractional values
$k < 3$ (except for $k = 2$) different gadgets and ideas might be needed to
prove NP-hardness of \rec{FHD,$k$}, which we leave  for future work.

\section{Efficient Computation of GHDs}
\label{sect:ghd}

As discussed in Section \ref{sect:introduction} we are interested in finding a realistic and non-trivial criterion on hypergraphs that 
makes the \rec{GHD,\,$k$} problem tractable for fixed $k$.
We thus propose here such a simple property, namely the bounded intersection of two or more edges.

\begin{definition}\label{def:bip}
The {\em intersection width} $\iwidth{\HH}$ of a hypergraph $\HH$ is the 
maximum 
cardinality of any intersection $e_1\cap e_2$ of two distinct edges $e_1$ and 
$e_2$ 
of $\HH$.
We say that a hypergraph $H$ has
the {\em $i$-bounded intersection property ($i$-BIP)} if 
$\iwidth\HH\leq i$ holds. 

Let $\classC$ be a class of hypergraphs. 
We say that $\classC$ 
has the  
{\em bounded intersection property (BIP)} if there exists some 
integer 
constant $i$ such that
every hypergraph $H$ in $\classC$
 has the $i$-BIP.
Class $\classC$ has the 
{\em logarithmically-bounded intersection property (LogBIP)}
if for each of its elements $\HH$, $\iwidth{\HH}$ is $\calO(\log n)$, where $n$ 
denotes the size of the
hypergraph $\HH$.
\end{definition}

The BIP criterion is indeed non-trivial, as several well-known 
classes 
of
unbounded $\ghw$  enjoy the 1-BIP, such as  cliques and grids. Moreover, 
our empirical study \cite{pods/FischlGLP19} 
suggests that the overwhelming number of CQs enjoys the $2$-BIP
(i.e., one hardly joins two relations over more than 2 attributes).
To allow for a yet bigger class of hypergraphs, the BIP can be 
relaxed as follows.

\begin{definition}\label{def:bmip}
The {\em $c$-multi-intersection width} $\cmiwidth{c}{\HH}$ of a 
hypergraph 
$\HH$ 
is the maximum cardinality of any intersection $e_1\cap\cdots\cap e_c$  of 
$c$ distinct edges $e_1, \ldots,  e_c$ of $\HH$.  
We say that a hypergraph $H$ has
the 
{\em $i$-bounded $c$-multi-intersection property (\icBMIP)} if
$\cmiwidth{c}{\HH}\leq i$ holds.

Let $\classC$ be a class of hypergraphs. 
We say that $\classC$ 
has the  
{\em bounded multi-intersection property (BMIP)} 
if there exist constants $c$ and $i$ 
such that
every hypergraph $H$ in $\classC$
has the \icBMIP.
Class $\classC$ of hypergraphs has the 
{\em logarithmically-bounded  multi-intersection property (LogBMIP)} if there 
is 
a constant $c$ such that 
for the hypergraphs $\HH\in \classC$, $\cmiwidth{c}{\HH}$ is 
$\calO(\log n)$, 
where $n$ denotes the size of the
hypergraph $\HH$.
\end{definition}

  \begin{figure}
    \centering
     \includegraphics{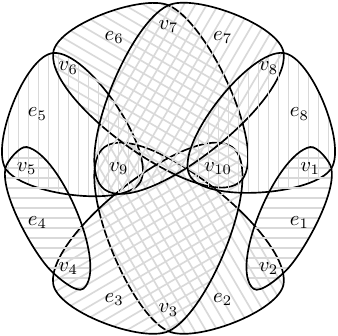}
    \caption{Hypergraph $\HH_0$ from Example \ref{ex:ghw1}} 
    \label{fig:GHDvsHD_graph}
  \end{figure}

\begin{example}
  \label{ex:ghw1}
  
  Figure~\ref{fig:GHDvsHD_graph} shows the hypergraph $H_0 = (V_0,E_0)$ with $ghw(\HH_0)=2$ but $hw(\HH_0)$=3. (which is from \cite{2009gottlob}, which, in turn, was inspired by work of Adler~\cite{adler2004marshals}). 
    Figure \ref{fig:H0_HD}
  shows an HD of width 3 and Figure \ref{fig:H0_GHD} shows GHDs of width 2 for the hypergraph $H_0$. 
  The BIP and the 3-BMIP of $H_0$ is 1. Starting from c=4, the c-BMIP is 0. 
  \hfill$\Diamond$
\end{example}

\begin{figure}
  \centering 
  \footnotesize
    \tikzstyle{htmatrix}=[%
      matrix,
      matrix of nodes,
      nodes={draw=none},
      column 1/.style={nodes={fill=gray!20,align=right},anchor=base east,minimum width=2em},
      column 2/.style={anchor=base west},
      ampersand replacement=\&
    ]
      \begin{tikzpicture}[sibling distance=13em,
     every node/.style = {shape=rectangle,draw,inner sep=1,align=center}]
  \node[htmatrix] { \htnode{$v_1,v_2,v_3,v_6,v_7,v_9,v_{10}$}{$e_1,e_2,e_6$}}
    child { node[htmatrix] {\htnode{$v_3,v_4,v_5,,v_6,v_7,v_9,v_{10}$}{$e_3,e_5$}} }
    child { node[htmatrix] {\htnode{$v_1,v_7,v_8,v_9,v_{10}$}{$e_7,e_8$} }
    };
\end{tikzpicture}
\caption{HD of hypergraph $H_0$ in Figure~\ref{fig:GHDvsHD_graph}}
\label{fig:H0_HD}
\end{figure}
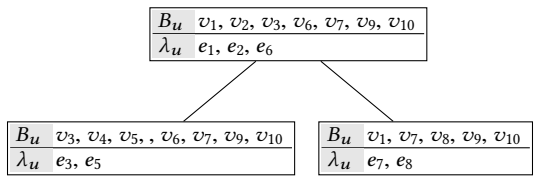

The LogBMIP is %
the most liberal restriction on classes of hypergraphs introduced in 
Definitions~\ref{def:bip} and \ref{def:bmip}. 
The main result in this section will be that the
\rec{GHD,\,$k$} problem with fixed $k$ 
is tractable for any class of hypergraphs satisfying this 
criterion.

Towards this result, first recall that the difference between HDs and GHDs lies 
in the ``special condition'' required by HDs. 
Assume a hypergraph $H = (V(H), E(H))$ and an arbitrary GHD $\mcH=\defGHD$ of $H$. 
Then $\mcH$ is not necessarily an HD, since it may 
contain a special condition violation (SCV), i.e.: there can exist a node $u$, an edge $e \in \lambda_u$ and a 
vertex $v \in V$,  s.t.\ $v \in e$ (and, hence, $v \in B(\lambda_u)$), $v \not\in B_u$ and $v \in V(T_u)$.
Clearly, if we could be sure that $E(H)$ also contains the edge $e ' = e \cap B_u$, then we would simply replace $e$ in $\lambda_u$ 
by $e'$ and would thus get rid of this SCV. 

\begin{example}[Example \ref{ex:ghw1} continued]
   \label{ex:ghw2}
   The GHDs in Figure \ref{fig:H0_GHD} (a) and (b) violate the special condition in node $u$ since the edge $e_2$ containing vertex $v_2$ is in $\lambda_u$ and $v_2$ is in $V(T_u)$ but not in $B_u$. Adding \[e'_2 = e_2 \cap B_u = \{ v_2, v_3, v_9\} \cap \{ v_3, v_6, v_7, v_9,v_{10} \} = \{ v_3, v_9 \}\]
   to $H_0$ and replacing $e_2$ with $e'_2$ in $\lambda_u$ would repair the SCV at node $u$ 
   of the GHDs in Figure \ref{fig:H0_GHD}.
   \hfill$\Diamond$
\end{example}

Now our goal is to define a polynomial-time computable 
function $f$  which, to each hypergraph $\HH$ and integer $k$, 
associates a set $f(\HH,k)$ of additional hyperedges such that $\ghw(\HH)=k$ 
iff 
$\hw(\HH')=k$ with 
$H = (V(H), E(H))$ and
$H' = (V(H), E(H) \cup f(\HH,k))$. 
From this it follows immediately that $\ghw(\HH)$ is computable in polynomial 
time. 
The function $f$ is defined in such a way that $f(\HH,k)$ 
only contains subsets of hyperedges 
of $\HH$. Thus, $f$ is a {\em subedge function} as described 
in~\cite{2009gottlob} and
a GHD of the same width can be easily obtained from any HD of 
$\HH'$.
It is easy to see and well-known \cite{2009gottlob}
that  for 
each subedge function $f$, and each $H$ and $k$,  
$\ghw(\HH) \leq \hw(\HH\cup f(\HH,k))\leq \hw(\HH)$. Moreover, 
for the ``limit'' subedge function $f^+$ where  $f^+(\HH,k)$ 
consists of all possible non-empty subsets of edges of $\HH$,  we have that 
$\hw(\HH\cup f^+(\HH,k)) = \ghw(\HH)$~\cite{adler2004marshals,2009gottlob}.
Of course, in general, $f^+$ contains an 
exponential number of edges. The important point is that our function $f$ will achieve the 
same, while generating a polynomial and \ptime-computable set of edges only. 

  \begin{figure}
    \centering 
  \footnotesize
    \tikzstyle{htmatrix}=[%
      matrix,
      matrix of nodes,
      nodes={draw=none},
      column 1/.style={nodes={fill=gray!20,align=right},anchor=base east,minimum width=2em},
      column 2/.style={anchor=base west},
      ampersand replacement=\&
    ]

    \begin{minipage}[c]{0.48\textwidth}
       \centering
       \begin{tikzpicture}[sibling distance=12em,
     every node/.style = {shape=rectangle,draw,inner sep=1,align=center}
     ]
      
  \node (root) [htmatrix] {\htnode{$v_3,v_6,v_7,v_9,v_{10}$}{$e_2,e_6$ }}
        child { node (c1) [htmatrix]{\htnode{$v_3,v_7,v_8,v_9,v_{10}$}{ $e_3,e_7$}} 
           child { node (c2) [htmatrix]{\htnode{$v_1,v_2,v_3,v_8,v_9,v_{10}$}{$e_2,e_8$ }}
                 }
              }
        child { node (c1a) [htmatrix]{\htnode{$v_3,v_6,v_9,v_{10}$}{$e_3,e_5$}}
     child { node (c3) [htmatrix]{\htnode{$v_3,v_4,v_5,v_6,v_9,v_{10}$}{$e_3,e_5$}}
        }
     };
     
     \node [above=of c1a.north west, draw=none, style={color=red}, anchor=south west, yshift=-1cm] { $u'$: };
     
     \node [above=of root.north west, draw=none, style={color=red}, anchor=south west, yshift=-1cm] { $u_0 = u$: };

     \node [above=of c1.north west, draw=none, style={color=red}, anchor=south west, yshift=-1cm] { $u_1$: };
     \node [above=of c2.north west, draw=none, style={color=red}, anchor=south west, yshift=-1cm] { $u_2 = u^*$: };

\end{tikzpicture}
     \centering\medskip\noindent (a)
     \end{minipage}
     \begin{minipage}[c]{0.48\textwidth}
       \centering
       \begin{tikzpicture}[sibling distance=12em,
     every node/.style = {shape=rectangle,draw,inner sep=1,align=center}
     ]
     
  \node (root) [htmatrix] {\htnode{$v_3,v_6,v_7,v_9,v_{10}$}{$e_2,e_6$}}
        child { node (c1) [htmatrix]{\htnode{$v_3,v_7,v_8,v_9,v_{10}$}{$e_3,e_7$}} 
           child { node (c2) [htmatrix]{\htnode{$v_1,v_2,v_3,v_8,v_9,v_{10}$}{$e_2,e_8$}}
                 }
              }
     child { node (c3) [htmatrix]{\htnode{$v_3,v_4,v_5,v_6,v_9,v_{10}$}{$e_3,e_5$}}
     };
     
     \node [above=of root.north west, draw=none, style={color=red}, anchor=south west, yshift=-1cm] { $u_0 = u$: };

     \node [above=of c1.north west, draw=none, style={color=red}, anchor=south west, yshift=-1cm] { $u_1$: };
     \node [above=of c2.north west, draw=none, style={color=red}, anchor=south west, yshift=-1cm] { $u_2 = u^*$: };
     
     \node (circ1) [shape=circle, style={color=red}, thick] at ([yshift=-5,xshift=-19]root) {  \;\;\; };
     \node (circ2) [shape=circle, style={color=red}, thick] at ([yshift=5,xshift=-13]c2) {  \;\;\; };
\end{tikzpicture}
     \centering\medskip\noindent (b)
     \end{minipage}

   \caption{(a) non bag-maximal vs. (b) bag-maximal GHD of hypergraph $H_0$ in Figure~\ref{fig:GHDvsHD_graph}}
   \label{fig:H0_GHD}
  \end{figure}
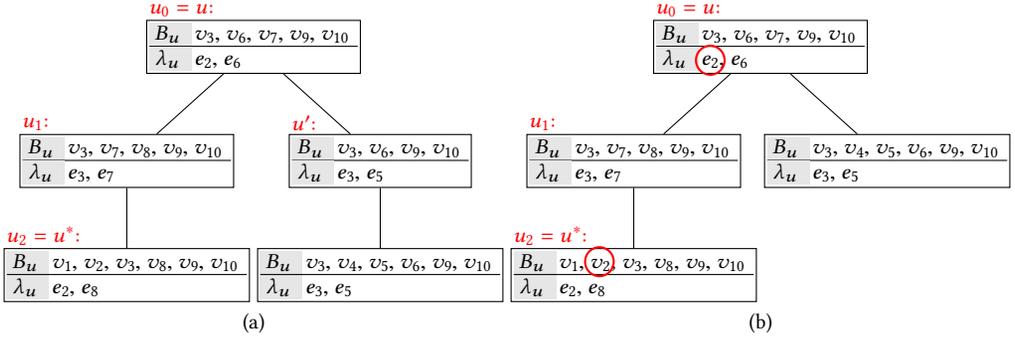

We start by introducing a useful property of GHDs, which we will call 
{\em bag-maximality\/}. 
Let $\mcH=\defGHD$ be a GHD of some hypergraph $H = (V(H), E(H))$. 
For each node $u$ in $T$, we have $B_u \subseteq B(\lambda_u)$ by definition of GHDs
and, in general, $B(\lambda_u) \setminus B_u$ may be non-empty.
We observe that it is sometimes possible to take some vertices from 
$B(\lambda_u) \setminus B_u$ and add them to $B_u$ 
without violating the connectedness condition. Of course, such an addition of vertices to $B_u$ 
does not violate any of the other conditions of GHDs. Moreover, it does not increase the width.

\begin{definition}
 Let $\mcH=\defGHD$ be a GHD of some hypergraph $H = (V(H), E(H))$.
 We call $\mcH$ {\em bag-maximal\/}, if for every node $u$ in $T$,
adding a vertex $v \in B(\lambda_u) \setminus B_u$ to $B_u$ would violate the connectedness condition. 
\end{definition}

It is easy to verify that if $H$ has a GHD of width $\leq k$, then it also has a bag-maximal GHD of width $\leq k$.

\begin{lemma}
 For every GHD $\mcH=\defGHD$ of some hypergraph $H = (V(H), E(H))$, there exists a bag-maximal 
 GHD $\mcH'=\left<T,(B'_u)_{u\in T},(\lambda_u)_{u\in T}\right>$ of $H$, such that $\mcH$ and $\mcH'$ have the same width. 
\end{lemma}

\begin{proof}
Start with a GHD of width $k$ of $\HH$. As long as there exists a
node $u \in T$ and a vertex $ v\in B(\lambda_u) \setminus B_u$, such that $v$ can be 
added to $B_u$ without destroying the GHD properties, select such a node
$u$ and vertex $v$ arbitrarily and add $v$ to $B_u$. By exhaustive application of this 
transformation, 
a bag-maximal GHD of width $k$ of $\HH$  is obtained.
\end{proof}

\begin{example}[Example \ref{ex:ghw2} continued]
\label{ex:ghw3}
Clearly, the GHD in Figure \ref{fig:H0_GHD}(a) violates bag-maximality in node $u'$, since 
the vertices $v_4$ and $v_5$ can be added to $B_{u'}$ without violating any GHD properties. 
If we add $v_4$ and $v_5$ to $B_{u'}$, then bag $B_{u'}$ at node $u'$ and the bag at its child node are the same, which allows us to delete one of the nodes. This results in the GHD given in Figure \ref{fig:H0_GHD}(b), which is bag-maximal. In particular, the vertex $v_2$ cannot be added to $B_{u_0}$: indeed, adding $v_2$ to $B_{u_0}$ would violate the connectedness condition, since $v_2$ is not in $B_{u_1}$ but in $B_{u_2}$.  
\hfill$\Diamond$
\end{example}

So from now on, we will restrict ourselves w.l.o.g.\ to bag-maximal GHDs. 
Before we prove a crucial lemma, 
we introduce some useful notation: 

\begin{definition}
\label{def:critical_ghd}
Let $\mcH=\defGHD$  
be an GHD
of a hypergraph $\HH$. Moreover, let 
$u$ be a node in $\mcH$ and 
let $e \in \lambda_u$ such that 
$e \setminus B_u \neq \emptyset$ holds. 
Let $u^*$ denote the node closest to $u$, such that $u^*$ covers $e$, i.e., 
$e \subseteq B_{u*}$. Then,
we call the path $\pi = (u_0,u_1,\ldots,u_l)$ with $u_0 = u$ and $u_l = u^*$ the
{\em critical path\/} of $(u,e)$ denoted as $\critp(u,e)$.
\end{definition}

\begin{lemma}
\label{lem:crit}
Let $\mcH=\defGHD$  be a bag-maximal GHD of a hypergraph 
$\HH = (V(H),$ $E(H))$, 
let $u \in T$, $e \in \lambda_u$, and $e \setminus B_u \neq \emptyset$. 
Let $\pi = (u_0, u_1, \dots, u_\ell)$ with $u_0 = u$  be the critical path of $(u,e)$.
Then the following equality holds.
$$e\cap B_u=\ \ \ \  e\,\cap\bigcap_{i=1}^\ell B(\lambda_{u_i})
$$
\end{lemma}

\begin{proof} ``$\subseteq$'':
Given that $e \subseteq B_{u_\ell}$ and by the connectedness 
condition, 
$e \cap B_u$ must be a subset of $B_{u_i}$ for every $i \in \{1,\dots,\ell\}$. 
Therefore, $e\cap B_u \subseteq e \cap \bigcap_{i=1}^\ell B(\lambda_{u_i})$ holds.

\smallskip

\noindent
``$\supseteq$'': Assume to the contrary that there exists some vertex 
$v \in e$ with  $v \not\in B_u$ but 
$v \in \bigcap_{i=1}^\ell B(\lambda_{u_i})$. 
By $e \subseteq B_{u_\ell}$, we have $v \in B_{u_\ell}$.
By the connectedness condition, along the path $u_0, \dots, u_\ell$ with $u_0 =u$, there exists $\alpha \in \{0, \dots, \ell-1\}$, s.t.\
$v \not \in B_{u_{\alpha}}$ and $v \in B_{u_{\alpha+1}}$. However, by the assumption, $v \in \bigcap_{i=1}^\ell B(\lambda_{u_i})$ holds. 
In particular, $v \in B(\lambda_{u_{\alpha}})$. Hence, we could 
safely add $v$ to $B_{u_{\alpha}}$ without violating the connectedness condition nor any other 
GHD condition. This contradicts the bag-maximality of $\mcH$. 
\end{proof}

\begin{example}[Example \ref{ex:ghw2} continued]
\label{ex:ghw4}
Consider root node $u$ of the GHD in Figure \ref{fig:H0_GHD}(b). 
We have $e_2 \in \lambda_u$ and $e_2 \setminus B_u = \{v_2\} \neq \emptyset$. On the other hand, 
$e_2$ is covered by $u_2$. Hence,
the critical path of $(u,e_2)$ is $\pi = (u,u_1,u_2)$.
It is easy to verify that 
$e_2 \cap B_u = e_2 \cap (e_3 \cup e_7) \cap (e_8 \cup e_2) = \{v_3,v_9\}$ indeed holds. 
\hfill$\Diamond$
\end{example}

We are now ready to prove the main result of this section.

\newcommand{\thmLogBMIP}{%
For every hypergraph class $\classC$ that enjoys the LogBMIP, and for every 
constant 
$k\geq 1$, the \rec{GHD,\,$k$} problem is tractable, i.e., 
given a hypergraph $\HH$, it is feasible in polynomial time to 
check $\ghw(\HH)\leq k$ and, if so, to compute a GHD of 
width $k$ of $\HH$.%
}

\begin{theorem}\label{theo:LogBMIP}
\thmLogBMIP
\end{theorem}

\begin{proof}
Let $\HH = (V(H),E(H))$ be an arbitrary hypergraph.
Our goal is to show that there exists a polynomially bounded, polynomial-time computable 
set $f(H,k)$ of subedges of $H$, such that 
$\ghw(\HH)=k$ iff  $\hw(\HH')=k$ with $H' = (V(H), E(H) \cup f(\HH,k))$. 
By our considerations above, in order to guarantee the equivalence $\ghw(\HH)=k$ iff  $\hw(\HH')=k$,
it suffices to construct $f(H,k)$ in such a way that, in every GHD $\mcG$ of $H$,
for every node $u$ in $\mcG$, and every edge $e \in \lambda_u$, 
the set $f(H,k)$ contains the 
subedge $e' = e \cap B_u$. 

Let $\mcH=\defGHD$  be a bag-maximal GHD of $\HH$, 
let $u \in T$, $e \in \lambda_u$, and $e \setminus B_u \neq \emptyset$. 
Let $\pi = (u_0, u_1, \dots, u_\ell)$ with $u_0 = u$  be the critical path of $(u,e)$.
By Lemma~\ref{lem:crit}, the equality 
$
e\cap B_u=  e\,\cap\bigcap_{i=1}^\ell B(\lambda_{u_i})
$
holds.
For $i \in \{1, \dots, \ell\}$, let $\lambda_{u_i} = \{e_{i1},\dots, e_{ij_i}\}$
with $j_i \leq k$. 
Then $e\,\cap\bigcap_{i=1}^\ell B(\lambda_{u_i})$ and, therefore, 
also $e\cap B_u$, is of the form 
$$e \cap (e_{11} \cup \dots \cup  e_{1j_1}) \cap \dots \cap (e_{\ell 1} \cup \dots \cup  e_{\ell j_\ell}).
$$
We want to construct $f(\HH,k)$ in such a way that it contains all possible sets  $e\cap B_u$ of vertices. 
To 
this end, we proceed by a stepwise transformation of the above intersection of unions into a union of intersections via distributivity of $\cup$ and $\cap$. 

For $i \in \{ 0, \dots, \ell\}$, 
let $I_i = e \cap \bigcap_{\alpha=1}^i B(\lambda_{u_\alpha}) = e \cap \bigcap_{\alpha=1}^i (e_{\alpha1} \cup \dots \cup  e_{\alpha j_\alpha})$. 
Then $I_0 = e$. For $I_1$ we have to distinguish two cases: if $e \in \lambda_{u_1}$, then $e \subseteq B(\lambda_{i_1})$ and, therefore, 
$I_1 = I_0$. If $e \not\in \lambda_{u_1}$, then 
$I_1 = (e \cap e_{11}) \cup \dots \cup  (e \cap e_{1j_1})$. 
In the latter case, for computing $I_2$, we have to go through all sets $(e \cap e_{1\beta})$ with  $\beta \in \{1, \dots, j_1\}$
and distinguish the two cases if
$\{e, e_{1\beta}\} \cap \lambda_{u_2} \neq \emptyset $ holds or not. If it holds, then we let $(e \cap e_{1\beta})$ in the disjunction of $I_1$ unchanged. Otherwise we replace it by 
$(e \cap e_{1\beta} \cap e_{21}) \cup \dots \cup  (e \cap e_{1\beta} \cap e_{2j_2})$. 
This splitting of intersections into unions of intersections can be iterated over all $i \in \{1, \dots, \ell\}$ in order to arrive 
at $I_\ell = e\,\cap\bigcap_{i=1}^\ell B(\lambda_{u_i}) =  e\cap B_u$, where $I_\ell$ is represented as a union of intersections.

We formalize the computation of the intersections in
$I_0$, \dots, $I_\ell$ by constructing the ``$\bigcup\bigcap$-tree'' in 
Algorithm~\ref{alg:uitree} ``\uitree''.
In a loop over all $i \in \{1, \dots, \ell\}$, we thus compute trees $\calT_i$ such that each node $p$ in $\calT_i$ 
is labelled by  a set $\lab(p)$ of edges. By $\inter(p)$ we denote the intersection of the edges in $\lab(p)$. The parent-child relationship
between a node $p$ and its child nodes $p_1, \dots, p_{j_\alpha}$ corresponds to a splitting step, where the intersection $\inter(p)$
is replaced by the union $(\inter(p) \cap e_{\alpha 1}) \cup \dots \cup (\inter(p) \cap e_{\alpha j_\alpha})$.
It can be proved by a straightforward induction on $i$ that, in the tree $\calT_i$, the union of $\inter(p)$ over all {\em leaf nodes\/} $p$ of $\calT_i$ yields
precisely the union-of-intersections representation of $I_i$. 

We observe that, in the tree $\calT_\ell$, each node has at most $k$ child nodes. Nevertheless, $\calT_\ell$ can become exponentially big since we have no appropriate bound on the length $\ell$ of the critical path. 
Recall, however, that we are assuming the LogBMIP, i.e., there exists a constant $c>1$, s.t.\
any intersection of $\geq c$ edges of $H$ has at most $a \log n$ elements, where $a$ is a constant and $n$ denotes the size of $H$.
Now let $\calT^*$ be the {\em reduced $\cupcap$-tree},
which is obtained from $\calT_\ell$ by cutting off all nodes of depth greater than $c-1$. Clearly, $\calT^*$ has at most 
$k^{c-1}$ leaf nodes and 
the total number of nodes in $\calT^*$  is bounded by $(c-1)  k^{c-1}$. 

The set $f(H,k)$ of subedges that we add to $H$ will consist in 
all possible sets $I_\ell$ that we can 
obtain from all possible critical paths $\pi = (u_0, u_1, \dots, u_\ell)$ in all possible bag-maximal GHDs
$\mcH$ of width $\leq k$ of $H$. We only show that, in case of the LogBMIP, 
the number of possible sets $I_\ell$ is polynomially bounded.
The polynomial-time computability of this set of sets is then easy to see. 
The set of all possible sets $I_\ell$ is obtained by first considering all possible reduced $\cupcap$-trees $\calT^*$
and then considering all sets $I_\ell$ that correspond to some extension $\calT_\ell$ of $\calT^*$.  

First, let 
$m$ denote the number of edges in $E(H)$, then the 
number of possible reduced $\cupcap$-trees $\calT^*$ for given $H$ and $k$ is bounded by 
$m \cdot m^{(c-1)  k^{c-1}}$. 
This can be seen as follows: we can first construct the complete $k$-ary tree of depth $c-1$. 
Clearly, this tree has $\leq (c-1)  k^{c-1}$ nodes. The root is labelled with edge $e$. Now we may label each other node in this tree
either by a set of edges which is obtained from the label of its parent by adding one new edge 
(in particular, by an edge different from $e$) to express that such a node with such a label exists in $\calT^*$ . 
Or we may label a node (and consequently all its descendants) by some stop symbol $\bot$ to express that $\calT^*$ shall not contain this node. Hence, in total, we have $m$ choices for the initial $\cupcap$-tree $\calT_0$ (namely the edge $e$ labelling the root)
and 
$\leq m^{(c-1)  k^{c-1}}$ choices to expand $\calT_0$ to $\calT^*$.

It remains to determine the number of possible sets $I_\ell$ that one can get from possible extensions $\calT_\ell$ of $\calT^*$.  
Clearly, if a leaf node in $\calT^*$ is at depth $< c-1$,  then no descendants at all of this node have been cut off. 
In contrast, a leaf node $p$ in $\calT^*$ at depth $ c-1$ may be the root of a whole subtree in $\calT_\ell$. 
Let $U(p)$ denote the union of the intersections represented by all leaf nodes below $p$. By construction of $\calT_\ell$, 
$U(p) \subseteq \inter(p)$ holds. Moreover, by the LogBMIP, 
$|\inter(p)| \leq a \log n$ for some constant $a$. 
Hence, $U(p)$ takes one out of at most $2^{a \log n} = n^a$ possible values. 

In total, an upper bound on the  number of possible sets $I_\ell$ (and, hence, on $|f(H,k)|$) 
is obtained as follows: there are at most $m \cdot m^{(c-1)  k^{c-1}}$ reduced trees $\calT^*$; each such tree 
has at most $k^{c-1}$ leaf nodes, and each leaf node represents at most $n^{a}$ different sets of vertices.
Putting all this together, we conclude that $|f(H,k)|$
is bounded by $m \cdot m^{(c-1)  k^{c-1}} \cdot n^{a  k^{c-1}}$ for some constant $a$.
\end{proof}

\begin{algorithm}[t]
\SetKwData{Left}{left}\SetKwData{This}{this}\SetKwData{Up}{up}
\SetKwFunction{Union}{Union}\SetKwFunction{FindCompress}{FindCompress}

\SetKwData{N}{N}

\SetKwInOut{Input}{input}\SetKwInOut{Output}{output}

\Input{GHD $\mcH$ of $H$, an edge $e \in E(H)$, critical path $\pi = (u_0, \dots, u_\ell)$ of $\mcH$}
\Output{$\cupcap$-tree $T_\ell$}
\BlankLine
\tcc{Initialization: compute $(N,E)$ for $T_0$ }
$N \la \{p\}$\;
$E \la \emptyset$\;
$\lab(p) \la \{e \}$\;
$T \la (N,E)$\;
\BlankLine
\tcc{Compute $T_i$ from $T_{i-1}$ in a loop over $i$}
\For{$i\leftarrow 1$ \KwTo $\ell$}{
\ForEach{leaf node $p$ of $\;T$}{
\If{$\lab(p) \cap \lambda_{u_i}  = \emptyset$}{
Let $\lambda_{u_i} = \{e_{i 1}, \dots,  e_{i j_i}\}$\;
Create new nodes $\{p_1,  \dots,  p_{j_i}\}$\;
           \lFor{$\alpha \la 1$ \KwTo $j_i$}{$\lab(p_\alpha) \la \lab(p_\alpha) \cup \{ e_{i \alpha} \}$}
           $N \la N \cup \{p_1,  \dots,  p_{j_i}\}$\;
           $E \la E \cup \{(p,p_1),  \dots,  (p,p_{j_i})\}$\;
}
}
$T\la(N,E)$\;
}
\caption{\uitree}\label{alg:uitree}
\end{algorithm}%

\begin{example}[Example \ref{ex:ghw4} continued]
\label{ex:ghw5}
The constructed $\cupcap$-tree of the critical path $(u,u_1,u^*)$ of $(u,e_2)$ in Figure \ref{fig:H0_GHD}(b) is given in Figure \ref{fig:uitree}. The intersection of unions $e_2 \cap (e_3 \cup e_7)$ is replaced by the unions of the leaf nodes $(e_2 \cap e_3) \cup (e_2 \cap e_7)$, which yields the same edge $e'_2 = \{ v_3, v_9 \}$ as in Example \ref{ex:ghw2}. 
\hfill$\Diamond$
\end{example}

\begin{figure}[b]
  \centering 
      \begin{tikzpicture}[every node/.style = {shape=rectangle,draw,align=center},level distance=20pt]
  \node { $e_2$ }
    child { node { $e_2, e_3$ } }
    child { node { $e_2, e_7$ }
    };
\end{tikzpicture}
\caption{$\cupcap$-tree of the critical path $(u,u_1,u^*)$ of $(u,e_2)$ in Figure \ref{fig:H0_GHD}(b)}
\label{fig:uitree}
\end{figure}
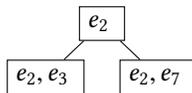

We have defined 
in Section \ref{sect:introduction}
the degree $d$ of a hypergraph $H$. 
We now consider hypergraphs of bounded degree.

\begin{definition}\label{def:bdegree}
We say that a hypergraph $H$ has
the 
{\em $d$-bounded degree property (\dBDP)} if \linebreak 
$\ddegree{\HH}\leq d$ holds.

Let $\classC$ be a class of hypergraphs. 
We say that $\classC$ 
has the  
{\em bounded degree property (BDP)} 
if there exists a constant $d$
such that
every hypergraph $H$ in $\classC$
has the \dBDP.
\end{definition}

The class of 
hypergraphs of bounded degree is an interesting
special case of the class of hypergraphs enjoying the BMIP.
Indeed, suppose that each vertex in a hypergraph $H$ 
occurs 
in at most $d$ edges for some constant $d$. 
Then the intersection of $d+1$ hyperedges is always empty. 
The following corollary is thus immediate.

\begin{corollary}
\label{corol:degreeGHD}
For every class $\classC$ of hypergraphs of bounded degree, for each constant 
$k$, 
the problem  \rec{GHD,\,$k$} is tractable.
\end{corollary}

For the important special case of the BIP, 
the upper bound on $|f(\HH,k)|$ in the proof of Theorem~\ref{theo:LogBMIP},
improves to $m^{k+1}\cdot 2^{k \cdot i}$ .
More specifically, in case of the BIP, the set $f(H,k)$ becomes
$$f(\HH,k)= \bigcup_{e\in E(\HH)}   
\Big({\bigcup_{e_1,\ldots,e_j\in (E(\HH) \setminus \{e\}),\, j\leq k }} 
2^{(e\cap (e_1\cup\cdots\cup e_j))}\Big),$$
i.e., $f(\HH,k)$ contains all subsets of intersections of edges $e \in E(H)$ with unions of $\leq k$ 
edges of $H$ 
different from $e$. In case of the BIP,  the intersection $e\cap (e_1\cup\cdots\cup e_j)$ has at most $i\cdot k$ 
elements. Hence, 
$|f(\HH,k)| \leq m^{k+1}\cdot 2^{k \cdot i}$ holds.
We thus get the following parameterized complexity result.

\newcommand{\thmBMIPfpt}{%
For each constant $k$,
the \rec{GHD,\,$k$} problem is fixed-parameter tractable 
w.r.t.\ the parameter $i$ for hypergraphs enjoying the BIP, i.e., 
in this case, 
\rec{GHD,\,$k$} can be solved in time $\calO(h(i) \cdot \poly(n))$, where 
$h(i)$ is a function depending  on the intersection width $i$ only
and 
$\poly(n)$ is a function that depends polynomially on the size $n$ of 
the given hypergraph $H$.
}

\begin{theorem}\label{theo:BMIPftp}
\thmBMIPfpt
\end{theorem}

\section{Efficient Computation of FHDs}
\label{sect:fhd-exact}

In Section \ref{sect:ghd}, we have shown that under certain 
conditions 
(with the BIP and BDP as most specific and the LogBMIP as most general conditions)
the problem of computing a GHD of width $\leq k$ 
can be reduced to the problem of computing an  HD of width $\leq  k$. 
The key to this problem reduction was to add subedges which allowed us
to repair all possible special condition violations (SCVs) in all possible GHDs of width $\leq k$. 
When trying to 
carry over these ideas from GHDs to FHDs, we encounter {\em two major challenges\/}: Can 
we repair SCVs in an FHD by ideas similar to GHDs? Does 
the 
special condition in case of FHDs allow us 
to extend the HD algorithm from \cite{2002gottlob}  to 
FHDs?

As for the first challenge, recall from Theorem~\ref{theo:LogBMIP} that the 
tractability of \rec{GHD,\,$k$} was achieved by adding polynomially many subedges $f(H,k)$ to a hypergraph $H$,
such that $B_u = B(\lambda_u)$ can be enforced in every node $u$ of a GHD of $H$. In other words, for $S_u = \cov(\lambda_u)$, 
we had $B_u = \bigcup S_u$ in case of GHDs. GHDs with this property clearly satisfy the special condition. We thus reduced the \rec{GHD,\,$k$}  problem to the \rec{HD,\,$k$} problem, which is well-known to be tractable \cite{2002gottlob}.
In contrast, for FHDs, the fractional edge cover function $\gamma_u$ at a node $u$ 
may take any value in $[0,1]$. Therefore, $e \in \cov(\gamma_u)$ (i.e., $\gamma_u(e) > 0$) does not imply $\gamma_u(e) = 1$.
Hence, substantially more work will be needed to achieve $B_u = \bigcup S_u$ with $S_u = \cov(\gamma_u)$
also for FHDs.

As for the second challenge, we will encounter another obstacle compared 
to the 
HD algorithm: a crucial step of the top-down 
construction of an  HD in  \cite{2002gottlob}
is to ``guess'' $\leq k$ edges with $\lambda_u(e) = 1$ for the 
next node $u$ in the HD. However, for a fractional cover $\gamma_u$, we do not 
have such a bound on the number of edges with non-zero weight.
In fact, it is easy to exhibit a family $(H_n)_{n \in \mathbb{N}}$ of 
hypergraphs where it is advantageous to have unbounded $\cov(\gamma_n)$ 
even if $(H_n)_{n \in \mathbb{N}}$ enjoys the BIP, as the following example
illustrates:

\begin{example}
\label{ex:fhwLongEdge}
{\em 
Consider the family $(H_n)_{n \in \mathbb{N}}$ of hypergraphs with $H_n 
=(V_n,E_n)$ defined as follows:

\smallskip
$V_n = \{v_0, v_1, \dots, v_n\}$

$E_n= \{ \{v_0,v_i\} \mid 1 \leq i \leq n\} \cup \{\{v_1, \dots, v_n\}\}$

\smallskip

\noindent
Clearly $\iwidth{H_n} = 1$, but an optimal fractional edge cover of $H_n$ is 
obtained by the following mapping $\gamma$ with 
$\cov(\gamma) = E_n$:

\smallskip

$\gamma(\{v_0, v_i\}) = 1/n$ for each $i \in \{1, \dots, n\}$ and

$\gamma(\{v_1,\dots,  v_n\}) = 1 - (1/n)$ 

\smallskip
\noindent
such that $\weight(\gamma) = 2 - (1/n)$, which is 
optimal in this case. 
}
\end{example}

Nevertheless, in this section, we use the ingredients from 
the \rec{GHD,\,$k$} problem to 
prove a similar (slightly weaker though) tractability result 
for the \rec{FHD,\,$k$} problem. 
More specifically, we shall show that the
\rec{FHD,\,$k$} problem becomes
trac\-table for fixed $k$, 
if we impose the bounded degree property.
Thus, the main result of this section is: 

\newcommand{\thmexactFHD}{%
For every hypergraph class $\classC$ 
that has bounded degree, and for every 
constant $k \geq 1$, 
the 
\rec{FHD,\,$k$} problem is tractable, i.e., 
given a hypergraph $\HH \in \classC$, it is feasible in polynomial time to 
check  $\fhw(\HH)\leq k$ and, if so, to compute an FHD of 
width $k$ of $\HH$.%
}

\begin{theorem}\label{theo:exactFHD}
\thmexactFHD
\end{theorem}

To prove this result, we tackle the two presented main challenges in reverse order. First, we will show that every hypergraph $H$ with degree $d$ allows for an FHD $\mcF$ with bounded support $\cov(\gamma_u)$ 
at every node $u$ of $\mcF$ (Lemma~\ref{lem:bsupp}).  Second, we will devise a polynomial subedge function that allows us to repair all possible SCVs of such $\mcF$  with bounded $\cov(\gamma_u)$ at every node $u$ 
(Lemma~\ref{lem:desired-subedge-function}).

\medskip
\noindent {\bf Bounded Support.} 
First, we show that, for every FHD $\calF$ of width $k$ of a hypergraph $H$
of degree $\leq d$, there exists an FHD $\calF'$ of width $\leq k$ of $H$
satisfying the
following important property: for every node $u$ in the FHD $\calF'$, the 
fractional edge cover $\gamma_u$ has support $\supp(\gamma)$ bounded by a constant that depends only 
on $k$ and $d.$
For this result, we need to introduce, analogously to edge-weight functions 
and edge covers in Section \ref{sect:prelim},
the notions of vertex-weight functions and vertex covers.

\begin{definition}
\label{def:cover-and-support}
A vertex-weight function $w$ for a hypergraph $\HH$ assigns a weight $w(v) \geq 0$ to each vertex $v$ of $\HH$. 
We say that $w$ is a 
{\em fractional vertex cover\/}
of $\HH$  if for each edge  $e\in E(\HH)$, $\Sigma_{v\in e} w(v)\geq 1$ holds.
For a vertex-weight function $w$ for hypergraph $\HH$, we denote by 
$\weight(w)$ its total weight, i.e. $\Sigma_{v\in V(\HH)}w(v)$.
The {\em fractional vertex cover number\/}  $\tau^*(\HH)$ is defined as the minimum 
$\weight(w)$ where $w$ ranges over all fractional  vertex covers of $\HH$.
The {\em vertex support\/} $\vsupp(w)$ of a hypergraph $\HH$ under a vertex-weight function $w$  is defined as $\vsupp(w)=\{v\in V(\HH) \, | \, w(v)>0\}$.
\end{definition}

For our result on bounded support, we will exploit the well-known dualities $\rho^*(\HH)=\tau^*(\HH^d)$ and $\tau^*(\HH)=\rho^*(\HH^d)$, where  $\HH^d$ denotes the dual of $\HH$. 
To make optimal use of this, we make, for the moment,  several assumptions.
First of all, we will assume w.l.o.g. that (1) hypergraphs have no isolated vertices and (2) no empty edges. 
In fact for hypergraphs with isolated vertices (empty edges), $\rho^*$ ($\tau^*$) would be undefined or at least not finite.
Furthermore, we make the following temporary assumptions. Assume that 
(3) hypergraphs never have two distinct vertices of the same ``edge-type'' 
(i.e., the two vertices occur in precisely the same edges) and 
(4) they never have two distinct edges of the same ``vertex-type'' (i.e., we exclude duplicate edges). 

Assumptions (1) -- (4) can be safely made.
Recall that we are ultimately interested in the computation of an FHD of width $\leq k$ for 
given $k$. As mentioned above, without assumption (1), the computation of an edge-weight function and, hence, of an FHD of width $\leq k$ makes no sense. Assumption (2) does not restrict the search for a specific FHD since 
we would never define an edge-weight function with non-zero weight on an empty edge.
As far as assumption (3) is concerned, suppose that a hypergraph $H$ has groups of multiple vertices of identical edge-type. Then it is sufficient to consider 
the reduced hypergraph $H^-$ resulting from $H$ by ``fusing'' each such group to a single vertex. Obviously $\rho^*(H)=\rho^*(H^-)$, and each edge-weight function for $H^-$ 
can be extended in the obvious way to an edge-weight function of the same total weight to $H$.
Finally, assumption (4) can also be made w.l.o.g., since we can again 
define a reduced hypergraph $H^-$ resulting from $H$ by retaining only one edge from each group of identical edges. Then every edge cover of $H^-$ is an edge cover of $H$. Conversely, every 
edge cover of $H$ can be turned into an edge cover of $H^-$ by assigning to each edge $e$ in $H^-$ the sum of the weights of $e$ and all edges identical to $e$ in $H$.

Under our above assumptions  (1) -- (4), 
for every hypergraph $\HH$, the property $\HH^{dd}=\HH$ holds 
and there is an obvious one-to-one correspondence between the edges (vertices) of $\HH$ and the vertices (edges)  of $\HH^d$.
Moreover, there is an obvious one-to-one correspondence between the fractional edge covers of 
$\HH$ and the fractional vertex covers of $\HH^{d}$. In particular, if there is a fractional edge cover $\gamma$ for $\HH$, then its corresponding ``dual'' $\gamma^d$  assigns to each vertex $v$ of $\HH^d$ the same weight as to the edge in $\HH$ that is represented by this vertex and vice versa. 

Note that if we do not make assumptions (3) and (4), then there are hypergraphs $H$ with $H^{dd}\neq H$. For instance, consider the hypergraph $\HH_0$ with $V(\HH_0)=\{a,b,c\}$ and $E(\HH_0)=\{\,e=\{a,b,c\}\,\}$, i.e., property (3) is violated.
The hypergraph $\HH_0^d$ has a unique vertex $e$ and a unique hyperedge $\{e\}$. 
Hence, $H_0^{dd}$ is (isomorphic to) the hypergraph with a unique vertex $a$ and 
a unique hyperedge $\{a\}$, which is clearly different from the original hypergraph~$H_0$.

To get an upper bound on the support $\supp(\gamma)$ of a 
fractional edge cover of a hypergraph $H$, we make use of the following result for fractional vertex covers.
This result is due to Zolt{\'a}n F{\"u}redi~\cite{furedi1988}, who 
extended earlier results by Chung et al.~\cite{chung1988}. 
Below, we appropriately reformulate F{\"u}redi's result for our purposes: 

\begin{proposition}[\cite{furedi1988}, page 152, \ Proposition 5.11.(iii)]\label{prop:furedi}
For every hypergraph $\HH$ of rank (i.e., maximal edge size) $r$, and every fractional vertex cover $w$ for $\HH$  satisfying  $\weight(w)= \tau^*(\HH)$, the 
property $|\vsupp(w)| \leq r\cdot \tau^*(\HH)$ holds.
\end{proposition}

By  duality, exploiting the relationship $\rho^*(\HH)=\tau^*(\HH^d)$ and by recalling that the degree of $\HH$ corresponds to the rank of $\HH^d$, we immediately get the following corollary:

\begin{corollary}\label{cor:bsupp}
For every hypergraph $\HH$ of degree $d$, and every fractional edge cover $\gamma$ for $\HH$ satisfying  $\weight(\gamma) = \rho^*(\HH)$,   
the property
$|\supp(\gamma)|\leq d\cdot \rho^*(\HH)$ holds.
\end{corollary}

From now on, we no longer need to make the assumptions (3) and (4) above. 
In fact, Proposition~\ref{prop:furedi} and Corollary~\ref{cor:bsupp} also hold for hypergraphs that do not fulfil these conditions as was pointed out above by our considerations 
on reduced hypergraphs $H^-$. 
Moreover, from now on, we exclusively concentrate on fractional {\em edge\/} covers. 
The excursion to  fractional {\em vertex\/} covers was only needed to make use of 
F{\"u}redi's result reformulated in Proposition \ref{prop:furedi} above. 

Proposition~\ref{prop:furedi} and Corollary~\ref{cor:bsupp} state bounded support properties for the {\it optimal} weight functions $\tau^*$ and $\rho^*$. 
The following lemma allows us to extend the upper bound 
$k\cdot d$ on the support of a fractional edge cover $\gamma$ of width $k$ of a hypergraph $H$
of degree $d$ to the fractional edge cover $\gamma_u$ in {\em every\/} node $u$ of an FHD of width $\leq k$ of 
$H$.

\begin{lemma}\label{lem:bsupp}
Let $\HH$ be a hypergraph of degree $d$ and let 
$\calF= \defFHD$ 
be an FHD 
of $\HH$ of width $k$. Then there exists  an FHD 
$\calF'= \left< T, (B_u)_{u \in T}, (\gamma'_u)_{u \in T} \right>$ of $\HH$ of width $\leq k$ such that 
$\calF$ and $\calF'$ have exactly the same tree structure $T$
and, 
for every node $u$ in $T$, we have 
$|\supp(\gamma'_u)| \leq k\cdot d$ and, 
$B(\gamma_u) \subseteq B(\gamma'_u)$.
\end{lemma}

\begin{proof}
Let $\HH$ and 
$\calF=\defFHDohne$ 
be as above. For each node $u$ in $T$, consider the sub-hypergraph $\HH_u$ of $\HH$ where  $V(\HH_u)=B(\gamma_u)$ and  
$E(\HH_u)\, =\, \{e\cap V(\HH_u)\,|\,e\in \supp(\gamma_u)\}\ =\ 
\{e\cap B(\gamma_u)\,|\,e\in \supp(\gamma_u)\}$. Note that one  or more edges from $\supp(\gamma_u)$ may give rise  to a same edge $e'$ of $\HH_u$, 
when deleting vertices $v \not\in B(\gamma_u)$ from edges $e \in \supp(\gamma_u)$.
We call all such edges the {\em originators} of $e'$ and denote the set of all originator edges for $e'$ by  $orig(e')$. 

Now let $\gamma^\downarrow_u:E(\HH_u)\rightarrow  (0,1]$ be the edge-weight function which assigns each edge $e'$ of $\HH_u$ weight 
$\gamma^\downarrow_u(e')=\Sigma_{e\in orig(e')}\gamma_u(e)$, i.e., the sum of all weights by $\gamma_u$ of its originators. 
Clearly,  $\gamma^\downarrow_u$ is a fractional edge cover of total weight at most $k$  for $\HH_u$.  Now take an optimal fractional edge cover $\gamma^*_u$ for $\HH_u$. The total weight of this cannot be greater than $k$ either. Hence, by Corollary~\ref{cor:bsupp}, $\supp(\gamma^*_u)\leq k\cdot d$. Now transform $\gamma^*_u$ into an edge-weight function $\gamma'_u$ of the entire  hypergraph $\HH$ by assigning for  edge $e'$ of $\HH_u$ the entire weight of $e'$ to only one of its originators, whilst assigning weight 0 to all other originators. Clearly, the support of $\gamma'_u$  is bounded by $k\cdot d$ and $B(\gamma_u) \subseteq B(\gamma'_u)$. $B_u$ is thus covered by $B(\gamma'_u)$, 
and  the resulting FHD $\calF'$
has all requested properties.
\end{proof}

We have now tackled the first challenge of this section by showing that whenever a hypergraph $H$ has an FHD $\calF$ of width $\leq k$ then $H$ also has an FHD $\calF'$ of width $\leq k$ such that in each node $u$ of $\calF'$ we have $|\supp(\gamma'_u)| \leq k\cdot d$. 
We yet have to overcome the following obstacle: in the alternating algorithm 
in \cite{2002gottlob} for deciding the \rec{{HD},\,$k$} problem, we guess at every node $u$ of the HD a set $S_u$ of edges
with $|S_u| \leq k$ such that the edges in $S_u$ get weight $1$ by $\lambda_u$ and all other edges get weight $0$. 
Hence, we get $B(\lambda_u) = \bigcup S_u$. From this, 
we determine the bag $B_u \subseteq B(\lambda_u)$ via the 
{\em special condition\/} recalled in Definition~\ref{def:HD},
which distinguishes HDs from GHDs. More specifically, let $u'$ denote the parent of $u$ in the hypertree decomposition and let
$C$ denote the vertices in the edges that have to be covered by some node in the subtree rooted at $u$. Then we may set $B_u = B(\lambda_u) \cap (B_{u'} \cup C)$. 

In our case, when trying to construct a {\em fractional\/} hypertree decomposition of width $\leq k$ for a hypergraph with degree bounded by $d$, we know by Lemma~\ref{lem:bsupp} that 
we may restrict ourselves to edge-weight functions $\gamma_u$ with 
$|\supp(\gamma_u)| \leq  k\cdot d$. Moreover, we can be sure that 
$B(\gamma_u) \subseteq \bigcup S$ with $S = \supp(\gamma_u)$ holds. 
However, in contrast to the 
HD-setting studied in \cite{2002gottlob},
$B(\gamma_u) =  \bigcup S$ does in general {\em not\/} hold. 
Consequently, it is, of course, also unclear how to determine $B_u$.

\medskip
\noindent{\bf Subedge Functions.}
We will now provide a solution to both problems: how to determine $B(\gamma_u)$ and
how to determine $B_u$ for each node $u$ in an FHD? But before we do this, we define some useful notation for certain unions and intersections of families of sets.

\begin{definition}
\label{def:unions-and-intersections}
Let $S$ be a family of sets. We define the following further families of sets.
\begin{description}
\item[$\Cup S$] denotes the set-family which consists in all possible unions of an arbitrary number of sets 
from~$S$. 
(Note that $|\Cup S|\leq 2^{|S|}$). 
\item[$\Cup_i S$] for an integer $i\geq 1$, denotes the set-family which consists in
all possible unions of $\leq i$ sets from $S$. 
(Note that $|\Cup_i(S)|\leq  |S|^{i+1}$). 
\item[$\Cap S$] denotes the set-family which consists in all possible intersections of an arbitrary number of sets from $S$. 
(Note that $|\Cap S|\leq 2^{|S|}$).
\item[$\Cap_i S$] for an integer $i\geq 1$, denotes the set-family which consists in 
all possible intersections of $\leq i$ sets from $S$. 
(Note that $|\Cap_i S|\leq |S|^{i+1}$).
\end{description}

\noindent
If $S$ and $S'$ are both families of sets, then 
$S\capdot S'$ denotes the {\em pointwise intersection} between $S$ and $S'$, i.e., $S\capdot S'=\{A\cap B\,|\, A\in S$ and $B\in S'\}$.
\end{definition}

We now 
establish a bound on the number of possible sets $B(\gamma)$ that can arise in a hypergraph 
for all possible choices of a weight function $\gamma$.

\begin{definition}
\label{def:INSET}
Let $\INSET(H)$ denote the set of all possible sets $B(\gamma)$ such that $\gamma$ 
is an edge-weight function of $H$. For $S \subseteq E(H)$, we denote by $\INSET(S)$ 
the set of all possible sets $B(\gamma)$ where $\gamma(e) > 0$ if $e \in S$ and $\gamma(e) = 0$ if $e \not\in S$. 
That is, $S$ denotes the support of $\gamma$.
\end{definition}

\begin{definition}
\label{def:TYPES-CLASSES}
An {\em intersection type\/} of a hypergraph $H = (V(H),E(H))$
(or, simply a {\em type\/}, for short), is a subset of $E(H)$. 
For a hypergraph $H$, $\TYPES(H)= 2^{E(H)}$ consists of all possible types of $H$. 
For a type $t \in \TYPES(H)$, we define its class $\class(t)=\bigcap_{e\in t} e$ as the intersection of all edges in $t$. 
The set of all classes of $H$ is denoted by $\CLASSES(H)$. 

For a class $c \in \CLASSES(H)$ there may be more than one type $t$ with
$\class(t)= c$. However there is only one maximal type, namely
$\{e'\in E(H)\,|\,c\subseteq e'\}$; we denote by $\type(c)$ this unique maximal type.
\end{definition}

Note that $\TYPES(H)$ and $\CLASSES(H)$ depend only on $H$ and not on any edge-weight function. Moreover, every set $B(\gamma)$, for whatever edge-weight function,  must be equal to the union  of some classes of $H$. In fact, for any particular edge-weight function $\gamma$, the set
$B(\gamma)$ consists of the union of all sets $\class(t)$  for all types $t$  that satisfy $\gamma(t)\geq 1$ where 
$\gamma(t)=\Sigma_{e\in t}\gamma(e)$. 
Finally, the inequality $|\CLASSES(H)| \leq  |\TYPES(H)|$ clearly holds. 
We thus get the following lemma.

\begin{lemma}\label{lem:inset} 
Let $H$ be a hypergraph. Then the following properties hold: 
\begin{enumerate}
\item If $\gamma$ is an edge-weight function, then  $B(\gamma)\in\Cup\,{\CLASSES}(H)$.
\item $\INSET(H)\subseteq\Cup\,{\CLASSES}(H).$
\item 
$|\INSET(H)|\leq  
2^{|{\CLASSES}(H)|}\leq 2^{|\TYPES(H)|} \leq  2^{2^{|E(H)|}}$ and all three sets, $\INSET$, $\CLASSES$ and $\TYPES$, can be computed from $H$ in polynomial time if  the cardinality
of $E(H)$ is bounded by a constant.
\end{enumerate}
\end{lemma}

\noindent
The above inclusion $\INSET(H)\subseteq\Cup\,{\CLASSES}(H)$
only gives us an {\em exponential\/} upper bound $2^{2^{|E(H)|}}$ on the number of possible 
sets $B(\gamma_u)$ at any node $u$ in an FHD. 
However, by 
Lemma \ref{lem:bsupp}, we may assume w.l.o.g.\ that  
$|S_u| \leq k\cdot d$ with $S_u = \supp(\gamma_u)$
holds for every edge-weight function $\gamma_u$ of interest.
Hence, we only need to consider {\em polynomially\/} many values
for $\supp(\gamma_u)$. 
Moreover, for each $S_u$, 
there exist only {\em polynomially\/} many 
possible sets $B(\gamma_u)$
with $\supp(\gamma_u) = S_u$, i.e. $|\INSET(S_u)|\leq 2^{2^{|S_u|}}$. 
Hence, with Lemma \ref{lem:inset}, the first problem stated above is essentially solved.

It remains to find a solution to the second problem stated above, i.e., 
how to determine $B_u$ for each node $u$ in an FHD of width $k$?
We tackle this problem by again using the idea of 
{\em subedge functions} as described in Section \ref{sect:ghd} for deriving tractability results for the \rec{{GHD},\,$k$} problem. 
A subedge function takes as input a hypergraph 
$H = (V(H),E(H))$ and produces as output a set $E'$ of subedges of the edges in $E(H)$, 
such that $E'$ is then added to $E(H)$.
Clearly, adding a set $E'$ of subedges does not change the $\fhw$ of $H$.
Below, we shall define a whole family of subedge functions $h_{d,k}$, which, for fixed upper bounds $d$ on the degree and $k$ on the $\fhw$, take a hypergraph $H$ as input and 
return a polynomially bounded, polynomial-time computable set $E'$ of subedges of $E(H)$. Adding these subedges to $E(H)$ will then allow us to define a {\em polynomial\/} 
upper bound on the set of 
all possible bags $B_u$ at a given node $u$ in an FHD of $H$.

Towards this goal, we follow a similar approach as in the proof of Theorem \ref{theo:LogBMIP}. There we have used the LogBMIP to devise a polynomially bounded subedge function. Here, we will restrict ourselves to hypergraphs of bounded degree.

There are mainly two issues when trying to adapt the construction from the GHD case.
First, we carry over the notion of bag-maximality from GHDs to FHDs in the obvious way: 
we say that an FHD $\calF$ is {\em bag-maximal\/} 
if for each node $u$ of $\calF$, 
for every vertex $v\in B(\gamma_u) \setminus B_u$,
adding $v$ to $B_u$ would violate the connectedness condition. Clearly, for every FHD 
$\defFHDohne$,
a bag-maximal FHD $\calF^+=\left< T,(B_u^+),(\gamma_u) \right>$ 
can be generated by adding vertices from $B(\gamma_u) \setminus B_u$ 
to bags $B_u$ as long as possible. We may thus assume  w.l.o.g. that our FHD $\calF$ is bag-maximal.

For the definition of an appropriate subedge function
(denoted $h_{d,k}$ to 
indicate that it depends on $d$ and $k$),
take 
a hypergraph $H$ with degree bounded by $d \geq 1$ and consider an arbitrary 
FHD $\calF$ of $H$ of width $\leq k$.
Let $u$ be an arbitrary node in $\calF$
with edge-weight function $\gamma_u$ and let $e \in \supp(\gamma_u)$ with 
$e \cap B(\gamma_u) \not\subseteq B_u$. 
As in Section \ref{sect:ghd},
our goal is to define $h_{d,k}$ in such a way that 
$e' = e \cap B_u$ is contained in $h_{d,k}$ 
for every edge $e \in \cov(\gamma_u)$ and every possible bag $B_u$ in $\mcF$.
As a first step towards this goal, 
we extend the notion of critical paths from the GHD setting to FHDs.

\begin{definition}
\label{def:critical_fhd}
Let $\calF=\defFHD$  
be an FHD
of a hypergraph $\HH$. Moreover, let 
$u$ be a node in $\calF$ and 
let $e \in \supp(\gamma_u)$ with 
$e \cap B(\gamma_u) \not\subseteq B_u$. 
Let $u^*$ denote the node closest to $u$, such that $u^*$ covers $e$, i.e., 
$e \subseteq B_{u^*}$. Then, analogously to Definition \ref{def:critical_ghd},
we call the path $\pi = (u_0,u_1,\ldots,u_l)$ with $u_0 = u$ and $u_l = u^*$ the
{\em critical path\/} of $(u,e)$ denoted as $\critp(u,e)$.
\end{definition}
The following lemma
allows us to characterize the subsets $e'$ of $e$ needed in the subedge function
$h_{d,k}$. The proof of Lemma \ref{lem:subedge-characterization} can be literally translated from the proof of Lemma \ref{lem:crit}.

\begin{lemma}
\label{lem:subedge-characterization}
Let $\calF=\defFHD$  
be a bag-maximal FHD
of a hypergraph $\HH=(V(H),$ \linebreak $E(H))$, let $u\in T$, $e \in \supp(\gamma_u)$, and $e \cap B(\gamma_u) \not\subseteq B_u$. Let $\pi = (u_0,u_1,\ldots,u_l)$ with $u_0 = u$ be the critical path of $(u,e)$. Then the following equality holds: 

\[e\cap B_u=e\cap \bigcap^l_{i=1} B(\gamma_{u_i})\]

\end{lemma}

We want to define the subedge function $h_{d,k}$ such that all subedges appearing on the right-hand side of above equality. To achieve this while abstracting from the knowledge of a particular decomposition and from the knowledge of particular edge-weight functions, we will make two bold over-approximations. First, instead of considering concrete critical paths, we will consider arbitrary finite sequences $\xi= (\xi_1,\ldots,\xi_{\max(\xi)})$ of groups of $\leq k\cdot d$ edges of $\HH$, where each such group represents a potential support $\supp(\gamma_u)$ at some potential node $u$ of a potential FHD of $\HH$. Clearly, each effective path  $\critp(u,e)$ for any possible combination of a decomposition node $u$ and an
edge $e$ of any possible FHD $\calF$ of $\HH$ is among these sequences.
The second over-approximation we make is that instead of considering particular edge-weight functions, we will simply 
consider (a superset of) all possible supports of $\leq k\cdot d$ atoms, and for each such support $S$ (a superset of) all possible sets $B(\gamma)$, i.e. $\INSET(S)$, that could possibly arise with this support. A support is simply given by a subset of $\leq k\cdot d$ edges of $\HH$. For each such support, by Lemma~\ref{lem:inset}, there are in fact no more than $2^{2^{k\cdot d}}$ $B(\gamma)$-sets and these are determined by unions of classes from $\CLASSES(H')$, where  $H'$ is the 
subhypergraph of $\HH$ given by the support. To make this more formal, 
we give the following definition.
Recall the notion of $\CLASSES$ (which denotes the set of intersections of edges contained in some type; in our case, each type consists of at most $k\cdot d$ edges)
from Definition~\ref{def:TYPES-CLASSES}.

\begin{definition}
\label{def:CLASSES}
Let $H$ be a hypergraph and let 
$\xi=(\xi_1,\ldots,\xi_{\max(\xi)})$ be an arbitrary sequence of groups of 
$\leq k\cdot d$ edges of $\HH$.
For $i \in \{1, \dots, \max(\xi)\}$, by slight abuse of notation,
we overload the notion of $\CLASSES$ from 
Definition~\ref{def:TYPES-CLASSES} as follows: 
we write $\CLASSES(\xi_i)$ to denote the set $\CLASSES(\HH_\xi^i)$, where  $\HH_\xi^i$ is the subhypergraph of $\HH$ whose edges are the $\leq k\cdot d$  
edges of $\xi_i$ and whose vertices are precisely all vertices occurring in these edges.

Let $\pi$ be a critical path of the form $\pi = (u_0, u_1, \ldots ,u_l)$
of some FHD $\calF=\defFHD$ of $\HH$.
Suppose that each edge-weight function $\gamma_{u}$ in $\calF$ has $(k\cdot d)$-bounded support. 
Then we denote by $\xi^{\pi}$ the sequence  $\xi^{\pi} = (\xi_1,\ldots,\xi_l)$ with 
$\xi_i= \supp(\gamma_{u_i})$
for $1\leq i\leq l$.
\end{definition}

Our goal is to compute a set of subedges of the edges in $E(H)$ containing all 
sets of the form $e\cap \bigcap^l_{i=1}B(\gamma_{u_i})$ with $\pi = \critp(u,e) = (u_0,u_1,\ldots,u_l)$.
We may use that each $\gamma_{u_i}$ has $(k\cdot d)$-bounded support.
By Lemma \ref{lem:inset}, we know that every possible $B(\gamma_{u_i})$-set 
is contained in $\Cup\,{\CLASSES}(H)$, i.e., 
every possible $B(\gamma_{u_i})$ along a critical path $\pi$ can be represented 
as the union of classes 
(where each class is in turn
the intersection of some edges selected from 
$\supp(\gamma_{u_i})$). Hence, to obtain $e\cap \bigcap^l_{i=1}B(\gamma_{u_i})$,
we need to compute the intersection of all unions of classes  along a critical path 
$\pi$. 

\begin{algorithm}[t]
\SetKwData{Left}{left}\SetKwData{This}{this}\SetKwData{Up}{up}
\SetKwFunction{Union}{Union}\SetKwFunction{FindCompress}{FindCompress}

\SetKwData{N}{N}

\SetKwInOut{Input}{input}\SetKwInOut{Output}{output}

\Input{A sequence $\xi = (\xi_1,\ldots,\xi_{\max(\xi)})$ of groups of $\leq k \cdot d$ edges of a hypergraph $H$.}
\Output{$\IF(\xi)$}
\BlankLine
\tcc{Initialization: create a tree $T_c$ for each $c \in \CLASSES(\xi_1)$}
$\IF(\xi) \la \emptyset$\;
\ForEach{$c \in \CLASSES(\xi_1)$}{
$N_c \la \{v\}$; $E_c \la \emptyset$; $T_c \la (N_c,E_c)$\;
\BlankLine
$\sett(v)\la c$\;
$\levels(v) \la \{1\}$\; 
$\edges(v) \la \{e\in E(\HH)\, |\, c\subseteq e\}$\;
$\markk(v) \la \ok$\;
\BlankLine
$\IF(\xi) \la \IF(\xi) \cup T_c$\;
}
\BlankLine
\tcc{Expand and update the trees as follows}
\For{$i\leftarrow 2$ \KwTo $\max(\xi)$}{
\ForEach{leaf node $v$ of a tree $T_v=(N_v,E_v) \in \IF(\xi)$ with $\max(\levels(v)) = i-1$ and $\markk(v) = \ok$}{
   \ForEach{$c \in \CLASSES(\xi_i)$}{
      \Switch{$\sett(v) \cap c$}{
         \lCase(\tcc*[f]{Dead End}){$\sett(v) \cap c = \emptyset$}{
             $\markk(v) \la \fail$
         }%
         \lCase(\tcc*[f]{Passing}){$\sett(v) \cap c = \sett(v)$}{
             $\levels(v) = \levels(v) \cup \{ i\}$
         }
         \Case(\tcc*[f]{Expand}){$\sett(v) \cap c \subsetneq \sett(v)$}{
             Create a new node $v'$\;
             $N_v \la N_v \cup \{v'\}$; $E_v \la E_c \cup \{v,v'\}$; $T_v \la (N_v,E_v)$\;
             \BlankLine             
             $\sett(v')\la \sett(v) \cap c$\;
             $\levels(v') \la \{i\}$\; 
             $\edges(v') \la \{e\in E(\HH) \mid \sett(v')\subseteq e\}$\;
             $\markk(v') \la \ok$\;
         }
      }
   }
}
}
\caption{\iforest}\label{alg:uiforest}
\end{algorithm}

Recall that we generalize the support of edge-weight functions $\gamma_{u_i}$ along
a concrete critical path $\pi$ in a concrete FHD $\calF$ of $H$ to 
sequences 
$\xi =(\xi_1,\ldots,\xi_{\max(\xi)})$, where each $\xi_i$ is an arbitrary set of $\leq k\cdot d$
edges from $H$. As the crucial data structure to compute the desired intersections
of unions of classes, we now define the {\em intersection forest} $\IF(\xi)$.
This data structure will give us a systematic way to convert the intersections of unions of classes for all possible sequences $\xi$ into a union of intersections. 
Intuitively, each branch (starting at a root) in $\IF(\xi)$ represents a 
possible transversal of the family $\{\CLASSES(\xi_i)\}_{1\leq i\leq \max(\xi)}$
for some sequence $\xi = (\xi_1,\ldots,\xi_{\max(\xi)})$, i.e., 
a transversal selects one class from $\CLASSES(\xi_i)$ for each $i \in \{1, \dots, \max(\xi)\}$.
On every branch, we will then compute the intersection of the classes selected along this branch. Since each class is in turn an intersection of edges (namely the edges contained in some type), every branch in 
$\IF(\xi)$ therefore simply yields an intersection of edges from $H$. Similar to the $\uitree$ algorithm presented in Section \ref{sect:ghd}, we formalize the construction of $\IF(\xi)$ in the Algorithm \ref{alg:uiforest}
``\iforest''. We define $\IF(\xi)$ as a rooted forest such that each of its nodes $v$ is labelled by 
\begin{itemize}
\item a set of vertices $\sett(v)\subseteq V(\HH)$, 
\item a set of levels $\levels(v) \subseteq \{1, \dots, \max(\xi)\}$,  
\item a set of edges  $\edges(v)$  such that $\edges(v) = \{e \in E(H) \mid \sett(v)\subseteq e\}$;\\
in other words,  $\sett(v)$ is a {\em class\/} and $\edges(v)$ is its (maximal) {\em type\/}, 
see Definition~\ref{def:TYPES-CLASSES},
\item and a mark $\markk(v)\in\{\ok, \fail\, \}$.  
\end{itemize}

The expansion of the trees in Algorithm \ref{alg:uiforest} can be seen as follows:
\begin{enumerate}
\item {\em Dead End.}\ If for each  class $c$ of $\CLASSES(\xi_{i})$, $\sett(v)\cap c=\emptyset$, then $v$ has no children, and its mark is set to $\markk(v)=\fail$. Intuitively this is a dead end as it cannot be continued to yield a non-empty intersection of a transversal of the family $\{\CLASSES(\xi_j)\}_{1\leq j\leq \max(\xi)}$.
\item{\em Passing.} For each  class $c$ of $\CLASSES(\xi_{i})$ fulfilling $\sett(v)\cap c=\sett(v)$, insert $i+1$ into $\levels(v)$.
Intuitively, this makes sure the same value $\sett(v)$ is never repeated on a branch, and, as a consequence, every child node must have a strictly smaller $\sett()$-component and, thus,
at least one more edge in its $edges()$-label than its parent (see also Fact~1 in Lemma \ref{lem:facts-IF} below).
\item {\em Expand}. For each class $c$ of $\CLASSES(\xi_{i+1})$ fulfilling $\sett(v)\cap c\subsetneq \sett(v)$, 
create a child $v'$ of $v$, and let $\sett(v')=\sett(v)\cap c$,
$\levels(v')=\{i+1\}$, 
$\edges(v')=\{ e\in E(\HH)\, |\, \sett(v') \subseteq e\}$,
and $\markk(v)= \ok$. Note that we thus clearly have $\sett(v') \subsetneq \sett(v)$
and $\edges(v) \subsetneq \edges(v')$.
\end{enumerate}

\smallskip
\noindent
We now define some useful notation for talking about the intersection forest $\IF(\xi)$.
\begin{definition}
\label{def:iforest}
Let $H$ be a hypergraph and let 
$\xi=(\xi_1,\ldots,\xi_{\max(\xi)})$ be an arbitrary sequence of groups of 
$\leq k\cdot d$ edges of $\HH$.
For $1 \leq i\leq \max(\xi)$, let $\iflevel_i(\xi)$ denote the set of all nodes $v$ of $\IF(\xi)$ such that $i \in \levels(v)$ and $\markk(v)=\ok$. 
We denote by $\FRINGE_i(\xi)$ the collection of all sets $\sett(v)$ where $v\in \iflevel_i(\xi)$. Finally, let the {\em fringe of $\xi$} be defined as $\FRINGE(\xi)=\FRINGE_{\max(\xi)}(\xi)$.
\end{definition}
Note that, by definition of  $\IF(\xi)$, there cannot be any {\em fail} node on level $\max(\xi)$.
Hence, $\FRINGE(\xi)$ coincides with the set of all $\sett$-labels at 
level $\max(\xi)$, i.e. the leaf nodes $v$ of all trees in $\IF(\xi)$ such that $\markk(v)= \ok$.
We now establish some easy facts about $\IF(\xi)$.

\begin{lemma}\label{lem:facts-IF} 
Let $H$ be a hypergraph with degree $d$ and let 
$\xi=\xi_1,\ldots,\xi_{\max(\xi)}$ be an arbitrary sequence of groups of 
$\leq k\cdot d$ edges of $\HH$.
Then the intersection forest $\IF(\xi)$ according to the construction in Algorithm \ref{alg:uiforest}
has the following properties:

\begin{description}
\item[Fact~1.] If node $v'$ is a child of node $v$ in $\IF(\xi)$, then $\edges(v')$ must contain at least one new edge in addition to the edges already present in $\edges(v)$.

\item[Fact~2.]The depth of $\IF(\xi)$ is at most $d-1$.

\item[Fact~3.]{\it 
Let $a=2^{k\cdot d}$. 
Then 
$\IF(\xi)$ has no more than  $a^{d+1}$ nodes and $|\FRINGE(\xi)|\leq a^d=2^{d^2\cdot k}$.}\ 
\end{description}
\end{lemma}

\begin{proof} The facts stated above can be seen as follows.

\begin{description}
\item[{\it Fact~1.}] 
Note that we can only create a child node through an {\em Expand} operation. 
This requires that $\sett(v') = \sett(v)\cap c \subsetneq \sett(v)$
for some class $c \in \CLASSES(\xi_i)$. Recall that $c$ is the intersection of all edges of $\type (c)$.
Hence, for $\sett(v')$ to shrink, $\type (c)$ must contain at least one new edge not yet contained in $\edges(v)$, and this edge is therefore included into  $\edges(v')$.
\item[{\it Fact~2.}] 
This follows from Fact~1 and the fact that, if $\HH$ is of degree $d$,
then any intersection of $d+1$  or more edges is empty. 
\item[{\it Fact~3.}] 
For each sequence $\xi$ as above, 
for whatever $\xi_i$, the inequality $|\CLASSES(\xi_i)|\leq a$
holds by Lemma~\ref{lem:inset}.
Hence, Fact~3 follows from the depth $d-1$ established in Fact~2 and the 
fact that we have at most $a$ such trees, each with branching not larger than $a$. 
\end{description}
This concludes the proof of the lemma.
\end{proof}

Recall that we are studying the set of arbitrary sequences 
$\xi=(\xi_1,\ldots,\xi_{\max(\xi)})$ of groups of $\leq k\cdot d$ edges of $\HH$
because they give us a superset of possible critical paths  $\pi = \critp(u,e)$ in possible FHDs of $\HH$, such that
each group $\xi_i$ of edges corresponds to the support of the edge-weight function 
$\gamma_{u_i}$ at the $i$-th node $u_i$ on path $\pi$.
The following lemma establishes that the intersection forests 
(and, in particular, the notion of $\FRINGE(\xi)$)
introduced above indeed
give us a tool to generate a superset of the set of all possible sets
$\bigcap^l_{i=1} B(\gamma_{u_i})$ 
in  all possible 
FHDs of $\HH$ of width $\leq k$.

\begin{lemma}
\label{lem:bags-via-fringe} 
Let $\HH$ be a hypergraph of degree $d \geq 1$ and let 
$\calF$ be an FHD of $H$ of width $\leq k$.
Consider a  critical path $\pi = (u_0, u_1, \dots, u_l)$
of FHD $\calF$, 
together with its associated sequence $\xi^{\pi}$
introduced in  Definition~ \ref{def:CLASSES}. 
We claim that the following relationship holds:
\[\bigcap^l_{i=1} B(\gamma_{u_i})\ \in\ \Cup \FRINGE(\xi^{\pi})\]
\end{lemma}

\begin{proof}
We show by induction on $i$ that, for  every $i \in \{1, \dots, l\}$, 
the following relationship holds:
\[\bigcap^i_{j=1} B(\gamma_{u_j}) \in \Cup 
\FRINGE_i(\xi^{\pi}) \]

\smallskip

\noindent
{\em Base Step.} 
For the base case $i=1$, recall that in the \iforest algorithm, the set  $\FRINGE_1(\xi^{\pi})$ is initialized 
to $\CLASSES(\xi_1)$ with $\xi_1 = \cov(\gamma_{u_1})$. 
Moreover, by Lemma~\ref{lem:inset}(2), 
$\INSET(H)\subseteq\Cup\,{\CLASSES}(H)$ and, hence, we have
$\INSET(\supp(\gamma_{u_1})) \subseteq\Cup\,{\CLASSES}(\supp(\gamma_{u_1}))$.
In total, $\INSET(\supp(\gamma_{u_1})) \subseteq\Cup\FRINGE_1(\xi^{\pi})$ indeed holds.

\smallskip

\noindent
{\em Inductive Step.} 
Assume for some $i < l$ that 
$\bigcap^i_{j=1} B(\gamma_{u_j}) \in \Cup \FRINGE_i(\xi^{\pi})$
holds. We show that the desired relationship also holds for $i+1$. Clearly,  
$\bigcap^{i+1}_{j=1} B(\gamma_{u_j}) = \bigcap^i_{j=1} B(\gamma_{u_j}) \cap B(\gamma_{u_{i+1}})$. 
From this, by the induction hypothesis, together with 
$B(\gamma_{u_{i+1}})\in \INSET(\xi^{\pi}_{i+1})$, and the inclusion $\INSET(\xi^{\pi}_{i+1})\subseteq 
\Cup \CLASSES(\xi^{\pi}_{i+1})$, which holds by Lemma~\ref{lem:inset}, we obtain:
$$\bigcap^{i+1}_{j=1} B(\gamma_{u_j}) \in  \FRINGE_i(\xi^{\pi})\capdot \Cup \CLASSES(\xi^{\pi}_{i+1}).$$
By using the distributivity of $\capdot$ over $\Cup$, we get: 
$$\FRINGE_i(\xi^{\pi})\capdot \Cup \CLASSES(\xi^{\pi}_{i+1})
= 
\Cup (\FRINGE_i(\xi^{\pi})\capdot  \CLASSES(\xi^{\pi}_{i+1})).$$
Moreover, by the construction of $\IF(\xi^{\pi})$, 
we have 
$$\FRINGE_i(\xi^{\pi})\capdot  \CLASSES(\xi^{\pi}_{i+1}) =
\FRINGE_{i+1}(\xi^{\pi}).$$ 
In fact, the {\em Passing} and {\em Expand} cases make precisely these intersections when producing level $i+1$ of $\IF(\xi^{\pi})$. Therefore, in total, we obtain that
$$\bigcap^{i+1}_{j=1} B(\gamma_{u_j}) \in \Cup  \FRINGE_{i+1}(\xi^{\pi})
$$
indeed holds, which settles the inductive step.
\end{proof}

\noindent
The desired subedge function $h_{d,k}$ therefore looks as follows: 

\begin{lemma}
\label{lem:desired-subedge-function} 
Let $\HH$ be a hypergraph of degree $d \geq 1$ and let $k \geq 1$.
Let the subedge function $h_{d,k}$ be defined as 
\[h_{d,k}(\HH)= E(\HH)\capdot (  \Cup_{2^{d^2\cdot k}}\,\Cap_d E(\HH) )\]
Then (for fixed constants $d$ and $k$), 
the size of $h_{d,k}(\HH)$ is polynomially bounded and 
$h_{d,k}(\HH)$ can be computed in polynomial time. 
Moreover $h_{d,k}(\HH)$
contains all subedges $e\cap B_u$ of all $e \in E(H)$ for all possible bags $B_u$ of whatever bag-maximal FHD of width $\leq k$ of $\HH$.
\end{lemma}

\begin{proof}
For any sequence $\xi$, each element of $\FRINGE(\xi)$ is  the intersection of at most $d$ edges.
Moreover, by Fact~3 of Lemma~\ref{lem:facts-IF},
$|\FRINGE(\xi)|\leq a^d=2^{d^2\cdot k}$ holds. 
Therefore, for all possible sequences $\xi$, we have
$$\Cup \FRINGE(\xi)\subseteq \Cup_{2^{d^2\cdot k}}\,\Cap_d E(\HH).$$
Given that $d$ and $k$ are constants, the set $\Cup_{2^{d^2\cdot k}}\,\Cap_d E(\HH)$ is of polynomial size and is clearly computable in polynomial time from $\HH$. 
By Lemma \ref{lem:subedge-characterization} together with 
Lemma \ref{lem:bags-via-fringe}, 
it is then also clear that the subedge function 
$h_{d,k}(\HH)$ 
contains all subedges $e\cap B_u$ for all possible bags of whatever bag-maximal FHD of 
width $\leq k$ of $\HH$. 
\end{proof}

\medskip\noindent{\bf Deciding the {\sc Check} Problem for Hypergraphs of Bounded Degree.}
With the subedge function $h_{d,k}$ at hand, we have a powerful tool that will allow us to 
devise a polynomial-time decision procedure for the \rec{{FHD},\,$k$} problem. Towards this
goal, we first observe that adding the edges in $h_{d,k}(H)$ to a hypergraph $H$ 
allows us to restrict ourselves to FHDs of a very peculiar form. This form is captured by the following 
definition.

\begin{definition}
An FHD $\calF=\defFHD$ of a hypergraph $\HH$ is {\em strict\/} 
if for every decomposition node $u$ in $T$, the equality 
$B_u=B(\gamma_u)=\bigcup \supp(\gamma_u)$ holds. 
\end{definition}

Below we show that, in case of bounded degree,
we can transform every FHD of width $\leq k$ into a strict FHD of width $\leq k$.

\begin{lemma}\label{lem:transform-into-strict}
Let $H = (V(H),E(H))$ be a hypergraph of degree $\leq d$ and let 
$\calF=\defFHD$ be an FHD of $H$ of width $\leq k$.  
Suppose that $H'$ is obtained from $H$ by adding the edges in $h_{d,k}(\HH)$, i.e., 
$H' = (V(H'),E(H'))$ with 
$V(H') = V(H)$ and 
$E(H') = E(H) \cup h_{d,k}(\HH)$.
Then $H'$ admits a {\em strict\/} FHD $\calF'= \left<T,(B_u)_{u \in T},(\gamma'_u)_{u \in T}\right>$ 
of width $\leq k$.%
\end{lemma}

\begin{proof}
Let $\calF=\defFHD$ be an arbitrary FHD of $H$ of width $\leq k$.  
W.l.o.g., assume that $\calF$ is bag-maximal.
Of course, $\calF$ is also an FHD of $H'$ of width $\leq k$.  
Let $u$ be a node in $\calF$ and let $e \in \supp(\gamma_u)$. 
Suppose that $e \cap B(\gamma_u) \not\subseteq B_u$. We modify $\gamma_u$ 
as follows: by Lemma~\ref{lem:desired-subedge-function}, 
$E(H) \cup h_{d,k}(\HH)$ is guaranteed to contain the subedge 
$e'=e\cap B_u$  of $e$. Then we ``replace'' $e$ in $\gamma_u$ by $e'$, i.e.,
we set $\gamma_u(e') := \gamma_u(e') + \gamma_u(e)$ and $\gamma_u(e) := 0$.
The FHD $\calF' = \left<T,(B_u)_{u \in T},(\gamma'_u)_{u \in T}\right>$ is obtained 
by exhaustive application of this transformation step. Clearly, such a transformation step never 
increases the support. Moreover, the resulting FHD $\calF'$ is strict. 
\end{proof}

Our strategy to devise a polynomial-time decision procedure 
for the \rec{{FHD},\,$k$} problem is to reduce it to the \rec{{HD},\,$k$} problem and then adapt the
algorithm from \cite{2002gottlob}. Note however, that the algorithm from \cite{2002gottlob} requires the HDs to 
be in a certain {\em normal form\/}. We thus have to make sure that also in the FHD-setting, we 
can always achieve an analogous normal form. Below we define the {\em fractional normal form\/} 
(FNF):

\newcommand{\deffhdnf}{
An FHD $\mcF = \defFHD$ of a hypergraph $H$ is in 
{\em fractional normal form (FNF)\/}
if for
each node $r \in T$, and for each child $s$ of $r$, the following
conditions hold:
 \begin{enumerate}
  \item there is {\em exactly one\/} \comp{$B_r$} $C_r$ such that 
  $\VTs = C_r \cup (B_r \cap B_s)$ holds; 
  \item $B_s \cap C_r \neq \emptyset$, where $C_r$ is the \comp{$B_r$} satisfying
        Condition~1;
  \item $B(\gamma_s) \cap B_r \subseteq B_s$.
  \end{enumerate}
}

\begin{definition}
\label{def:fhdnf}
\deffhdnf
\end{definition}

\smallskip

An HD  $\calH = \defHD$ can be considered as a special case of an 
FHD where the edge-weight functions $\lambda_u$ only assign weights 0 or 1 to each edge
and where the so-called {\em special condition\/} holds.
When applied to HDs, the fractional normal form recalled above coincides with the normal
form defined in \cite{2002gottlob}. Indeed, the transformation of an arbitrary FHD into FNF given in Appendix \ref{app:FHD-normal-form} follows closely the transformation of HDs into
the normal form given in \cite{2002gottlob}.

We now strengthen Lemma~\ref{lem:transform-into-strict} such that also 
FNF and bounded support can be guaranteed.

\begin{lemma}
\label{lem:transform-into-strict-FNF}
Let  $H = (V(H),E(H))$ be a hypergraph of degree $\leq d$ and let
$\calF=\defFHD$ be an FHD of $H$ of width $\leq k$.   
Suppose that $H'$ is obtained from $H$ by adding the edges in $h_{d,k}(\HH)$, i.e., 
$H' = (V(H'),E(H'))$ with 
$V(H') = V(H)$ and 
$E(H') = E(H) \cup h_{d,k}(\HH)$.
Then $H'$admits a {\em strict\/} FHD $\calF'= \left<T,(B_u)_{u \in T},(\gamma'_u)_{u \in T}\right>$  
{\em in fractional normal form\/}
of width $\leq k$ that has $(k\cdot d)$-bounded support. 
\end{lemma}

\begin{proof}
By Lemma~\ref{lem:bsupp} there exists an FHD $\calF_1$ of $\HH$ whose supports are all bounded by 
$k\cdot d$. Without changing the supports, we can transform this FHD into a bag-maximal one, and we thus assume w.l.o.g. that $\calF_1$ is bag-maximal.

Now transform $\calF_1$ into an FHD $\calF_2$ of width $k\cdot d$  {\em in FNF\/},
by proceeding  according to the proof of Theorem \ref{thm:nf} in Appendix \ref{app:FHD-normal-form},
which, in turn follows closely the transformation in the proof of  Theorem 5.4 in~\cite{2002gottlob}.
Note that this transformation preserves the support bound of $k\cdot d$. In fact, 
the component split made for ensuring condition 1 of FNF can never lead to a larger support, given that the 
bags never increase. Ensuring condition 2 results in eliminating nodes from the tree, so nothing bad can happen. 
Observe that condition 3, which is $B(\gamma_s)\cap B_r\subseteq B_s$ for a child node $s$ of decomposition node $r$,
is initially satisfied, because the initial FHD $\calF_1$ is bag-maximal. Observe further that the splitting of a node (subtree) into several nodes (subtrees) performed to 
achieve condition~1 of FNF does not destroy the validity of condition~3.

Finally transform the FHD $\calF_2$ via Lemma~\ref{lem:transform-into-strict} into a 
strict FHD $\calF' =\left<(T,(B_u)_{u\in T},(\gamma'_u)_{u \in T}\right>$ of 
$\HH'$ of width $\leq k$ and with  $(k\cdot d)$-bounded support. 
This strict FHD $\calF'$ is still in FNF. To see this, first note that the  tree structure $T$ of the decomposition and all bags $B_u$ remain exactly the same. Moreover, for whatever set $S \subseteq V(H)$,
$\HH$ and $\HH'$ have exactly the same $[S]$-components. This can be seen 
by recalling from \cite{2002gottlob} 
that two vertices $v_1,v_2$ in a hypergraph $H$ are $[S]$-adjacent if they are
adjacent in the subhypergraph of $H$ induced by $V(H) \setminus S$. 
Hence, $[S]$-adjacency remains unaltered when adding subedges. 

Given that conditions 1 and 2 of FNF are only formulated in terms of $B_i$-bags  
and $[B_i]$-components 
-- and all such bags and components are the same for $\calF$ and $\calF'$ -- 
they remain valid. Condition~3, which requires that $B(\gamma'_s)\cap B_r\subseteq B_s$ for
child node $s$ of $r$, is now trivially satisfied, because $\calF'$ is {\em strict} and,
therefore, even $B(\gamma'_s)=B_s$ holds.
\end{proof}

The following theorem finally establishes the close connection between 
the \rec{{FHD},\,$k$} and \rec{{HD},\,$k$} problems for hypergraphs $H$ of degree
bounded by some constant $d \geq 1$.  
Recall that the edge-weight functions $\lambda_u$ in an HD only assign values 0 or 1 to 
edges. As in \cite{2002gottlob}, it is convenient to identify $\lambda_u$
with a set $S_u$ of edges, namely the edges in $E(H)$ that are assigned value 1. In other words, $S_u = \supp(\lambda_u)$. Moreover, we can identify a set of edges $S_u$ with the 
hypergraph whose set of vertices is $\bigcup S_u$ and whose set of edges is $S_u$. 
For given edge-weight function $\lambda_u$ with $S_u = \supp(\lambda_u)$, 
we shall write $H_{\lambda_u}$ to denote this hypergraph.

\begin{theorem}
\label{theo:FHD-vs-HD}
Let $\HH$ be a hypergraph whose degree is bounded by $d \geq 1$
and define 
$H'$ as above, i.e., 
$H' = (V(H'),E(H'))$ with 
$E(H') = E(H) \cup h_{d,k}(\HH)$.
Then the  following statements are equivalent:
\begin{enumerate}
\item $\fhw(\HH) \leq k$.
\item $\HH'$ admits a strict hypertree  decomposition 
(thus a query decomposition) 
$\calH = \left<T,(B_u)_{u\in T},\right.$
$\left.(\lambda_u)_{u\in T}\right>$ of width  $\leq k\cdot d$ in normal form
such that for each decomposition node $u$ of $\calH$, $\rho^*(H_{\lambda_u})\leq k$ holds
(i.e., $H_{\lambda_u}$ has a fractional edge cover of weight $\leq k$). 
\end{enumerate}
\end{theorem}

\begin{proof}
To prove $1\Rightarrow 2$, we use Lemma~\ref{lem:transform-into-strict-FNF}
to conclude from $\fhw(H) \leq k$ that there exists a strict FHD 
$\calF= \left<T,(B_u)_{u\in T},(\gamma_u)_{u\in T}\right>$  of $H$ in  
{\em fractional normal form\/}
of width $\leq k$ that has $(k\cdot d)$-bounded support. 
The FHD $\calF$ can be naturally transformed into a GHD 
$\calH = \left< T, (B_u)_{u\in T}, (\lambda_u)_{u\in T} \right>$ by leaving the tree structure and the bags $B_u$ 
unchanged and by defining $\lambda_u$ as the characteristic function of $\supp(\gamma_u)$, i.e., 
$\lambda_u(e) = 1$ if $e \in \supp(\gamma_u)$ and $\lambda_u(e) = 0$ otherwise.
Since $\calF$  is strict, it follows immediately that the resulting GHD satisfies the special condition, i.e., 
$\calH$ is in fact an HD. Moreover, since $\mcF$ is in (fractional) normal form, 
also the HD $\calH$ is in the normal form from \cite{2002gottlob}. 

To see $2\Rightarrow 1$, assume $2$ holds with query decomposition 
$\calH=\defHD$ and assume further that, 
for each decomposition node $u$ of $\calH$, there exists a fractional edge cover  
$\gamma'_u$ for $H_{\lambda_u}$  of width $\leq k$. 
In particular, we thus have  $B(\lambda_u)=B(\gamma'_u)= B_u$. 
Similarly to the proof of Lemma~\ref{lem:bsupp}, we can transform each fractional edge cover
$\gamma'_u$ of the induced subhypergraph $H_{\lambda_u}$ of $H'$ into an 
edge-weight function $\gamma_u$ of $H$ by moving the weights $\gamma'_u(e')$ 
of each edge $e'$ in 
$H_{\lambda_u}$ to one of its ``originator edges'' $e$ in $H$ (i.e., an edge $e$ in $H$
with $e' \subseteq e$).
By replacing $\lambda_u$ with  $\gamma_u$, we obtain an FHD
$\calF=\defFHD$ of width $\leq k$ of $H$.
\end{proof}

\nop{*************************
Note that the HDs with bounded supports used here seem to be of the type of {\em fractionally improved  HDs} Davide Longo is currently implementing; it may be worthwhile formulating the results in a more general context for such HDs).
*************************}

\noindent
We are now ready to prove the main result of this section:

\begin{proof}[Proof of Theorem \ref{theo:exactFHD}]
By Theorem~\ref{theo:FHD-vs-HD}, it is sufficient to look for a strict hypertree  decomposition (thus a query decomposition) $\calH= \defHD$ 
of $\HH'$ of width  $\leq d\cdot k$ such that for each decomposition node $u$ of $\calH$, $\rho^*(H'_{\lambda_u})\leq k$ holds. 
This is achieved by modifying the alternating algorithm $k$-decomp 
from~\cite{2002gottlob} by 
inserting the following two checks at each node $u$:
\begin{itemize}
\item if $u$ has a parent $r$, then 
$\bigcup S_u \subseteq B(\lambda_r)\cup \treecomp(u)$ where
$S_u = \supp(\lambda_u)$
and $\treecomp(u)$ is defined as the set of vertices which have to appear in $T_u$ and which do not appear outside $T_u$.
This makes sure that we may set $B_u= \bigcup S_u$ without violating the connectedness condition. Hence, 
the resulting decomposition is strict.
\item $\rho^*(H'_{\lambda_u})\leq k$.
\end{itemize}
The so modified algorithm clearly runs in 
\alogspace = \ptime
\end{proof}

\newcommand{\rarity}[1]{\ensuremath{\mathit{rank}(#1)}}

\medskip\noindent{\bf Deciding the {\sc Check} Problem for Hypergraphs of Bounded Rank.}
We conclude this section by presenting another class of hypergraphs for which 
the \rec{FHD,$k$} problem is tractable. More specifically, we study the class of hypergraphs with bounded rank, which properly contains the class of graphs, whose
rank is bounded by 2. 
We first formally define
the bounded rank property:

\begin{definition}\label{def:brank}
The {\em rank} of a hypergraph $\HH$ (denoted $\rarity{\HH}$)
is the 
maximum cardinality of any edge  $e$ of $\HH$.
We say that a hypergraph $H$ has
the {\em $r$-bounded rank property ($r$-BRP)} if 
$\rarity{\HH}\leq r$ holds. 
For a class $\classC$ of hypergraphs,
we say that $\classC$ 
has the  
\emph{bounded rank property (BRP)} 
if there exists a constant $r$
such that
every hypergraph $H$ in $\classC$
has the $r$-BRP.
%
\end{definition}

For every hypergraph with $\rarity{\HH}\leq r$ for some constant $r$, the following lemma is immediate:

\begin{lemma}
\label{lem:brp}
Let $H$ be a hypergraph whose arity is bounded by some constant $r$. Then for every (fractional) edge weight function $\gamma$ for $H$ satisfying $\weight(\gamma) \leq k$, the property 
$|B(\gamma)| \leq r \cdot k$ holds.  
\end{lemma}

We can now use the above lemma to devise a polynomial-time decision procedure 
for the \rec{{FHD},\,$k$} problem by a straightforward adaptation of the 
$k$-decomp  algorithm from \cite{2002gottlob}
for the \rec{{HD},\,$k$} problem.
The key idea in $k$-decomp is, in a top-down construction of the HD, to guess for the 
current component $C$ (initially, $C = V(H)$)
an edge cover $\lambda_u$ of the next node $u$ and 
to call the decomposition procedure recursively for all
$[B(\lambda_u)]$-components inside $C$. Strictly speaking, one would need the 
$[B_u]$-components inside $C$. However, by the special condition of HDs, it was 
shown in \cite{2002gottlob} that the $[B(\lambda_u)]$-components and $[B_u]$-components 
coincide.
For our algorithm to solve the \rec{{FHD},\,$k$} problem for hypergraphs with the BRP, the 
situation is even easier: by the constant bound on $|B(\gamma_u)|$ for every $u$, we can even afford to directly guess $B_u$, followed by a check (via linear programming) that 
$B_u$ indeed has a fractional edge cover of weight $\leq k$. If not, we reject. If so, we
call the decomposition procedure recursively for all
$[B_u]$-components. We thus get the following tractability result:

\newcommand{\thmRankFHD}{%
For every hypergraph class $\classC$ 
that has bounded rank, and for every 
constant $k \geq 1$, 
the 
\rec{FHD,\,$k$} problem is tractable, i.e., 
given a hypergraph $\HH \in \classC$, it is feasible in polynomial time to 
check  $\fhw(\HH)\leq k$ and, if so, to compute an FHD of 
width $k$ of $\HH$.%
}

\begin{theorem}\label{theo:exactFHD}
\thmRankFHD
\end{theorem}

\section{Efficient Approximation of FHW}
\label{sect:fhd}

We now turn our attention to approximations of the $\fhw$. 
It is known from~\cite{DBLP:journals/talg/Marx10}
that 
a tractable cubic approximation of the $\fhw$ always 
exists, i.e.: for $k \geq 1$,
there exists a polynomial-time algorithm that, given a hypergraph $H$ with 
$\fhw(H) \leq k$, finds an FHD of $H$ of  width $\calO(k^3)$. 
In this section, we search for conditions which guarantee a better 
approximation 
of the $\fhw$.

Natural first candidates for restricting hypergraphs are the BIP and, more 
generally, the BMIP. 
For the \rec{GHD,\,$k$} problem, these restrictions guarantee tractability. 
We have to leave it as an open question for future research 
if the BIP or even the BMIP also guarantees tractability 
of the \rec{FHD,\,$k$} problem for fixed $k \geq 1$.
However, 
in this section, we will show that a significantly better polynomial-time approximation 
of the $\fhw$ than in the general case is possible for hypergraphs enjoying the BIP or BMIP.

\subsection{Approximation of FHW in case of the BIP}
\label{sect:fhw-bip}

We first inspect the case of the bounded intersection property. We will show that the BIP allows for 
an arbitrarily close approximation of the $\fhw$ in polynomial time. 
Formally, the main result of this section is 
as follows:

\newcommand{\thmFHDwithBIP}{%
Let $\classC$  be a hypergraph class 
that enjoys the BIP and let $k, \epsilon$ be  arbitrary constants
with $k \geq 1$ and $\epsilon > 0$. 
Then there exists a polynomial-time algorithm that, given a hypergraph $H \in 
\classC$ with 
$\fhw(H) \leq k$, finds an FHD of $H$ of width $\leq k + \epsilon$.%
}

\begin{theorem}\label{theo:FHDwithBIP}
\thmFHDwithBIP
\end{theorem}

In the remainder of this section, we develop the necessary machinery to finally prove
Theorem~\ref{theo:FHDwithBIP}.
For this, we first introduce the crucial concept of a 
{\em $c$-bounded fractional part\/}.   
Intuitively, FHDs with $c$-bounded fractional part 
are FHDs, where the fractional edge cover $\gamma_u$ in every
node $u$ is ``close to an (integral) edge cover'' -- with the possible exception 
of up to $c$ vertices in $B(\gamma_u)$. 

It is convenient to first introduce the following notation: 
let $\gamma : E(H) \ra [0,1]$  and let $S \subseteq \cov(\gamma)$. We 
write $\gamma |_S$ to denote the {\em restriction of $\gamma$ to $S$\/}, i.e., 
$\gamma |_S (e) = \gamma(e)$ if $e \in S$ and  $\gamma|_S (e)= 0$  otherwise.

\begin{definition}
\label{def:c-bounded}
Let 
$\mcF= \left< T, (B_u)_{u\in T}, (\gamma_u)_{u\in T} \right>$ be 
an  FHD of some hypergraph $H$
and let $c \geq 0$. 
We say that $\mcF$ has {\em $c$-bounded fractional part\/} if
in every node $u \in T$, the following property holds: 
Let $R = \{e \in \cov(\gamma_u) \mid \gamma_u(e) < 1\}$.  
Then $|B(\gamma_u|_{R})| \leq c$.
\end{definition}
Clearly, for the special case $c= 0$, an FHD with 
$c$-bounded fractional part is essentially a GHD; in this case, we can simply define 
$\lambda_u(e) = 1$ if $\gamma_u(e) = 1$ and $\lambda_u(e) = 0$ otherwise, for every 
decomposition node $u$.
We next generalize the special condition (i.e., condition 4 of the definition of HDs) 
to FHDs. To this end, we define the {\em weak special condition\/}. 
Intuitively, it requires that the special condition
must be satisfied by the integral part of each fractional edge cover. 
For the special case $c = 0$, an FHD with 
$c$-bounded fractional part satisfying the weak special condition is thus
essentially a GHD satisfying the special condition, i.e., an HD.

\begin{definition}
\label{def:weak-SC}
Let 
$\mcF= \left< T, (B_u)_{u\in T}, (\gamma_u)_{u\in T} \right>$ be 
an  FHD of some hypergraph $H$.
We say that 
$\mcF$ satisfies the {\em weak special condition\/} if
in every node $u \in T$, the following property holds: 
for $S = \{e \in E(H) \mid \gamma_u(e) = 1\}$, we have 
$B(\gamma_u|_{S}) \cap V(T_u) \subseteq B_u$.
\end{definition}

The proof of Theorem~\ref{theo:FHDwithBIP} will be based on 
two key lemmas. Consider constants $k \geq 1$, $i \geq 0$, and $\epsilon > 0$, 
and
suppose that we are given a hypergraph $H$ with 
$\iwidth{\HH} \leq i$. 
We will show the following properties: 
if  $H$ has an FHD of width $\leq k$, then 
(1) $H$ also has an FHD of width $\leq k + \epsilon$ with 
$c$-bounded fractional part, where $c$ only depends on $k$, $\epsilon$, and  $i$, but not on the size of $H$, and
(2) we can extend $H$ to a hypergraph $H'$ by adding polynomially many edges, such that 
$H'$ has an FHD of width $\leq k+\epsilon$ with 
$c$-bounded fractional part satisfying the weak special condition. 
Theorem~\ref{theo:FHDwithBIP} will then be proved by appropriately adapting the 
\rec{HD,\,$k$} algorithm from \cite{2002gottlob}.

\medskip
\noindent {\bf $c$-bounded fractional part and weak special condition.} 
We start by proving two lemmas which, taken together, ensure that in case of the BIP, we
can transform an arbitrary FHD into an FHD with $c$-bounded fractional part satisfying 
the weak special condition at the expense of increasing the width by at most $\epsilon$.

\newcommand{\newlemCBounded}{%
Consider constants $k \geq 1$, $i \geq 0$, and $\epsilon > 0$, 
and
suppose that we are given a hypergraph $H$ with 
$\iwidth{\HH} \leq i$. 
If $H$ has an FHD of width $\leq k$, then it also has an FHD of width $\leq k + \epsilon$ with 
$c$-bounded fractional part, where $c$ only depends on $k$, $\epsilon$, and  $i$, but not on the size of $H$. 
More precisely, we have
$c = 2 i k^2 + \frac{4k^3i}{\epsilon}$.
}

\begin{lemma}\label{lem:fhwCBounded}
\newlemCBounded
\end{lemma}

\def\El{E_\ell}
\def\Eh{E_h}
\def\Ehs{E^s_h}
\def\Ehb{E^b_h}
\def\Vl{V_\ell}
\def\Vs{V_s}
\def\Vb{V_b}

\begin{proof}
Consider an arbitrary node $u$ in an FHD 
$\mcF= \left< T, (B_u)_{u\in T}, (\gamma_u)_{u\in T} \right>$ 
of $H$ and let $\gamma_u$ be an optimal fractional edge cover of $B_u$ with 
$\weight(\gamma_u) \leq k$. It suffices to show that there exists a fractional edge cover
$\gamma^*_u$ of $B(\gamma_u)$ (and, hence, of $B_u$) with $\weight(\gamma^*_u) \leq k+\epsilon$ and 
$|B(\gamma^*_u|_{R^*})| \leq c$ for $R^* = \{e \in \cov(\gamma^*_u) \mid \gamma^*_u(e) < 1\}$.  

We first partition the edges in $\cov(\gamma_u)$ into ``heavy'' ones (referred to as $\Eh$)
and ``light-weight'' ones (referred to as $\El$). Moreover, we further partition the heavy edges into 
``big'' and ``small'' ones (referred to as $\Ehb$ and $\Ehs$, respectively). 
For the border between
``heavy'' and ``light-weight'' edges, we could, in principle, 
choose any value $w \in (0,1)$. To keep the notation simple, we choose $w = 0.5$.
For the border between ``big'' and ``small'' edges, we have to 
choose a specific constant $d$, which will be introduced below. 
We thus define the following sets of edges: 

\smallskip

$E_\ell = \{ e \in  \cov(\gamma_u) \mid \gamma_u(e) < 0.5\}$,

$E_h = \{ e \in \cov(\gamma_u) \mid \gamma_u(e) \geq 0.5\}$,

$\Ehs = \{ e \in E_h \mid |e \cap B(\gamma_u)| < d\}$, and 

$\Ehb = \{ e \in E_h \mid |e \cap B(\gamma_u)| \geq d \}$ with $d = \frac{2 k^2 i}{\epsilon}$

\medskip
\noindent
This allows us to define the subsets $\Vl$, $\Vs$, and $\Vb$ of the vertices in $B(\gamma_u)$, 
s.t.\ $\Vl$ consists of the vertices only covered by light-weight edges, 
$\Vs$ consists of the vertices contained in at least one heavy edge but not in a big one, and
$\Vb$ consists of the vertices  contained in at least one big heavy edge,  i.e.: 

\smallskip

$\Vl = \{ x \in B(\gamma_u) \mid \forall e \in \Eh$: $x \not\in e\}$, 

$\Vs = \{ x \in B(\gamma_u) \mid \exists e \in \Ehs \mbox{ with } x \in e \mbox{ and } 
             \forall e \in \Ehb$: $x \not\in e\}$, 

$\Vb = \{ x \in B(\gamma_u) \mid \exists e \in \Ehb \mbox{ with } x \in e\}$.

\smallskip
\noindent
Our proof proceeds in three steps. We will first show that $\Vl$ and  $\Vs$ are bounded by constants.
$\Vb$ can in principle become arbitrarily large (cf.\ edge $\{v_1, \dots, v_n\}$ in Example \ref{ex:fhwLongEdge}). 
However, we will show that $\gamma_u$ can be transformed into $\gamma^*_u$ with 
$\weight(\gamma^*_u) \leq \weight(\gamma_u) + \epsilon$, s.t.\ all vertices of $\Vb$ are 
covered by edges of weight 1 in $\gamma^*_u$.

\medskip
{\sc  Claim A.} $|\Vl| < c_1$ with $c_1 = 2 i k^2$. 

\medskip
{\sc Proof of Claim A.} Let $e \in E_\ell$ and let $m = |e \cap B(\gamma_u)|$. 
We first show that $m < 2 i k$ holds, i.e., 
the size of light-weight edges is bounded by a constant. (Recall from Example \ref{ex:fhwLongEdge} that this is, 
in general, not the case for heavy edges.)
Indeed, 
by $e \in E_\ell$, we know that $e$ puts weight $< 0.5$ on each of its vertices. Hence, weight $> 0.5$ must be put on each vertex of 
$e \cap B(\gamma_u)$ by other edges $e' \in E_\ell$. In total, the other edges must put weight $> 0.5 m$ on the vertices of $e$.
Now we make use of the assumption that 
$\iwidth{H} \leq i$ holds, i.e., the intersection of any two edges of $H$ contains at most $i$ vertices. 
Hence, whenever $\gamma_u$ puts 
weight $w$ on some edge $e' \neq e$, then $e'$ puts at most weight $i  w$  in total on the vertices in $e$. 
By $weight(\gamma_u) \leq k$, the edges $e' \in E_\ell$ with $e' \neq e$ can put at most $i k$ total weight on the vertices in $e$. 
We thus get the inequality $i k  > 0.5 m$. In other words, we have $m < 2 i k$.

Now suppose that, for an arbitrary edge in $E_\ell$, we have $\gamma_u(e) = w$. Then $e$ can put at most 
weight $m w$ in total on the vertices in $\Vl$. By $weight(\gamma_u) \leq k$, all edges in $E_\ell$ together can 
put at most weight $m k$ on the vertices in $\Vl$. By our definition of $\Vl$, we have $\Vl \subseteq B(\gamma_u)$,
i.e.,
each vertex in $\Vl$ receives at least weight 1. Hence, there can be at most $m k$ vertices in $\Vl$. In total, we
thus have $|\Vl| < 2 i k^2$. 
\hfill$\diamond$

\medskip
{\sc  Claim B.} $|\Vs| < c_2$ with $c_2 = \frac{4k^3i}{\epsilon}$.

\medskip
{\sc Proof of Claim B.} First, we show that $|\Eh| \leq 2k$,
i.e., 
the number of heavy edges is bounded by a constant. (Recall from Example \ref{ex:fhwLongEdge} that this is, 
in general, not the case for the light-weight edges.)
By the definition of $E_h$, each edge in $E_h$ has weight $\geq 0.5$. By $weight(\gamma_u) \leq k$, the
total weight of the edges in $E_h$ is $\leq k$. Hence, there can be at most $2k$ edges in $E_h$. 

By $\Ehs \subseteq \Eh$ this implies that also $|\Ehs| \leq 2k$ holds.
Since 
$|e \cap B(\gamma_u)| < d$ with $d = \frac{2 k^2 i}{\epsilon}$ holds for all edges in $\Ehs$,
we thus have $|\Vs| < \frac{4k^3i}{\epsilon}$.
\hfill$\diamond$

\medskip
{\sc  Claim C.} For every $e \in \Ehb$, we have $\gamma_u(e) \geq 1 - \frac{\epsilon}{2k}$.

\medskip
{\sc Proof of Claim C.} Let $e \in \Ehb$ and let $ e \cap B(\gamma_u) = m$ with 
$m \geq \frac{2 k^2 i}{\epsilon}$.
Let $e' \in \cov(\gamma_u)$ with $e' \neq e$. 
Moreover, since $\iwidth{\HH} \leq i$,
we have $|e' \cap e| \leq i$. Hence, 
if $\gamma_u(e') = w$, then $e'$ can put at most total weight $wi$ on the vertices in $e$. 
By $\weight(\gamma_u) \leq k$, all edges $e' \in \cov(\gamma_u)$ with $e' \neq e$ taken together can  put 
at most total weight $ki$ on the vertices in $e$. The average weight thus put on each vertex in 
$e \cap B(\gamma_u)$ is
at most $\frac{ki}{m}$. Together with the condition $m \geq \frac{2 k^2 i}{\epsilon}$,
we thus get the upper bound $\frac{ki\epsilon}{2 k^2 i} = \frac{\epsilon}{2k}$ on the average weight put on 
each vertex in $e\cap B(\gamma_u)$ by all of the edges $e' \neq e$. Hence, there is at least one vertex 
$x \in e\cap B(\gamma_u)$ for which the total weight of the edges $e' \neq e$ with $x \in e'$ is 
at most $\frac{\epsilon}{2k}$. Therefore, for $x$ to be in $B(\gamma_u)$,  
$\gamma_u(e) \geq  1 - \frac{\epsilon}{2k}$ must hold.
\hfill$\diamond$

\medskip 
{\sc Conclusion of the proof of Lemma \ref{lem:fhwCBounded}.}
We are now ready to construct the desired fractional cover $\gamma^*_u$ of $B(\gamma_u)$ (and, hence, of $B_u$):
\[\gamma^*_u (e) = \left\{
  \begin{array}{ll}
    1 & \mbox{if } e \in \Ehb \\
    \gamma_u(e) & \mbox{otherwise}
  \end{array}
\right.
\]
The fractional cover $\gamma^*_u$ has the following properties: 
\begin{itemize}
\item In $\gamma^*_u$, the weight of an edge is never decreased compared with $\gamma_u$. Hence, 
$B(\gamma_u) \subseteq B(\gamma^*_u)$ and, therefore, $B_u \subseteq B(\gamma^*_u)$ holds. 
\item By Claim C, for every edge 
$e \in \Ehb$, we have $\gamma_u(e) \geq 1 - \frac{\epsilon}{2k}$ and, 
thus, $\gamma^*_u(e) = 1 \leq \gamma_u(e) + \frac{\epsilon}{2k}$.
Moreover, in the proof of Claim B, we have shown that 
$|\Eh| \leq 2k$ 
(and, thus, also $|\Ehb| \leq 2k$) holds. 
Hence, we have $\weight(\gamma^*_u) \leq \weight(\gamma_u) + \epsilon$.
\item By the definition of $\Vl,\Vs$, and $\gamma^*_u$, we have 
$B(\gamma^*_u|_{R^*}) \subseteq \Vl \cup \Vs$ with 
$R^* = \{e \in \cov(\gamma^*_u) \mid \gamma^*_u(e) < 1\}$, since 
all of the big, heavy edges under $\gamma_u$ have weight 1 in $\gamma^*_u$.  
Moreover, 
by Claims A and B, we have $|\Vl \cup \Vs| \leq c_1 + c_2 = 2 i k^2 + \frac{4k^3i}{\epsilon}$.
Hence, we also have 
$|B(\gamma^*_u|_{R^*})| \leq c_1 + c_2$.
\end{itemize}
By carrying out this transformation of $\gamma_u$ into $\gamma^*_u$ for every node $u$ of $\calF$, 
we thus get the desired FHD of $H$ of width $\leq k + \epsilon$ with 
$c$-bounded fractional part for 
$c = 2 i k^2 + \frac{4k^3i}{\epsilon}$
\end{proof}

The following lemma shows that, in case of the BIP, the weak special condition can be enforced without a
further increase of the width.

\begin{lemma}\label{lem:FHDweakSC}
Let $c \geq 0, i \geq 0$, and $k \geq 1$. 
There exists a polyno\-mial-time
computable function $f_{(c,i,k)}$ which takes as input a hypergraph $H$ 
with $\iwidth{H} \leq i$
and yields
as output a set of subedges of $E(H)$ with the following property:
if $H$ has an FHD of width $\leq k$ with 
$c$-bounded fractional part then 
$H'$ has an FHD of width $\leq k$ with 
$c$-bounded fractional part satisfying the weak special condition, where
$H' = (V(H), E(H) \cup f_{(c,i,k)}(H))$.
More specifically, $f_{(c,i,k)}(H)) = \{e' \mid e'$ is a subedge of some $e \in E(H)$ with $|e'| \leq ki + c\}$.
\end{lemma}

\begin{proof}
Let $H$ be a hypergraph with 
$\iwidth{\HH} \leq i$ and let 
$H'$ and $f_{(c,i,k)}(H))$ be as defined above.
Moreover, let $\mcF= \left< T, (B_u)_{u\in T}, (\gamma_u)_{u\in T} \right>$ 
be an FHD of $H$ (and hence also of $H'$) of width $\leq k$ with 
$c$-bounded fractional part. We assume $T$ to be rooted, where the root can be arbitrarily chosen among the 
nodes of $T$.
We have to show that $\mcF$ can be transformed into an 
FHD of $H'$ of width $\leq k$ with 
$c$-bounded fractional part satisfying the weak special condition.

We proceed similarly as in the proof of 
Theorem~\ref{theo:LogBMIP} -- with some simplifications due to the assumption of the BIP 
(rather than the less restrictive BMIP) and with some slight complications due to the fractional part.  
Suppose that $\mcF$ contains a violation of the weak special condition
(a weak-SCV, for short), i.e., 
there exists a node $u$ in $T$, an edge $e \in E(H)$ with $\gamma_u(e)  = 1$ and a vertex
$x \in e \cap V(T_u)$, s.t.\ $x \not\in B_u$ holds. We write $(u,e,x)$ to denote 
such a weak-SCV. W.l.o.g., we can choose a weak-SCV in such a way that there exists no weak-SCV for any node 
$u'$ below~$u$. 
We show that this weak-SCV can be eliminated by appropriately modifying 
the FHD $\mcF$ of $H'$.

By the connectedness condition, 
$e$ must be covered by some node $u^* \in T_u$, i.e., 
$u^*$ is a descendant of $u$ and $e \subseteq B_{u^*}$ holds. 
Let $\pi$ denote the path in $T$ from $u$ to $u^*$. 
We distinguish two cases: 

\smallskip
\noindent
{\em Case~1.} Suppose that for every node $u'$ along the path $\pi$ with $u' \neq u$, 
we have $x \in B_{u'}$. Then we simply add $x$ to $B_u$.
Clearly, this modification does not violate
any of the conditions of FHDs, i.e., the connectedness condition and 
the condition $B_{u} \subseteq B(\gamma_{u})$ are still fulfilled. 
Moreover, the weak-SCV $(u,e,x)$ has been eliminated and no new weak-SCV is introduced.
Finally, note that $\gamma_u$ is left unchanged by this transformation. Hence, the resulting FHD still 
has $c$-bounded fractional part.

\smallskip
\noindent
{\em Case~2.} Suppose that there exists a node $u'$ along the path 
$\pi$ with $u' \neq u$ and $x \not \in B_{u'}$. Of course, also 
$u' \neq u^*$ holds, since $x \in e$ and $e$ is covered by $u^*$.
We may also conclude that 
$\gamma_{u'}(e) < 1$. Indeed, suppose to the contrary that
$\gamma_{u'}(e) = 1$. Then $\mcF$ would contain the weak-SCV 
$(u',e,x)$ where $u'$ is below $u$, which contradicts our choice of 
$(u,e,x)$.

By the connectedness condition, $e \cap B_u \subseteq B_{u'}$ holds and, hence, 
also $e \cap B_u \subseteq e \cap B_{u'}$. Moreover, 
$B_{u'} \subseteq B(\gamma_{u'})$ holds by the definition of FHDs and
$B(\gamma_{u'})$ is of the form 
$B(\gamma_{u'}) = B_1 \cup B_2$ with 
$B_1 =  B(\gamma_{u'}|_{R})$ and $B_2 =  B(\gamma_{u'}|_{S})$
with 
$R = \{e \in \cov(\gamma_{u'}) \mid \gamma_{u'}(e) < 1\}$
and 
$S = \{e \in \cov(\gamma_{u'}) \mid \gamma_{u'}(e) = 1\}$.
Now let $S = \{e_1, \dots, e_\ell\}$ denote the set of edges with weight $1$ in 
$\gamma_{u'}$. 
Clearly, $\ell \leq k$, since the width of $\mcF$ is $\leq k$.
In total, we have: 
\[ (e \cap B_{u})  \subseteq e \cap B_{u'} = 
e \cap (e_1 \cup \dots \cup e_\ell \cup B_1) =
(e \cap e_1) \cup \dots \cup (e \cap e_\ell) \cup (e \cap B_1).
\]
The first $\ell$ intersections each have cardinality $\leq i$ and the 
last intersection has cardinality $\leq c$
by our assumption of $c$-fractional boundedness.
Together with $\ell \leq k$, we thus have
$|e \cap B_{u}| \leq k i + c$. 

Now let $e' = e \cap B_{u}$. We have just shown that $e'$ is a subset of $e$ with 
$|e'| \leq k  i + c$. Hence, $e'$ is an edge in $H'$. We can thus modify $\mcF$
by modifying $\gamma_u$ to $\gamma'_u$ as follows: 
we set $\gamma'_u(e) = 0$, $\gamma'_u(e') = 1$,  and let 
$\gamma'_u$ be identical to $\gamma_u$ everywhere else.
Clearly, we still have $B_u \subseteq B(\gamma'_u)$
and also $\weight(\gamma'_u) \leq k$ still holds.
Moreover, the weak-SCV 
$(u,e,x)$ has been eliminated and no new weak-SCV $(u,e',z)$ can arise since $e' = e \cap B_{u}$ 
implies $e' \subseteq B_u$. Finally, note that $B(\gamma_u|_R) = B(\gamma'_u|_R)$ and 
$R = \{e \in \cov(\gamma_u) \mid \gamma_u(e) < 1\}
= \{e \in \cov(\gamma'_u) \mid \gamma'_u(e) < 1\}$. Hence, 
the resulting FHD still 
has $c$-bounded fractional part.

\smallskip
\noindent
To conclude,  every modification of $\mcF$ by either Case 1 or Case 2 strictly decreases the number of 
weak-SCVs in our FHD. Moreover, the FHD resulting from such a modification still has $c$-bounded fractional part.
By exhaustively eliminating the weak-SCVs as described above, we thus end up with 
an FHD of $H'$ of width $\leq k$ with 
$c$-bounded fractional part satisfying the weak special condition.
\end{proof}

\medskip
\noindent
{\bf Normal Form of FHDs.} 
\label{sect:FNF}
In order to adapt the HD algorithm from \cite{2002gottlob} to turn it into an FHD algorithm that searches for FHDs with $c$-bounded fractional part for some constant $c$ and satisfying the weak special condition, we will carry over the normal form for HDs to FHDs having the two above properties. Recall the normal form introduced in Definition \ref{def:fhdnf}.

\medskip
\noindent
{\sc Definition \ref{def:fhdnf}}
\deffhdnf

We have already discussed in Section \ref{sect:fhd-exact} that the transformation of HDs into normal form presented in \cite{2002gottlob} can be taken
over almost literally to transform FHDs into FNF (see Theorem \ref{thm:nf} in the appendix). Below we argue that if this transformation is applied to an FHD with $c$-bounded fractional part and weak special condition then the resulting FHD in FNF also 
satisfies these two properties.

\begin{lemma}%
\label{lem:nf-fhw-approx}
Let $c \geq 0$.
For each FHD $\mcF$ of a hypergraph $H$ with $c$-bounded fractional part satisfying the weak special condition 
and with $\width(\mcF) \leq k$ there 
exists an FHD $\mcF^+$ of $H$ in FNF with $c$-bounded fractional part satisfying the weak special condition and with 
$\width(\mcF^+) \leq k$.
\end{lemma}

\begin{proof}[Proof Sketch]
The crucial part of the transformation into normal form is to ensure Conditions~1 and 2. Here, the proof of Theorem~5.4 from \cite{2002gottlob} can be taken over literally (as is done in the proof of Theorem \ref{thm:nf}  in the appendix) because it only makes use of the tree structure of the decomposition, the bags, and the connectedness condition. 
Ensuring also Condition~3 of our FNF is easy, because we may always 
extend $B_s$ by vertices from $B(\gamma_s)  \cap B_r$ without violating the connectedness condition.
Moreover, the transformation of HDs in \cite{2002gottlob} preserves the special condition. Analogously, when applying
this transformation to FHDs (as detailed in Theorem \ref{thm:nf}), the weak special condition is preserved. Finally, if $\mcF$ has $c$-bounded fractional part then so has $\mcF^+$. This is due to the fact that, by this transformation, no set
$B(\gamma_u|_{R})$ with $R = \{e \in \cov(\gamma_u) \mid \gamma_u(e) < 1\}$ is ever increased.
\end{proof}

Suppose that an FHD 
$\mcF = \defFHD$
is in FNF. Then, for every node $s \in T$, we define  
$\tc(s)$ as follows: 

\begin{itemize}
\item 
If $s$ is the root of $T$, then we set $\tc(s) = V(H)$.
\item 
Otherwise, let $r$ be the parent of $s$ in $T$. 
Then we set $\tc(s) = C_r$, where $C_r$ is the unique $[B_r]$-compo\-nent 
with 
$V(T_s) = C_r \cup (B_r \cap B_s)$ according to Condition~1 of the FNF.
\end{itemize}

We now carry Lemmas~5.5 -- 5.7 from \cite{2002gottlob} over to fractional 
hypertree decompositions
in fractional normal  form. The proofs from  \cite{2002gottlob} can be easily 
adapted to our setting. 
We therefore state the lemmas without proof.

\begin{lemma}[Lemma~5.5 from \cite{2002gottlob}]
\label{lem:55}
Let $\mcF = \defFHD$ be an arbitrary FHD 
of a hypergraph $H$ in fractional normal  
form, let $u \in T$, and let $W = \tc(u) \setminus B_u$. 
Then, for any \comp{$B_u$} C such
that $(C \cap W) \neq \emptyset$, we have $C \subseteq W$.
Therefore, $\mcC = \{ C' \subseteq V(H) \mid C' \text{ is a \comp{$B_u$} and }$ $ C'$ 
$\subseteq \tc(u) \}$ is a partition of $W$.
\end{lemma}

\begin{lemma}[Lemma~5.6 from \cite{2002gottlob}]
\label{lem:56}
Let $\mcF = \defFHD$  be  an arbitrary FHD 
of a hypergraph $H$ in fractional normal  
form and let 
$r \in T$. Then, $C = \tc(s)$ for some child $s$ of $r$ if and only if $C$
is a \comp{$B_r$} of $H$ and $C \subseteq \tc(r)$.
\end{lemma}

\begin{lemma}[Lemma~5.7 from \cite{2002gottlob}]
\label{lem:57}
For every FHD $\mcF = \defFHD$ of a hypergraph $H$ in fractional normal  
form, $|\nodes(T)| \leq |V(H)|$.
\end{lemma}

The next lemma is crucial for designing an algorithm that computes a concrete 
FHD.
The lemma is based on  Lemma~5.8 from \cite{2002gottlob}. However, the proof in 
the FHD-setting 
requires a slightly more substantial modification of the proof in the HD-setting. We 
therefore state
the lemma together with a full proof below.

\begin{lemma}[Lemma~5.8 from \cite{2002gottlob}]
\label{lem:58}
Let $c \geq 0$ and let $\mcF = \defFHD$ be an FHD in FNF of a hypergraph 
$H$ such that $\mcF$ 
has $c$-bounded fractional part and satisfies the {\em weak special condition\/}.
Further, let $s$ be 
a node in $T$ and let $r$ be the parent of $s$ in $T$. 
Let $R_1 = \{e \in E(H) \mid \gamma_s(e) = 1\}$ and $R_2 = \{e \in E(H) \mid \gamma_s(e) < 1\}$, 
and let $B_s = B_1 \cup B_2$ with 
$B_1 =  B_s \cap B(\gamma_s|_{R_1})$ and $B_2 =  B_s \cap B(\gamma_s|_{R_2})$. 
Finally, let $C$ be a set of 
vertices
such that $C \subseteq \tc(s)$. 
Then the following equivalence holds:

\begin{center}
$C$ is a $[B_s]$-component if and only if 
$C$ is a 
$[B(\gamma_s|_{R_1}) \cup B_2]$-component.
\end{center}
\end{lemma}

{\sc Remark.} The crux of the proof of Lemma~5.8 from \cite{2002gottlob} and likewise of Lemma~\ref{lem:58}
stated here is the following: 
by the definition of FHDs, we have 
$B_s \subseteq B(\gamma_s|_{R_1}) \cup B_2$. Hence, every  
$[B(\gamma_s|_{R_1}) \cup B_2]$-path is also a $[B_s]$-path, but the converse is, at first glance, not clear.
However, by the weak special condition, $(B(\gamma_s|_{R_1}) \cup B_2) \setminus B_s$ only contains elements from $B_r \cap B_s$. Moreover, we are assuming that $C$ is a subset of $\tc(s)$, i.e., it is in the complement of $B_r$. 
Hence, $[B_s]$-paths and $[B(\gamma_s|_{R_1}) \cup B_2]$-paths actually coincide. 
From this it is then straightforward to conclude that, inside $\tc(s)$,
$[B_s]$-components and 
$[B(\gamma_s|_{R_1}) \cup B_2]$-components coincide.

\begin{proof}
Let $W = B(\gamma_s|_{R_1}) \cup B_2$. We first prove the following Property (1), which is the analogue of Property (1) in the proof of Lemma~5.8 from \cite{2002gottlob}:  
 \begin{gather} 
 W  \cap  \tc(s) \subseteq B_s. 
 \end{gather}
 
\smallskip
{\sc Proof of Property (1).} By the definition of FHDs, we have $B_s \subseteq B(\gamma_s|_{R_1}) \cup B_2 = W$.
By the weak special condition, we have $B(\gamma_s|_{R_1}) \cap V(T_s) \subseteq B_s$. 
By the definition of $\tc(s)$, we have $V(T_s) = \tc(s)  \cup (B_s \cap B_r)$, i.e., 
also $\tc(s) \subseteq V(T_s)$ clearly holds. 
In total, we thus have: 

\medskip
$W  \cap \tc(s) = 
      (B(\gamma_s|_{R_1}) \cup B_2)  \cap \tc(s) \subseteq$
      
$(B(\gamma_s|_{R_1}) \cap \tc(s)) \cup B_2 \subseteq 
      (B(\gamma_s|_{R_1}) \cap V(T_s)) \cup B_2  \subseteq  B_s. \hfill \diamond$
\medskip
\noindent
It remains to show for 
$C \subseteq \tc(s)$, that $C$ is a 
$[B_s]$-component if and only if 
$C$ is a 
$[W]$-component.
This proof follows the line of argumentation in the proof of Lemma~5.8 from \cite{2002gottlob} 
-- replacing Property (1) there with our Property~(1) proved here. For the sake of completeness, we present a detailed proof of the desired equivalence below.

\medskip
{\sc Proof of the ``only if''-direction.}
Suppose that $C$ is a $[B_s]$-component with $C \subseteq \tc(s)$. 
Then, in particular, $C \cap B_s = \emptyset$.
Hence, by Property (1), we have 
$C \cap W = \emptyset$. This can be seen as follows:
$C \cap W \subseteq \tc(s) \cap W \subseteq B_s$ (the last inclusion uses Property~(1)). 
Hence, also $C \cap W \subseteq C \cap B_s$ holds. 
Together with $C \cap B_s = \emptyset$, we thus have $C \cap W = \emptyset$.

We have to show that $C$ is a $[W]$-component, i.e., 
$C$ is  $[W]$-connected and $C$ is maximal $[W]$-connected.
For the $[W]$-connectedness, consider an arbitrary pair of vertices 
$\{x,y\} \subseteq C$, i.e., 
there exists a [$B_s$]-path $\pi$ between $x$ and $y$. Note that this [$B_s$]-path $\pi$ only 
goes through vertices in $C$. Hence, by $C \cap W = \emptyset$, $\pi$ is also a [$W$]-path. 
Hence, $C$ is indeed $[W]$-connected.
For the maximality, we simply make us of the relationship $B_s \subseteq W$. This means that 
since $C$ is maximal $[B_s]$-connected, it is also maximal $[W]$-connected.

\medskip
{\sc Proof of the ``if''-direction.}
Suppose that $C$ is a  $[W]$-component with $C \subseteq \tc(s)$. 
By $B_s \subseteq W$, we conclude that the 
$[W]$-connectedness of $C$ implies the $[B_s]$-connectedness. It remains to show that $C$ is 
maximal $[B_s]$-connected. Clearly, there
exists a $[B_s]$-component $C'$ with $C \subseteq C'$. By Lem\-ma~\ref{lem:55}, we have 
$C' \subseteq \tc(s) \setminus B_s$. In particular, $C' \subseteq \tc(s)$. Hence, by the ``only if'' part of
this lemma, $C'$ is a $[W]$-component and, therefore, $C$ cannot be a proper subset of $C'$.
Hence, $C = C'$. Thus, $C$ is indeed a $[B_s]$-component. 

\medskip
\noindent
We have thus shown for $C \subseteq \tc(s)$, that $C$ is a 
$[B_s]$-component if and only if 
$C$ is a 
$[W]$-component. This concludes the proof of the lemma.
\end{proof}

\medskip
\noindent
{\bf A PTime algorithm for FHDs with $c$-bounded fractional part.}
We now adapt the HD algorithm from \cite{2002gottlob} to 
turn it into an FHD algorithm that searches for 
FHDs with $c$-bounded fractional part for some constant $c$ 
and satisfying the weak special condition.
By Lemmas \ref{lem:fhwCBounded} and \ref{lem:FHDweakSC}, we know that 
for any constants $k \geq 1$, $i \geq 0$, and $\epsilon > 0$, there exists
a constant $c$ (which only depends on $k,i, \epsilon$) with the following property: 
for every hypergraph $H$ with $\iwidth{H} \leq i$, 
if $H$ has an FHD of width $\leq k$, then 
$H$ also has an FHD of width $\leq k + \epsilon$ 
with $c$-bounded fractional part
and satisfying the weak special condition.
Moreover, by Lemma \ref{lem:nf-fhw-approx},
the transformation into FNF can be done in such a way that it
does not increase the width and preserves
the 
$c$-boundedness of the fractional part and the weak special condition. 
Hence, in our Algorithm~\ref{alg:decompalg} ``\decomp'', we restrict our search to 
FHDs of width $\leq k + \epsilon$  
in FNF with $c$-bounded fractional part
and satisfying the weak special condition. Moreover, 
throughout the remainder of Section~\ref{sect:fhw-bip},
we assume that 
for every edge $e$ in a given hypergraph $H$, all subedges $e'$ of $e$ of size 
$|e'| \leq ki+c$ (cf.\ Lemma \ref{lem:FHDweakSC}) 
have already been added to $H$.

\nop{****************************
Consider an FHD 
$\mcF= \left< T, (B_u)_{u\in T}, (\gamma_u)_{u\in T} \right>$ 
with these properties 
of some hypergraph $H$.
For an arbitrary node $s \in T$, 
let $S = \{e \in E(H) \mid \gamma_s(e) = 1\}$
and let 
$B_s = B_1 \cup B_2$ with 
$B_1 =  B_s \cap B(\gamma_s|_{S})$
and $B_2 = B_s \setminus B_1$. 
Moreover, by the weak special condition, 
$B(\gamma_u|_{S}) \cap V(T_s) \subseteq B_s$ holds.
****************************} %

\begin{algorithm}[t]
\SetKwData{Left}{left}\SetKwData{This}{this}\SetKwData{Up}{up}
\SetKwFunction{Union}{Union}\SetKwFunction{FindCompress}{FindCompress}
\SetKw{Halt}{Halt}\SetKw{Reject}{Reject}\SetKw{Accept}{Accept}

\SetKwData{N}{N}

\SetKwInOut{Input}{input}\SetKwInOut{Output}{output}

\Input{Hypergraph $H$.}
\Output{``Accept'', if $H$ has an FHD of width $\leq k+\epsilon$ \linebreak
        \hspace*{36pt} with $c$-bounded fractional part and weak special condition \linebreak 
        ``Reject'', otherwise.}
\BlankLine
  \SetKwFunction{FDecomp}{\fdecomp}
  \SetKwProg{Fn}{Function}{}{}
  \Fn{\FDecomp{$C_r$, $W_r$: Vertex-Set, $R$: Edge-Set}}{
        \Begin(\tcc*[f]{\bfseries (1) Guess}){
        Guess a set $S \subseteq E(H)$ with $|S| = \ell$, s.t.\ $\ell \leq k + \epsilon$ \tcc*[r]{(1.a)}
        Guess a set $W_s \subseteq (V(R) \cup W_r \cup C_r)$ with $|W_s| \leq c$ \tcc*[r]{(1.b)}
        }
        \Begin(\tcc*[f]{\bfseries (2) Check}){
        Check if $\exists \gamma$ with $W_s \subseteq B(\gamma)$ and 
       $\weight(\gamma) \leq k + \epsilon - \ell $ \tcc*[r]{(2.a)}
        Check if  $\forall e \in \edges(C_r)\colon e \cap ( V(R) \cup W_r ) 
          \subseteq  (V(S) \cup W_s)$ \tcc*[r]{(2.b)}
        Check if  $(V(S) \cup W_s) \cap C_r \neq \emptyset$ \tcc*[r]{(2.c)}
        }
        \lIf(\tcc*[f]{\bfseries (3)}){one of these checks fails}{\Halt and \Reject}
        \Else{
            Let $\mcC := \{ C \subseteq V(H) \mid C$ is a $[V(S) \cup W_s]$-component 
     and $C \subseteq C_r\}$\;
        }
        \ForEach(\tcc*[f]{\bfseries (4)}){$C \in \mcC$}{
           \If{\fdecomp ($C, W_s, S$) returns \Reject}{\Halt and \Reject}
        }
        \KwRet \Accept\;
  }
  \BlankLine
  
  \Begin(\tcc*[f]{\bfseries Main}){
        \KwRet \fdecomp ($V(H), \emptyset, \emptyset$)
  }
\caption{\decomp}\label{alg:decompalg}
\end{algorithm}

Let $\tau$ be a computation tree of the alternating algorithm \decomp. We can associate with every $\tau$
an FHD  $\delta(\tau) = \defFHD$, called witness tree, defined as follows:
For each existential configuration in $\tau$ corresponding to the ``guess'' of 
some 
sets $S \subseteq E(H)$ and $W_s \subseteq V(H)$ in Step 1 
during the execution of  a procedure call 
\fdecomp~$(C_r,W_r,R)$, 
the tree $T$ contains a node~$s$.
In particular, at the initial call 
\fdecomp~$(V(H),\emptyset,\emptyset)$, 
the guesses $S$ nad $W_s$ give rise to the root node $s_0$ of $T$.
Moreover, there is an edge between nodes $r$ and $s$ of $T$, if $s \neq s_0$ 
and
$r$ is the node in $T$ corresponding to the guess of sets 
$R \subseteq E(H)$ and $W_r \subseteq V(H)$.
We will
denote $C_r$ by $\component(s)$, and $r$ by $\parent(s)$. Moreover, for the 
root $s_0$ of $T$, we define $\component(s_0) = V(H)$.

Each node $s \in T$ is labelled by $B_s$ and $\gamma_s$ as follows. 
First, we define $\gamma_s$ via the
mapping $\gamma$, which exists according to the check in Step~2.a: 
\[ \gamma_s(e) = \begin{cases}
                 \gamma(e) & \text{for } e \in \cov(\gamma) \setminus S\\
                 1 & \text{for } e \in S \\
                 0 & \text{otherwise}
              \end{cases}
\]
Note that for the sets $S$ of edges and $W_s$ of vertices guessed in Step 1, 
we have $B(\gamma_s) = (V(S) \cup W_s)$.
As far as the definition of the bags is concerned, we 
set $B_{s_0} = B(\gamma_{s_0})$
for the root node $s_0$. For any other node $s$ with 
$r = \parent(s)$ and $C = \component(s)$,
we set $B_s = B(\gamma_s) \cap (B_r \cup C)$.

The correctness proof of the  \decomp algorithm
is along the same lines 
as the correctness proof of the alternating algorithm for the 
\rec{HD,\,$k$} problem 
in 
\cite{2002gottlob}. We therefore 
state the analogues of the lemmas and theorems of \cite{2002gottlob}
without proofs, since these can be easily  
``translated'' from the HD setting in \cite{2002gottlob} to our 
FHD setting. 

\begin{lemma}[based on Lemma~5.9 from \cite{2002gottlob}]
 \label{lem:59}
Let $H $ be a hypergraph and let $k \geq 1$, $\epsilon > 0$, and $c \geq 0$. 
If 
$H$ has an FHD of width $\leq k + \epsilon$ in FNF
with $c$-bounded fractional part and satisfying the weak special condition, 
then the algorithm \decomp accepts input $H$. Moreover, every such FHD
is equal to some witness tree $\delta(\tau)$ of  \decomp 
when run on input $H$.
\end{lemma}

The next three lemmas will help to show the converse: whenever \decomp has 
an accepting computation, then the corresponding witness tree is 
an FHD of $H$ of width $\leq k + \epsilon$ in FNF
with $c$-bounded fractional part and satisfying the weak special condition.

\begin{lemma}[based on Lemma~5.10 from \cite{2002gottlob}]
\label{lem:510}
Let $H $ be a hypergraph and let $k \geq 1$, $\epsilon > 0$, and $c \geq 0$. 
Assume that \decomp accepts input $H$ with an accepting 
computation tree $\tau$ and corresponding witness tree $\delta(\tau) = \defFHD$.
Then, for each  node $s$ in $T$:
 \begin{enumerate}
  \item[(a)] if $s \neq \treeroot(T)$, then $\component(s)$ is a 
\comp{$B_r$} with $r = \parent(s)$;
  \item[(b)] for any $C \subseteq \component(s)$, $C$ is a \comp{$B_s$} if and
             only if $C$ is a \comp{$V(S) \cup W_s$}.
 \end{enumerate}
\end{lemma}

\begin{lemma}[based on Lemma~5.11 from \cite{2002gottlob}] 
\label{lem:511}
Let $H $ be a hypergraph and let $k \geq 1$, $\epsilon > 0$, and $c \geq 0$. 
Assume that \decomp accepts input $H$ 
with an accepting 
 computation tree $\tau$. Let $\delta(\tau) = \defFHD$ be the corresponding
 witness tree and $s \in T$. Then, for each node $u \in T_s$:
 \begin{align*}
  B_u &\subseteq \component(s) \cup B_s \\
  \component(u) &\subseteq \component(s).
 \end{align*}
\end{lemma}

\begin{lemma}[based on Lemma~5.12 from \cite{2002gottlob}] 
\label{lem:512}
Let $H $ be a hypergraph and let $k \geq 1$, $\epsilon > 0$, and $c \geq 0$. 
Assume that \decomp accepts  input $H$  with an accepting 
 computation tree $\tau$ and corresponding  witness tree $\delta(\tau) = \defFHD$. 
 Let $s \in T$ and $C_r = \component(s)$. 
 Then, for every edge $e \in \edges(C_r)$ and every edge $e' \in E(H) \setminus \edges(C_r)$, 
 we have $e \cap e' \subseteq B_s$. 
\end{lemma}

We are now ready to show that, whenever \decomp has an accepting computation on an 
input hypergraph $H$, 
then the corresponding witness tree is an FHD of $H$ of width $\leq k + \epsilon$ in FNF with $c$-bounded fractional part and satisfying the weak special condition. As before, the following lemma can be shown similarly 
to the corresponding result in \cite{2002gottlob}.

\begin{lemma}[based on Lemma~5.13 from \cite{2002gottlob}]
 \label{lem:513}
Let $H $ be a hypergraph and let $k \geq 1$, $\epsilon > 0$, and $c \geq 0$. 
If  \decomp accepts input $H$, then
$H$ has an FHD of width $\leq k+\epsilon$ in FNF with $c$-bounded fractional part
and satisfying the weak special condition. 
Moreover, in case of acceptance, every witness tree $\delta(\tau)$ for $H$ is an FHD $\mcF$ 
of $H$ in FNF
with $c$-bounded fractional part satisfying the weak special condition.
\end{lemma}

\nop{***********************

\begin{proof}[Proof idea.]
 Consider an accepting computation tree $\tau$ of \decomp on an input
 hypergraph $H$. We have to show that $\delta(\tau) = \defFHD$ is an 
 FHD$^+$ of $H$ in fractional normal form FHD$^+$ with the 
$(\alpha,\beta)$-condition, i.e.\
 we have to show that $\delta(\tau)$ fulfils the properties of an FHD$^+$. 
 This can be shown similarly as in the proof of Lemma~5.13 in 
 \cite{2002gottlob}.
\end{proof}

****************************}

The following result follows immediately from the above Lemmas~\ref{lem:59} and 
\ref{lem:513}.

\begin{theorem}
\label{theorem:correct}
Let $H$ be a hypergraph and let $c \geq 0$ and $\epsilon > 0$.
Then, 
\decomp accepts input $(H,c,\epsilon)$ if
and only if $H$ has an FHD of width $\leq k+ \epsilon$ with $c$-bounded fractional part and 
satisfying the weak special condition.
Moreover, in case of acceptance, every witness tree $\delta(\tau)$ for $H$ is an FHD $\mcF$ 
of $H$ in FNF
with $c$-bounded fractional part satisfying the weak special condition.
\end{theorem}

It remains to establish the {\sc Ptime} membership of our algorithm. Again, we can easily carry over the
corresponding tractability result from Lemma 5.15 in \cite{2002gottlob}. The crux of the proof in \cite{2002gottlob}
is that all data structures involved in the alternating algorithm fit into logspace. 
In total, our alternating algorithm 
\decomp
has to maintain the following 6 data structures: 
the input parameters $C_r$, $W_r$, and $R$ of procedure \fdecomp and 
the local variables  $S$, $W_s$, and the component $C$ of the next recursive procedure call.
In the alternating algorithm in  \cite{2002gottlob}, only 4 data structures are needed, which correspond to 
$C_r$, $R$, $S$, and $C$ in our setting. The data structures $W_r$ and $W_s$ are only used in our algorithm. However, these
are just sets of constantly many vertices. Hence, they can of course also be stored in logspace. 
The rest of the proof arguments can then be easily carried over from \cite{2002gottlob}.
When it comes to the complexity of the checks in step 2, we additionally have to solve a linear program in our algorithm, which  can also be done in  {\sc Ptime} (or on an ATM using logspace). We thus get:

\begin{lemma}[based on Lemma~5.15 from \cite{2002gottlob}]
 \label{lem:algorithmcomplexity} 
The  alternating algorithm 
\decomp can be implemented on a logspace ATM.
\end{lemma}

\noindent
Theorem \ref{theo:FHDwithBIP}, now follows immediately by putting together
Lemmas~\ref{lem:fhwCBounded},~\ref{lem:FHDweakSC}, and~\ref{lem:nf-fhw-approx}, Theorem~\ref{theorem:correct}, and 
Lemma~\ref{lem:algorithmcomplexity}.

\medskip
\noindent
{\bf A polynomial time approximation scheme for finding optimal FHDs.}
Recall the definition of the \boundedopt\ problem from Section~\ref{sect:introduction}, i.e.: given a hypergraph $H$, we are interested in $\fhw(H)$,
but only if $\fhw(H) \leq K$. 
We will now show that, with the alternating algorithm \decomp at our disposal, we are able to give a polynomial time absolute approximation scheme (PTAAS) for 
the bounded optimization problem. More precisely, we aim at an approximation algorithm with the following properties: 

\begin{definition}[PTAS \cite{DBLP:books/lib/Ausiello99,DBLP:books/daglib/0004338}]
 \label{def:PTAS}
 Let $\Pi$ be an (intractable) minimization problem with positive objective function $f_\Pi$. 
 An algorithm \texttt{Alg} is called an {\it approximation scheme} for $\Pi$ if on input $(I,\epsilon)$, where $I$ is an instance of $\Pi$ and $\epsilon > 0$ is an error parameter, it outputs a solution $s$ such that:
   \[ f_\Pi(I,s) \leq (1+\epsilon) \cdot f_\Pi(I,s^*) \]
 where $s^*$ is an optimal solution of $I$, i.e. for all other solutions $s'$ of $I$ it holds that $f_\Pi(I,s^*) \leq f_\Pi(I,s')$.
 
 \texttt{Alg} is called a \emph{polynomial time approximation scheme} (PTAS), if for every fixed $\epsilon > 0$, its running time is bounded by a polynomial in the size of instance $I$.
\end{definition}

In our case we can even achieve something slightly better. As seen in Definition \ref{def:PTAS}, the error of the solution given by an approximation scheme is {\em relative\/} to the optimal value of the optimization problem. Clearly, it would be preferable, if the gap between the solution returned  by the approximation scheme and an optimal solution
has an {\em absolute\/} bound (i.e., not depending on the optimal value). This leads to the following definition of {\em absolute}  approximation schemes:

\begin{definition}[PTAAS]
 \label{def:PTAS}
 Let $\Pi$ be an (intractable) minimization problem with positive objective function $f_\Pi$. 
We say that algorithm \texttt{Alg} is an {\it absolute approximation scheme} for $\Pi$ if on input $(I,\epsilon)$, where $I$ is an instance of $\Pi$ and $\epsilon > 0$ is an error parameter, it outputs a solution $s$ such that:
   \[ f_\Pi(I,s) \leq f_\Pi(I,s^*) + \epsilon \]
 where $s^*$ is the optimal solution of $I$, i.e. for all other solutions $s'$ of $I$ it holds that $f_\Pi(I,s^*) \leq f_\Pi(I,s')$.
 
 \texttt{Alg} is called a \emph{polynomial time absolute approximation scheme} (PTAAS), if for each fixed $\epsilon > 0$, its running time is bounded by a polynomial in the size of instance $I$.
\end{definition}
Note that every PTAAS is also a PTAS, since for any $\epsilon > 0$ it holds that $f_\Pi(I,s^*) + \epsilon \leq (1+\epsilon) \cdot f_\Pi(I,s^*)$, provided that $f_\Pi(I,s^*) \geq 1$. Actually, we can assume w.l.o.g. that this is indeed the case \cite{DBLP:books/lib/Ausiello99}. 
We now show that, in case of the BIP, the 
$K$-\textsc{Bounded-FHW-Optimization} problem indeed allows for a PTAAS (and, hence, for a PTAS):

\begin{algorithm}[t]
\SetKwData{Left}{left}\SetKwData{This}{this}\SetKwData{Up}{up}
\SetKwFunction{Union}{Union}\SetKwFunction{FindCompress}{FindCompress}
\SetKw{not}{not}
\SetKwData{N}{N}

\SetKwInOut{Input}{input}\SetKwInOut{Output}{output}

\Input{hypergraph $H$ with $\iwidth{H} \leq i$, numbers $K \geq 1$, $\epsilon \geq 0$}
\Output{approximation of $\fhw(H)$, i.e.,  FHD $\mcF$ with $\width(\mcF) \leq \fhw(H) + \epsilon$ if $\fhw(H) \leq K$}
\BlankLine
\tcc{Check upper bound}
\If{\not ($\mcF =$ \findFHD ($H$, $K$, $\epsilon$, $i$))}{
   \Return{\tt fails} \tcc*[r]{$\fhw(H) > k$}
}
\BlankLine
\tcc{Initialization}
$L = 1$\;
$U = K + \epsilon$\;
$\epsilon'$ = $\epsilon / 3$\;
\BlankLine
\tcc{Main computation}
\Repeat{$U-L<\epsilon$}{
   \If{$\mcF' =$ \findFHD ($H$, $L + (U-L)/2$, $\epsilon'$, $i$)}{
      $U = L + (U-L)/2 +\epsilon'$\;
      $\mcF = \mcF'$\;
   }
   \Else{
      $L = L + (U-L)/2$\;
   }
}
\Return{$\mcF$}\;
\caption{{\tt FHW-Approximation}} \label{alg:ptaas}
\end{algorithm}%

\begin{theorem}\label{theo:PTAS}
For every hypergraph class $\classC$ that enjoys the BIP, there exists a PTAAS for the \boundedopt\ problem.
\end{theorem}
\begin{proof}
By Theorem \ref{theo:FHDwithBIP}, there exists a function \findFHD ($H$, $k$, $\epsilon$, $i$) with the following properties: 

\begin{itemize}
\item \findFHD\  takes as input a hypergraph $H$ with $\iwidth{H} \leq i$ and numbers $k \geq 1$, $\epsilon \geq 0$; 
\item \findFHD returns an FHD $\mcF$ of $H$ of width $\leq k+\epsilon$ if such exists and \texttt{fails} otherwise (i.e., $\fhw(H) > k$ holds).
\item \findFHD\  runs in time polynomial in the size of $H$, where $k$, $\epsilon$, and $i$ are considered 
as~fixed.
\end{itemize}
Then we can construct Algorithm \ref{alg:ptaas} ``{\tt FHW-Approximation}'', which uses \findFHD\ as subprocedure. We claim that 
{\tt FHW-Approximation} is indeed a PTAAS for the \boundedopt\ problem. First we argue that the algorithm is correct; we will 
then also show its polynomial-time upper bound. 

As for the correctness, note that the algorithm first checks
if $K$ is indeed an upper bound on $\fhw(H)$. This is done via a call of function \findFHD. 
If the function call {\tt fails}, then we know that 
$\fhw(H) > K$ holds. Otherwise, we get an FHD $\mcF$ of width $\leq K + \epsilon$. In the latter case, 
we conclude that $\fhw(H)$ is in the interval $[L,U]$ with 
$L = 1$ and $U = K + \epsilon$.

The loop invariant for the repeat loop is, that $\fhw(H)$ is in the interval $[L,U]$ and 
the width of the FHD $\mcF$ is $\leq U$. 
To see that this invariant is preserved by every iteration of the loop, consider the 
function call \findFHD ($H$, $L + (U-L)/2$, $\epsilon'$, $i$): 
if this call succeeds then the function returns an FHD $\mcF'$ of width $\leq L + (U-L)/2 +\epsilon'$. Hence,
we may indeed set $U = L + (U-L)/2 +\epsilon'$ and $\mcF = \mcF'$ without violating the loop invariant.
On the other hand, suppose that the call of \findFHD\ fails. This means that 
$\fhw(H) > L + (U-L)/2$ holds. Hence,
we may indeed update $L$ to $L + (U-L)/2$ and the loop invariant still holds.
The repeat loop terminates when $U - L < \epsilon$ holds. 
Hence, together with the loop invariant, we conclude that, on termination, 
$\mcF$ is an FHD of $H$ with $\width(\mcF) - \fhw(H) < \epsilon$. 

It remains to show that algorithm {\tt FHW-Approximation} runs in polynomial time w.r.t.\ the size of $H$. 
By Theorem \ref{theo:FHDwithBIP}, the function \findFHD\ works in polynomial time w.r.t.\ $H$. 
We only have to show that the number of iterations of the repeat-loop is bounded by a polynomial in $H$. 
Actually, we even show that it is bounded by a constant (depending on $K$ and $\epsilon$, but not on $H$): 
let $K' := K + \epsilon - 1$. Then the size of the interval $[L,U]$ initially is $K'$. 
In the first iteration of 
the repeat-loop, we either set 
$U = L + (U-L)/2 +\epsilon'$ or $L = L + (U-L)/2$ holds. In either case, 
at the end of this iteration, we have $U - L \leq K' / 2 + \epsilon'$. 
By an easy induction argument, it can be verified that after $m$ iterations (with $m \geq 1$), 
we have 
$$U - L \leq \frac{K'}{2^m} + \epsilon' + \frac{\epsilon'}{2} + \frac{\epsilon'}{2^2} + \dots + \frac{\epsilon'}{2^{m-1}}
$$ 
Now set $m = \lceil \log (K' / \epsilon') \rceil$. Then we get
$$
\frac{K'}{2^m}  \leq \frac{K'}{2^{\log (K' / \epsilon')}} =
\frac{K'}{(K' / \epsilon')} = \epsilon'.$$
Moreover, $\epsilon' + \frac{\epsilon'}{2} + \frac{\epsilon'}{2^2} + \dots + \frac{\epsilon'}{2^{m-1}} < 2 \epsilon'$
for every $m \geq 1$.
In total, we thus have that, after  $m = \lceil \log (K' / \epsilon')\rceil$ iterations of the repeat loop, 
$U - L < 3 \epsilon' = \epsilon$ holds, i.e., the loop terminates. 
\end{proof}

\subsection{Approximation of FHW in case of the BMIP}
\label{sect:fhw-bmip}

We now present a polynomial-time approximation of the $\fhw$ for classes of hypergraphs enjoying the BMIP. 
Actually, the approximation even works for slightly more general classes of hypergraphs, 
such that hypergraph classes with BMIP constitute a familiar subcase. 
For this, we will combine some classical results 
on the Vapnik-Chervonenkis (VC) dimension with some novel observations.
This will yield an approximation of the $\fhw$ 
up to a logarithmic factor
for hypergraphs enjoying the BMIP. We first recall the definition of the 
VC-dimension of hypergraphs.

\begin{definition}[\cite{1972sauer,1971vc}]
\label{def:vc}
Let $\HH=(V(H),E(H))$ be a hypergraph and $X\subseteq V(H)$ a set of vertices. 
Denote by $E(H)|_X$ the set 
$E(H)|_X =\{X \cap e\, |\, e\in E(H)\}$. 
The vertex set $X$ is called {\em shattered} if 
$E(H)|_X=2^X$.
The {\em Vapnik-Chervonenkis dimension (VC dimension) $\vc(\HH)$} of $\HH$ is 
the maximum cardinality of a shattered subset of $V(H)$. 
\end{definition}

We now provide a link between the VC-dimension and 
our approximation of the $\fhw$.

\begin{definition}
\label{def:transversality}
Let $H = (V(H),E(H))$ be a hypergraph. A {\em transversal}  (also known as {\em 
hitting set\/})
of $H$ is a subset $S \subseteq V(H)$ that has a non-empty intersection with 
every edge of $H$. 
The {\em transversality} $\tau(H)$  of $H$  is the 
minimum cardinality of all transversals of $H$.

Clearly, $\tau(H)$ corresponds to the minimum of the following integer linear 
program: 
find a mapping $w: V\rightarrow \{0,1\}$ 
which minimizes $\Sigma_{v\in V(H)}w(v)$ under the condition that
$\Sigma_{v\in e}w(v)\geq 1$ holds for each hyperedge $e\in E$.

The {\em fractional transversality}  $\tau^*$ of $H$ is defined as the minimum 
of
the above linear program when dropping the integrality condition,
thus allowing mappings $w: V\rightarrow \mathbb{R}_{\geq 0}$.
Finally, the {\em transversal integrality gap} $\tigap{\HH}$ of $\HH$ is the 
ratio $\tau(\HH)/\tau^*(\HH)$.
\end{definition}

Recall that computing the mapping $\lambda_u$ for 
some 
node $u$ in a GHD can be seen as searching for a minimal edge cover $\rho$ of 
the vertex set $B_u$, whereas computing 
$\gamma_u$ in an FHD 
corresponds to the search for a minimal fractional edge cover $\rho^*$
\cite{2014grohemarx}. Again, 
these problems can be cast as linear programs where the first problem has the 
integrality condition and the second one has not. 
Further, we can define the {\em cover integrality gap} $\cigap{H}$ of $H$ as 
the 
ratio $\rho(\HH)/\rho^*(\HH)$.
With this, we state the following approximation result for 
$\fhw$.

\newcommand{\thmApproxVC}{%
Let $\classC$ be a class of hypergraphs with VC-dimension bounded by some 
constant $d$
and let $k \geq 1$.
Then there exists a polynomial-time algorithm that, given a hypergraph $H \in 
\classC$ with 
$\fhw(H) \leq k$, finds an FHD of $H$ of width $\calO(k \cdot \log k)$.%
}

\begin{theorem}\label{theo:ApproxVC}
\thmApproxVC
\end{theorem}

\newcommand{\ggnew}[1]{#1}
\newcommand\hcancel[2][black]{\setbox0=\hbox{$#2$}%
\rlap{\raisebox{.45\ht0}{\textcolor{#1}{\rule{\wd0}{2pt}}}}#2} 
\newcommand{\ggkill}[1]{\hcancel[violet]{#1}}

\begin{proof}
The proof proceeds in several steps.

\smallskip

\noindent
{\em Reduced hypergraphs.} 
As in Section \ref{sect:fhd-exact}, Proposition \ref{prop:furedi}, 
we consider, w.l.o.g., only  
hypergraphs $H$ that satisfy the following 4 conditions: (1) $H$ has no  isolated vertices 
and (2) no empty edges. Moreover, (3) no two distinct vertices in $H$  
have the same edge-type (i.e., the two vertices occur in precisely the same edges) and 
(4) no two distinct edges in $H$ have the same vertex-type (i.e., we exclude duplicate edges). 
Hypergraphs satisfying these conditions will be called ``reduced''.

\smallskip

\noindent
{\em Dual hypergraphs.} Given a hypergraph $H = \{V,E)$, the dual hypergraph 
$H^d  = (W,F)$ 
is defined as $W = E$ and $F = \{ \{e \in E \mid v \in e\} \mid v \in V\}$. 
We are  assuming that $H$ is {\em reduced\/}. This ensures that 
$(H^d)^d = H$  holds.
Moreover, it is well-known and easy to verify that the following relationships
between $H$ and $H^d$ hold for any reduced hypergraph $H$,  
(see, e.g., \cite{duchet1996hypergraphs}):

\smallskip

(1) The edge coverings of $\HH$ and the transversals of $\HH^d$
coincide.

(2) The fractional %
edge coverings of $\HH$ and the fractional transversals of 
$\HH^d$
coincide.

(3) $\rho(\HH)=\tau(\HH^d)$, $\rho^*(\HH)=\tau^*(\HH^d)$, and
$\cigap{\HH}=\tigap{\HH^d}$.

\smallskip

\noindent
{\em VC-dimension.} By a classical result (\cite{ding1994} Theorem (5.4), see also \cite{Bronnimann1995} for related results), for every 
hypergraph $H = (V(H),E(H))$  \ggnew{with at least two edges} we have:
$$\tigap{H}= \tau(H)/\tau^*(H) \leq 2\vc(H)\log(11\tau^*(H)).$$
\ggnew{For hypergraphs $H$ with a single edge only, $\vc(H)=0$, and thus the above inequation does not hold. However, for such hypergraphs 
$\tau(H)=\tau^*(H)=1$. By putting this together, we get:}  
\ggnew{$$\tigap{H}= \tau(H)/\tau^*(H) \leq \max(1,2\vc(H)\log(11\tau^*(H))).$$}
Moreover, in \cite{assouad1983}, it is shown that $\vc(\HH^d)<2^{\vc(\HH)+1}$ 
always
holds.
In total, we thus get 

\begin{align*}
\cigap{\HH}=\tigap{\HH^d} & \leq \ggnew{\max(1,}2\vc(\HH^d)\log(11\tau^*(\HH^d))\ggnew{)} \\
& \leq \ggnew{\max(1,}2^{\vc(\HH)+2}\log(11\rho^*(\HH))\ggnew{)}\\
& \ggnew{\leq \max(1,2^{d+2}\log(11\rho^*(\HH)))},\ 
\ggnew{\mbox{which is\ } O(\log\rho^*(H))}.
\end{align*}

\noindent
{\em Approximation of $\fhw$ by $\ghw$.} 
Suppose that $H$ has an FHD $\left< T, (B_u)_{u\in V(T)}, (\lambda)_{u\in V(T)} 
\right>$
of width $k$. Then there exists a GHD of $H$ of width %
$\calO(k \cdot \log k)$. Indeed, we can find such a GHD by leaving the 
tree 
structure $T$ and the bags $B_u$ for every node $u$ in $T$ unchanged and 
replacing each fractional edge cover $\gamma_u$ of $B_u$ by an optimal integral 
edge cover $\lambda_u$ of $B_u$. By the above inequality, we thus increase the 
weight at each node $u$ only by a factor $\calO(\log k)$. Moreover, we know 
from 
\cite{DBLP:journals/ejc/AdlerGG07} that $\hw(H) \leq 3 \cdot \ghw(H) + 1$ holds.
In other words, we can compute an HD of $H$ (which is a special case of an FHD) 
in polynomial time,  whose width is $\calO(k \cdot\log k)$.
\end{proof}

\ggnew{The following} 
Lemma~\ref{lemma:BMIPvsVC} establishes a relationship between BMIP and VC-dimension.
Together with Theorem \ref{theo:ApproxVC}, Corollary~\ref{cor:ApproxBMIP} is 
then immediate.

\newcommand{\lemBMIPvsVC}{%
If a class $\classC$ of hypergraphs has the BMIP then it has bounded 
VC-dimension.
However, there exist classes $\classC$ of hypergraphs with bounded VC-dimension 
that do not have 
the BMIP.}

\begin{lemma}\label{lemma:BMIPvsVC}
\lemBMIPvsVC
\end{lemma}

\begin{proof} \mbox{}
[BMIP $\Rightarrow$ bounded VC-dimension.]
Let $c \geq 1, i \geq 0$ and let 
$H$ be a hypergraph with $\cmiwidth{c}{\HH}\leq i$.
We claim that then $\vc(H) \leq c +  i$ holds.

Assume to the contrary that there exists a set  $X \subseteq V$, such 
that $X$ is  shattered and $|X| > c + i$. 
We pick $c$ arbitrary, pairwise distinct vertices $v_1, \dots, v_{c}$ from $X$ 
and 
define $X_j = X \setminus \{v_j\}$ for each $j$. 
Then $X = (X_1 \cap \dots \cap X_c) \cup \{v_1, \dots, v_c\}$ holds
and also $|X| \leq   |X^*| + c$ with 
$X^* \subseteq X_1 \cap \dots \cap X_c$.

Since $X$ is shattered, for each $1 \leq j \leq c$, there exists a distinct
edge $e_j \in E(H)$ with $X_j = X \cap e_j$. 
Hence, $X_j = X \setminus \{v_j\} \subseteq e_j$ and also 
$X^* \subseteq e_1 \cap e_2 \cap  \dots \cap e_c$ holds, i.e., 
$X^*$ is in the intersection of $c$ edges of $H$. 
By $\cmiwidth{c}{\HH}\leq i$, we thus get $|X^*| \leq i$.
In total, we have $|X| \leq |X^*| + c \leq i + c$, which 
contradicts our assumption that $|X| > c+i$ holds.

\smallskip
\noindent
[bounded VC-dimension $\not\Rightarrow$ BMIP.]
It suffices to exhibit a family $(H_n)_{n \in \mathbb{N}}$ 
of hypergraphs such that 
$\vc(H_n)$ is bounded whereas 
$\cmiwidth{c}{\mbox{$H_n$}}$ is unbounded for any constant $c$.
We define $H_n = (V_n,E_n)$ as follows:

\smallskip
$V_n = \{v_1, \dots, v_n\}$ 

$E_n = \{ V_n \setminus \{v_i\} \mid 1 \leq i \leq n\}$

\smallskip
\noindent
Clearly, $\vc(H_n) < 2$. Indeed, take an arbitrary set $X \subseteq V$ with 
$|X| \geq 2$. Then $\emptyset \subseteq X$ but
$\emptyset \neq X \cap e$ for any $e \in E_n$. 
On the other hand, let $c \geq 1$ be an arbitrary constant and let $X = e_{i_1} 
\cap \dots \cap e_{i_\ell}$ for some $\ell \leq c$ and edges $e_{i_j} \in E_n$. 
Obviously, $|X| \geq n - c$ holds. Hence, also $\cmiwidth{c}{\mbox{$H_n$}} \geq 
n - c$, i.e., 
it is not bounded by any constant $i \geq 0$.
\end{proof}

\noindent
In the first part of Lemma \ref{lemma:BMIPvsVC}, 
we have shown that $\vc(H) \leq c +  i$ holds.
For an approximation of an FHD by a GHD, we need
to approximate the fractional edge cover $\gamma_u$ of each bag $B_u$ 
by an integral edge cover $\lambda_u$, i.e.,
we  consider fractional vs.\ integral edge covers of the induced hypergraphs
$H_u = (B_u, E_u)$ with $E_u = \{e  \cap B_u \mid e \in E(H)\}$.
Obviously, the bound $\vc(H) \leq c +  i$ carries over to 
$\vc(H_u) \leq c +  i$.

\begin{corollary}\label{cor:ApproxBMIP}
Let $\classC$ be a class of hypergraphs enjoying the BMIP and let $k \geq 1$.
Then there exists a polynomial-time algorithm that, given 
$H \in 
\classC$ with 
$\fhw(H) \leq k$, finds an FHD 
(actually, a GHD and even an HD) of $H$ 
of width $\calO(k \cdot \log k)$.%
\end{corollary}

\section{Conclusion and Future Work}
\label{sect:conclusion}
In this work,  we have settled the complexity of deciding $\fhw(H) \leq k$ for fixed 
constant $k\geq 2$ and $\ghw(H) \leq k$ for 
$k = 2$ by proving the $\NP$-completeness of both problems. 
This gives negative answers to two open problems.
On the positive side, we have identified rather 
mild restrictions such as the BDP (i.e., the bounded degree property), 
BIP (i.e., the bounded intersection property), LogBIP, 
BMIP (i.e., the bounded multi-intersection property), and LogBMIP, 
which give rise  to  a \ptime algorithm 
for the 
\rec{GHD,\,$k$} problem. Moreover, we have shown that the BDP 
ensures tractability also of the 
\rec{FHD,\,$k$} problem. For the BIP, we have shown that an arbitrarily close approximation 
of the $\fhw$ in polynomial time exists. In case of the BMIP, we have proposed a polynomial-time
algorithm for approximating the $\fhw$ up to a logarithmic factor.
As our empirical analyses reported 
in 
\cite{pods/FischlGLP19} show, these restrictions are very
well-suited for %
instances of CSPs and, even more so, of CQs. 
We believe that they deserve 
further attention. 

Our work does not finish here. We plan to explore several further issues 
regarding the computation and approximation of the fractional hypertree width. 
We find the following questions particularly appealing: (i)  Does the special 
condition defined by  Grohe and Marx~\cite{2014grohemarx} lead to tractable 
recognizability also for FHDs,  i.e., in case we 
define ``{\it sc-fhw$(H)$}'' as the smallest width an FHD of $H$ satisfying the 
special condition,  can $\mbox{\it sc-fhw}(H) \leq k$ be recognized efficiently?\ 
(ii)  Our tractability result in Section 5 for the \rec{FHD,\,$k$} problem is weaker than for \rec{GHD,\,$k$}. 
In particular, for the BIP and BMIP, we have only obtained efficient approximations of the $\fhw$. It is open if the 
BIP or even the BMIP suffices to ensure tractability of \rec{FHD,\,$k$}. And if not, we should at
least search for a better approximation of the $\fhw$ in case of the BMIP.
Or can non-approximability results be obtained under reasonable 
complexity-theoretic assumptions? 

\begin{acks}
 This work was supported by the Engineering and Physical
Sciences Research Council (EPSRC), Programme Grant EP{\slash}M025268{\slash} VADA:
Value Added Data Systems --- Principles and Architecture as well as by the Austrian
Science Fund (FWF):P30930 and Y698.
\end{acks}

\clearpage

\appendix

\section{Normal Form of FHDs}
\label{app:FHD-normal-form}

In Sections \ref{sect:fhd-exact} and \ref{sect:fhw-bip}, we made use of the fractional normal form (FNF), which generalizes the normal form of 
HDs from \cite{2002gottlob}. We will show here, that we can transform any FHD into fractional normal form. This transformation follows closely the transformation of HDs into normal form given in \cite{2002gottlob}. 
We first recall the definition of the FNF introduced in~Section~\ref{sect:fhd-exact}.

\medskip
\textit{Definition \ref{def:fhdnf}.} \deffhdnf

\noindent
We now carry over several properties of the normal form from  \cite{2002gottlob} 
to our  FNF defined above. An inspection of the corresponding proofs in \cite{2002gottlob} reveals that these properties hold with minor modifications also in the fractional case. We thus state the following results below without explicitly ``translating'' the proofs of \cite{2002gottlob} to the fractional setting.

Note that \cite{2002gottlob} deals with HDs and, therefore, in all decompositions 
considered there, the special condition holds. However, 
for all properties of the normal form carried over from HDs to FHDs 
in Lemmas~\ref{lem:52} and \ref{lem:53} below, 
the special condition is not needed.  

\begin{lemma}[Lemma~5.2 from \cite{2002gottlob}]
\label{lem:52}
Consider an arbitrary FHD $\mcF = \defFHD$ of a hypergraph $H$. 
Let $r$ be a node in $T$, 
let $s$ be a child of $r$ and let $C$ be a \comp{$B_r$} of $H$ such that
$C \cap \VTs \neq \emptyset$. Then, $\nodes(C,\mcF) 
\subseteq \nodes(T_s)$.
\end{lemma}
 
\begin{lemma}[Lemma~5.3 from \cite{2002gottlob}]
\label{lem:53}
Consider an arbitrary FHD $\mcF = \defFHD$ of a hypergraph $H$. 
Let $r$ be a node in $T$ and let 
$U \subseteq V(H) \setminus B_r$ such that
$U$ is [$B_r$]-connected. Then
$\nodes(U,\mcF)$ induces a (connected) subtree of $T$.
\end{lemma}

We will now show that any FHD $\mcF$ of width $k$ can be transformed into an FHD of width $k$ in fractional normal form.

\begin{theorem}[Theorem~5.4 from \cite{2002gottlob}]
\label{thm:nf}
For every FHD $\mcF$ of a hypergraph $H$ with $\width(\mcF) = k$ there 
exists an FHD $\mcF^+$ of $H$ in FNF with $\width(\mcF^+) = k$.
\end{theorem}

\nop{*********************
\noindent
{\em Remark.}
The crucial part of the transformation into normal form is to ensure Conditions~1 and 2. Here, the proof of Theorem~5.4 from \cite{2002gottlob} can be taken over literally because it only makes use of the tree structure of the decomposition, the bags and the connectedness condition. Ensuring also Condition~3 of our FNF is easy, because we may always 
extend $B_s$ by nodes from $B(\gamma_s)  \cap B_r$ without violating the connectedness condition. To make these steps more clear, we will give a detailed proof of Theorem \ref{thm:nf} here. This proof follows closely the proof of Theorem~5.4 from \cite{2002gottlob}.
*********************}

\begin{proof}
 Let $\mcF = \defFHD$ be an arbitrary FHD of $H$ of width $k$. We show how to transform $\mcF$ into an FHD $\mcF^+$ of width $k$ in fractional normal form. The proof follows closely the proof of Theorem~5.4 from \cite{2002gottlob}.
 
 Assume that there exist two nodes $r$ and $s$ such that $s$ is a child of $r$, and $s$ violates any condition of Definition~\ref{def:fhdnf}. If $s$ satisfies Condition (1), but violates Condition (2), then $B_s \subseteq B_r$ holds. In this case, simply eliminate node $s$ from the tree $T$ and attach all children of $s$ to $r$. It is immediate to see that this transformation preserves all conditions of Definition~\ref{def:FHD}. 
 
Now assume that $T_s$ does not meet Condition (1) of Definition~\ref{def:fhdnf}, and let $C_1,\ldots,C_h$ be all the $[r]$-components containing some vertex occurring in $V(T_s)$. Hence, $V(T_s) \subseteq \left(\bigcup_{i=1}^h C_i \cup B_r\right)$. For each $[r]$-component $C_i$ ($1 \leq i \leq h$), consider the set of nodes $\nodes(C_i,\mcF)$. By Lemma~\ref{lem:53}, $\nodes(C_i,\mcF)$ induces a subtree of $T$, and by Lemma~\ref{lem:52}, $\nodes(C_i,\mcF) \subseteq \nodes(T_s)$. 
Hence $\nodes(C_i,\mcF)$ induces in fact a subtree of $T_s$.
 
 \newcommand{\new}{\ensuremath{u}}
 For each node $n \in \nodes(C_i,\mcF)$ define a new node $\new_n^{C_i}$, and let $\gamma_{\new_n^{C_i}}=\gamma_n$ and $B_{\new_n^{C_i}} = B_n \cap (C_i \cup B_r)$. Note that $B_{\new_n^{C_i}} \neq \emptyset$, because by definition of $\nodes(C_i,\mcF)$, $B_n$ contains some vertex belonging to $C_i$. Let $N_i$ = $\{ u_n^{C_i} \mid n \in \nodes(C_i,\mcF)\}$ and, for any $C_i$ ($1 \leq i \leq h$), let $T_i$ denote the (directed) graph $(N_i,E_i)$ such that $u_p^{C_i}$ is a child of $u_q^{C_i}$ if and only if $p$ is a child of $q$ in $T$. $T_i$ is clearly isomorphic to the subtree of $T_s$ induced by $\nodes(C_i,\mcF)$, hence $T_i$ is a tree as well.
 
 Now transform the FHD $\mcF$ as follows: Delete every node in $\nodes(T_s)$ from $T$, and attach to $r$ every tree $T_i$ for $1 \leq i \leq h$. Intuitively, we replace the subtree $T_s$ by a set of trees $\{ T_1, \ldots, T_h \}$. By construction, $T_i$ contains a node $u_n^{C_i}$ for each node $n$ belonging to $\nodes(C_i,\mcF)\ (1 \leq i \leq h)$. Then, if we let $\ffo{children}(r)$ denote the set of children of $r$ in the new tree $T$ obtained after the transformation above, it holds that for any $s' \in \ffo{children}(r)$, there exists an $[r]$-component $C$ of $H$ such that $\nodes(T_{s'})=\nodes(C,\mcF)$, and $V(T_{s'}) \subseteq (C\cup B_r)$.
 
 Furthermore, it is easy to verify that all the conditions of Definition~\ref{def:FHD} are preserved during this transformation. As a consequence, Condition~(2) of Definition~\ref{def:FHD} immediately entails that $(V(T_{s'})\cap B_r) \subseteq B_{s'}$. Hence, $V(T_{s'})=C \cup (B_{s'} \cap B_r)$. Thus, any child of $r$ satisfies both Condition~(1) and Condition~(2) of Definition~\ref{def:fhdnf}. 
 
Now assume that some node $u \in \ffo{children}(r)$ violates Condition~(3) of Definition~\ref{def:fhdnf}. 
Then we add to $B_u$ the set of vertices $B(\gamma_u)\cap B_r$. Because vertices in $B_r$ induce connected subtrees of $T$, and $B_r$ does not contain any vertices occurring in some $[r]$-component, this further transformation never invalidates any other condition. Moreover, for this transformation, we only use already covered vertices of $\gamma_u$ and therefore do not change the width of the FHD. 
 
 Note that the root of $T$ cannot violate any of the normal form conditions, because it has no parent in $T$. Moreover, the transformations above never change the parent $r$ of a violating node $s$. Thus, if we apply such a transformation to the children of the root of $T$, and iterate the process on the new children of the root of $T$, and so on, we eventually get a new FHD $\mcF^+ = \left<T^+,(B_u)_{u \in T^+},(\gamma_u)_{u\in T^+}\right>$ of $H$ in fractional normal form.
\end{proof}

\end{document}